\documentclass[onecolumn,showpacs,preprint,amsmath,amssymb,superscriptaddress]{revtex4-1}

\usepackage{dcolumn}
\usepackage{bm}
\usepackage{algpseudocode}
\usepackage{algorithm2e}

\usepackage{tikz}

\usepackage{hyperref}

\newcommand{\ignore}[1]{}
\usepackage{graphicx}
\usepackage{amsmath,amssymb,amsfonts,amsthm,color,latexsym}
\usepackage{soul}
\newcommand\strikeout{\bgroup\markoverwith{\textcolor{red}{\rule[0.4ex]{1pt}{0.3ex}}}\ULon}
\usepackage[caption=false]{subfig}
\def\mb{\mathbf}

\def\bm{\boldsymbol}
\newcommand{\Rmnum}[1]{\uppercase\expandafter{\romannumeral #1\relax}}
\def\mb{\mathbf}
\usepackage{amsmath}
\newcommand*\diff{\mathop{}\!\mathrm{d}}

\mathchardef\mhyphen="2D
\usepackage{amsmath}
\DeclareMathOperator{\Tr}{Tr}
\newtheorem{theorem}{Theorem}[section]
\newtheorem{lemma}[theorem]{Lemma}
\newtheorem{proposition}[theorem]{Proposition}

\newcommand{\revision}[1]{\color{black}#1 \color{black}}

\usepackage[normalem]{ulem}

\begin{document}
\preprint{}

\title{DeePN$^2$: A deep learning-based non-Newtonian hydrodynamic model}
\author{Lidong Fang}
\thanks{The first two authors contributed equally}
\affiliation{Department of Computational Mathematics, Science and Engineering, Michigan State University, MI 48824, USA}%
\author{Pei Ge}
\thanks{The first two authors contributed equally}
\affiliation{Department of Computational Mathematics, Science and Engineering, Michigan State University, MI 48824, USA}%
\author{Lei Zhang}
\affiliation{School of Mathematical Sciences, Institute of Natural Sciences and MOE-LSC, Shanghai Jiao Tong University, 800 Dongchuan Road, Shanghai 200240, China}%
\author{Weinan E}
\affiliation{Center for Machine Learning Research and School of Mathematical Sciences, Peking University, Beijing 100871, China}%
\affiliation{AI for Science Institute, Beijing 100080, China}%
\affiliation{Department of Mathematics and Program in Applied and Computational Mathematics, Princeton University,
NJ 08544, USA}%
\author{Huan Lei}
\email{leihuan@msu.edu}
\affiliation{Department of Computational Mathematics, Science and Engineering, Michigan State University, MI 48824, USA}%
\affiliation{Department of Statistics and Probability, Michigan State University, MI 48824, USA.}%


\begin{abstract}
A long standing problem in the  modeling of non-Newtonian hydrodynamics of polymeric flows is 
the availability of reliable and interpretable hydrodynamic
models that faithfully encode the underlying micro-scale polymer dynamics. 
The main complication arises from the long polymer relaxation time, the complex molecular structure and heterogeneous interaction. 
DeePN$^2$, a deep learning-based non-Newtonian hydrodynamic model,  has been proposed and  has shown some 
success in systematically passing the micro-scale structural mechanics information to the 
macro-scale hydrodynamics for suspensions with simple polymer conformation and bond potential.
The model retains a multi-scaled nature by mapping the polymer configurations into
a set of symmetry-preserving macro-scale features. 
The extended constitutive laws for these macro-scale features can be directly learned from the kinetics of
 their micro-scale counterparts. 
In this paper, we  develop DeePN$^2$ using more complex micro-structural models.
We show that DeePN$^2$ can faithfully capture the broadly overlooked viscoelastic differences arising from the specific molecular
structural mechanics without human intervention.
\end{abstract}

\pacs{}
\maketitle

\section{Introduction}
Accurate modeling of  non-Newtonian hydrodynamics plays a central role in the modeling of the transport, diffusion, and 
synthesis processes in many scientific and engineering applications. Unlike simple fluids,  non-Newtonian fluids may exhibit enormously 
complex flow behavior as a result of  the micro-scale polymer dynamics. 
In particular, the polymer relaxation time often becomes comparable to the hydrodynamic time scale.
As a result, the macro-scale fluid evolution can not be uniquely determined by the instantaneous flow field and the memory effect is generally important. 
To close the hydrodynamic equations, existing models 
are primarily based on the following two approaches. 
The first approach relies on empirical  constitutive models  \cite{Larson88,Owens_Phillips_2002}.
Notable examples  include the Hookean model \cite{Oldroyd_Wilson_PRSLA_1950, Energy_Variation_Lin_Liu_Zhang_CPAM_2005}, the FENE-P model \cite{Peterlin_Polymer_Science_1966, Bird_Doston_JNNFM_1980}, the Giesekus model \cite{Giesekus_JNNFM_1982}, and the 
  Phan-Thien and Tanner models \cite{PTT_JNNFM_1977}. 
Despite their popularity, the accuracy of these models is almost always in doubt.
The second approach resorts to  various sophisticated micro-macro coupling algorithms, e.g.,
 by directly solving the Fokker-Planck equation using
 Lattice Boltzmann method \cite{ammar2010lattice}, Galerkin method \cite{Fan_Acta_1989,lozinski2003fast,chauviere2004simulation,shen2012approximation}, and particle method \cite{carrillo2019blob, degond1990deterministic,lacombe1999presentation, wang2021particle, bao2021deterministic},   or 
sampling the polymer configuration via micro-scale simulations \cite{Laso_Ottinger_JNNFM_1993, Hulsen_Heel_JNNFM_1997, REN_E_HMM_complex_fluid_2005}. 
While the effects of the polymer interaction 
can be carried over to the macro-scale model, the computational cost can be exceedingly large due to
 the retaining of the micro-scale description. 
Methods based on  asymptotic analysis \cite{Warner_IECF_1972, Warner_PhD_1971} or the direct fitting of the strain-stress relationship \cite{Zhao_Li_JCP_2018} 
are limited to simple flows such as the steady flow. Several semi-analytical approaches have been proposed \cite{Grosso_Maffettone_JNNFM_2000, Feng_Leal_JoR_1998, WANG_JNNFM_1997_1, Forest_Zhou_Wang_JR_2003, FENE_L_S_JNNFM_1999, Yu_Du_mms_2005, Hyon_Du_mms_2008}
using moment closure to approximate the  micro-scale polymer configuration probability density function (PDF)  and
to derive the constitutive equations for the FENE
dumbbell solution \cite{FENE_L_S_JNNFM_1999, Yu_Du_mms_2005, Hyon_Du_mms_2008}.
However, these approaches are all based on restricted ansatz for the PDF and therefore are not reliable for more general flow regimes.

To construct  truly reliable and interpretable hydrodynamic models with  molecular-level fidelity, 
it is essential to be able to efficiently code the information from the micro-scale interaction into the macro-scale 
transport equations. Ideally, the construction should meet the following requirements:
\begin{itemize}
\item be interpretable;
\item be reliable  -- it should be accurate for all kinds of practical situations that one might encounter;
\item respect physical constraints, including symmetries and conservation laws;
\item be numerically robust and efficient. 
\end{itemize}

As a first step towards constructing models that meet these requirements,
we developed a machine learning-based approach \cite{Lei_Wu_E_2020}, ``deep learning-based non-Newtonian hydrodynamic model'' or DeePN$^2$, that learns the non-Newtonian
hydrodynamic model from the underlying micro-scale description of the dumbbell solution. 
Rather than approximating the closure with  standard moments, DeePN$^2$ 
finds a set of encoders, i.e., a set of macro-scale features that best represent the micro-scale dumbbell structure.
It also finds accurate closed-form equation for these macro-scale features. 
The constructed model retains a clear 
physical interpretation
and accurately captures the nonlinear viscoelastic responses, where the conventional Hookean and FENE-P models show limitations.

Beyond  dumbbell suspensions, one major challenge towards constructing truly reliable hydrodynamic models 
arises from the heterogeneous polymer micro-structural mechanics. 
%
In this work, we aim to fill the gap by developing the generalized DeePN$^2$ model for  multi-bead polymer molecules with arbitrary structure and interaction. 
Firstly, with the proper design of the generalized micro-macro encoders and the machine learning-based
symmetry-preserving constitutive dynamics, we demonstrate that the heterogeneous molecular structural-induced interaction can be systematically encoded into the macro-scale hydrodynamics. 
Unlike  moment closure approximations, the encoders are not designed to recover the 
high-dimensional configuration PDF.  Instead, they take an \emph{interpretable} form and are learned to probe the optimal approximation of the polymer stress and constitutive dynamics. 
This essential difference enables  DeePN$^2$ to circumvent the high-dimensionality of the polymer configuration PDF.
%
Secondly, the explicit form of the micro-macro encoders enables us to reliably learn the dynamics of the macro-scale features directly from the
kinetic equations of their micro-scale analog.
In this sense, this learning framework retains a multi-scaled nature 
where micro-scale interaction and physical constraints can be seamlessly inherited. 
Moreover,  the learning only requires instantaneous micro-scale samples. This unique property differs from the common sophisticated data-driven approaches \cite{Rudy_Kutz_Science_Ad_2017, schaeffer2018extracting, raissi2019physics, qin2019data, Han_Ma_PNAS_2019, Naoki_Takeshi_PRR_2020,Yu_E_Onsagernet_2020,Huang_Yong_JNET_2021}, where  time-derivative samples are often needed to learn the governing dynamics. 
This is particularly suited for multi-scale fluid models where  
accurate time-derivative samples may not be readily accessible. 
We demonstrate the power of the DeePN$^2$ model for polymer molecules of three distinct shapes with
training samples collected from one-dimensional (1D) homogeneous shear flow. Numerical results show that the broadly overlooked heterogeneous molecular structural mechanics plays an important role in the rheology of non-Newtonian fluids, which, fortunately, can be faithfully encoded into DeePN$^2$.
The constructed model successfully captures the hydrodynamics
with different viscoelastic responses for a variety of 1D and 2D flows when compared with the micro-scale simulation results.
The present work also paves the way towards constructing truly reliable non-Newtonian hydrodynamic models for general 3D flows.
 
 
\section{Methods}
\subsection{Micro-scale and continuum hydrodynamic models}
Let us start with the micro-scale description of the semi-dilute polymer suspension. We assume each molecule consists of $N$ particles 
with the position vector $\mb q = [\mb q_1; \mb q_2; \cdots; \mb q_N]$, where $\mb q_i \in \mathbb{R}^3$ is the position of the $i\mhyphen$th
particle. The intramolecular potential energy $V(\mb q)$ takes the form
\begin{equation}
V(\mb q) = \sum_{j=1}^{N_b} V_b\left(\vert \mb q_{j_1} - \mb q_{j_2}\vert\right), \quad 
V_b (l) =-\frac{k_s}{2}l^2_{0} \log\left[ 1-\frac{l^2}{l^2_{0}}\right],
\label{eq:micro-scale_model}
\end{equation}
where $N_b$ is the bond number and $(j_1, j_2)$ represents the indices of \revision{beads associated with} the $j\mhyphen$th bond.
Without  loss of generality, the individual bond interaction $V_b$ takes the form of the FENE potential \cite{Warner_FENE_1972}, where $k_s$ is the spring constant 
and $l_0$ is the maximum of the extension length. It is worth mentioning that the polymer molecule is not restricted to the
dumbbell shape. Instead, it generally consists of multiple particles with arbitrary structure and bond connection. 
Fig. \ref{fig:molecule_shape} shows a sketch of the polymer molecules with three different structures. As we will show, given the same
form of the individual bond interaction $V_b$, the different polymer micro-structural mechanics  leads to 
distinct non-Newtonian hydrodynamics.  

\begin{figure}[htbp]
\centering

\begin{tikzpicture}[x=0.75pt,y=0.75pt,yscale=-0.3,xscale=0.3]

\draw  [line width=0.75]  (200,226) .. controls (200,212.19) and (211.19,201) .. (225,201) .. controls (238.81,201) and (250,212.19) .. (250,226) .. controls (250,239.81) and (238.81,251) .. (225,251) .. controls (211.19,251) and (200,239.81) .. (200,226) -- cycle ;
\draw  [line width=0.75]  (301,276) .. controls (301,262.19) and (312.19,251) .. (326,251) .. controls (339.81,251) and (351,262.19) .. (351,276) .. controls (351,289.81) and (339.81,301) .. (326,301) .. controls (312.19,301) and (301,289.81) .. (301,276) -- cycle ;
\draw  [line width=0.75]  (351,426) .. controls (351,412.19) and (362.19,401) .. (376,401) .. controls (389.81,401) and (401,412.19) .. (401,426) .. controls (401,439.81) and (389.81,451) .. (376,451) .. controls (362.19,451) and (351,439.81) .. (351,426) -- cycle ;
\draw  [line width=0.75]  (351,125) .. controls (351,111.19) and (362.19,100) .. (376,100) .. controls (389.81,100) and (401,111.19) .. (401,125) .. controls (401,138.81) and (389.81,150) .. (376,150) .. controls (362.19,150) and (351,138.81) .. (351,125) -- cycle ;
\draw  [line width=0.75]  (201,376) .. controls (201,362.19) and (212.19,351) .. (226,351) .. controls (239.81,351) and (251,362.19) .. (251,376) .. controls (251,389.81) and (239.81,401) .. (226,401) .. controls (212.19,401) and (201,389.81) .. (201,376) -- cycle ;
\draw  [line width=0.75]  (451,226) .. controls (451,212.19) and (462.19,201) .. (476,201) .. controls (489.81,201) and (501,212.19) .. (501,226) .. controls (501,239.81) and (489.81,251) .. (476,251) .. controls (462.19,251) and (451,239.81) .. (451,226) -- cycle ;
\draw  [line width=0.75]  (500,325) .. controls (500,311.19) and (511.19,300) .. (525,300) .. controls (538.81,300) and (550,311.19) .. (550,325) .. controls (550,338.81) and (538.81,350) .. (525,350) .. controls (511.19,350) and (500,338.81) .. (500,325) -- cycle ;
\draw [line width=0.75]    (225,226) -- (376,125) ;
\draw [line width=0.75]    (525,325) -- (476,226) ;
\draw [line width=0.75]    (226,376) -- (376,426) ;
\draw [line width=0.75]    (225,226) -- (326,276) ;
\draw [line width=0.75]    (326,276) -- (476,226) ;
\draw [line width=0.75]    (326,276) -- (376,426) ;
\draw  [line width=0.75]  (-249,226) .. controls (-249,212.19) and (-237.81,201) .. (-224,201) .. controls (-210.19,201) and (-199,212.19) .. (-199,226) .. controls (-199,239.81) and (-210.19,251) .. (-224,251) .. controls (-237.81,251) and (-249,239.81) .. (-249,226) -- cycle ;
\draw  [line width=0.75]  (-148,276) .. controls (-148,262.19) and (-136.81,251) .. (-123,251) .. controls (-109.19,251) and (-98,262.19) .. (-98,276) .. controls (-98,289.81) and (-109.19,301) .. (-123,301) .. controls (-136.81,301) and (-148,289.81) .. (-148,276) -- cycle ;
\draw  [line width=0.75]  (-98,426) .. controls (-98,412.19) and (-86.81,401) .. (-73,401) .. controls (-59.19,401) and (-48,412.19) .. (-48,426) .. controls (-48,439.81) and (-59.19,451) .. (-73,451) .. controls (-86.81,451) and (-98,439.81) .. (-98,426) -- cycle ;
\draw  [line width=0.75]  (-98,125) .. controls (-98,111.19) and (-86.81,100) .. (-73,100) .. controls (-59.19,100) and (-48,111.19) .. (-48,125) .. controls (-48,138.81) and (-59.19,150) .. (-73,150) .. controls (-86.81,150) and (-98,138.81) .. (-98,125) -- cycle ;
\draw  [line width=0.75]  (-248,376) .. controls (-248,362.19) and (-236.81,351) .. (-223,351) .. controls (-209.19,351) and (-198,362.19) .. (-198,376) .. controls (-198,389.81) and (-209.19,401) .. (-223,401) .. controls (-236.81,401) and (-248,389.81) .. (-248,376) -- cycle ;
\draw  [line width=0.75]  (2,226) .. controls (2,212.19) and (13.19,201) .. (27,201) .. controls (40.81,201) and (52,212.19) .. (52,226) .. controls (52,239.81) and (40.81,251) .. (27,251) .. controls (13.19,251) and (2,239.81) .. (2,226) -- cycle ;
\draw  [line width=0.75]  (51,325) .. controls (51,311.19) and (62.19,300) .. (76,300) .. controls (89.81,300) and (101,311.19) .. (101,325) .. controls (101,338.81) and (89.81,350) .. (76,350) .. controls (62.19,350) and (51,338.81) .. (51,325) -- cycle ;
\draw [line width=0.75]    (-224,226) -- (-73,125) ;
\draw [line width=0.75]    (76,325) -- (27,226) ;
\draw [line width=0.75]    (-223,376) -- (-73,426) ;
\draw [line width=0.75]    (-224,226) -- (-223,376) ;
\draw [line width=0.75]    (-123,276) -- (27,226) ;
\draw [line width=0.75]    (-123,276) -- (-73,426) ;
\draw  [line width=0.75]  (651,226) .. controls (651,212.19) and (662.19,201) .. (676,201) .. controls (689.81,201) and (701,212.19) .. (701,226) .. controls (701,239.81) and (689.81,251) .. (676,251) .. controls (662.19,251) and (651,239.81) .. (651,226) -- cycle ;
\draw  [line width=0.75]  (752,276) .. controls (752,262.19) and (763.19,251) .. (777,251) .. controls (790.81,251) and (802,262.19) .. (802,276) .. controls (802,289.81) and (790.81,301) .. (777,301) .. controls (763.19,301) and (752,289.81) .. (752,276) -- cycle ;
\draw  [line width=0.75]  (802,426) .. controls (802,412.19) and (813.19,401) .. (827,401) .. controls (840.81,401) and (852,412.19) .. (852,426) .. controls (852,439.81) and (840.81,451) .. (827,451) .. controls (813.19,451) and (802,439.81) .. (802,426) -- cycle ;
\draw  [line width=0.75]  (802,125) .. controls (802,111.19) and (813.19,100) .. (827,100) .. controls (840.81,100) and (852,111.19) .. (852,125) .. controls (852,138.81) and (840.81,150) .. (827,150) .. controls (813.19,150) and (802,138.81) .. (802,125) -- cycle ;
\draw  [line width=0.75]  (652,376) .. controls (652,362.19) and (663.19,351) .. (677,351) .. controls (690.81,351) and (702,362.19) .. (702,376) .. controls (702,389.81) and (690.81,401) .. (677,401) .. controls (663.19,401) and (652,389.81) .. (652,376) -- cycle ;
\draw  [line width=0.75]  (902,226) .. controls (902,212.19) and (913.19,201) .. (927,201) .. controls (940.81,201) and (952,212.19) .. (952,226) .. controls (952,239.81) and (940.81,251) .. (927,251) .. controls (913.19,251) and (902,239.81) .. (902,226) -- cycle ;
\draw  [line width=0.75]  (951,325) .. controls (951,311.19) and (962.19,300) .. (976,300) .. controls (989.81,300) and (1001,311.19) .. (1001,325) .. controls (1001,338.81) and (989.81,350) .. (976,350) .. controls (962.19,350) and (951,338.81) .. (951,325) -- cycle ;
\draw [line width=0.75]    (676,226) -- (827,125) ;
\draw [line width=0.75]    (976,325) -- (927,226) ;
\draw [line width=0.75]    (677,376) -- (827,426) ;
\draw [line width=0.75]    (676,226) -- (777,276) ;
\draw [line width=0.75]    (777,276) -- (927,226) ;
\draw [line width=0.75]    (777,276) -- (827,426) ;
\draw [line width=0.75]  [dash pattern={on 4.5pt off 4.5pt}]  (927,226) -- (676,226) ;
\draw [line width=0.75]  [dash pattern={on 4.5pt off 4.5pt}]  (927,226) -- (827,426) ;
\draw [line width=0.75]  [dash pattern={on 4.5pt off 4.5pt}]  (827,426) -- (676,226) ;

\end{tikzpicture}

\caption{A sketch of 7-bead polymer molecules with chain-, star- and net-shaped structures (from left to right). The solid lines represent the FENE bond potential with the \emph{same}
interaction parameters. 
The dashed lines of the net-shaped molecule represent the three additional side chains connecting the polymer arms.
While both the chain- and the star-shaped molecules are connected with six bonds; the suspensions exhibit different hydrodynamics due to the different micro-structural mechanics as shown below.}
\label{fig:molecule_shape}
\end{figure}
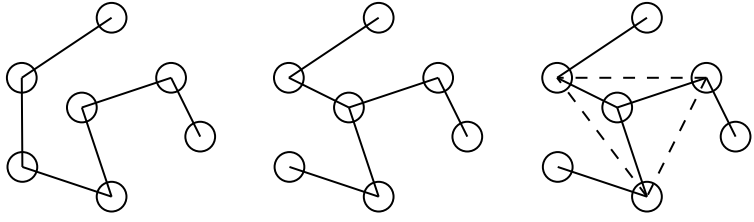

In principle, the  viscoelastic response of the system is  determined by the full micro-scale interaction. However, 
 direct simulation for the full micro-scale interaction is often limited by the 
prohibited computational cost. Continuum hydrodynamics models based on various empirical constitutive
models are often used, with the general form
\begin{equation}
\begin{split}
\nabla \cdot \mb u &= 0, \\
\rho \frac{\textrm d \mb u}{\textrm d t} &= -\nabla p + \nabla \cdot (\bm\tau_{\textrm s} + \bm\tau_{\textrm p}) + \mb f_{\textrm{ext}},
\end{split}
\label{eq:momentum_transport_close}
\end{equation}
where $\rho$, $\mb u$ and $p$ represent the fluid density, velocity and pressure field, respectively. $\mb f_{\textrm{ext}}$ is the external body force and 
$\bm\tau_{\textrm s} = \eta_{\textrm s}(\nabla \mb u + \nabla \mb u^T)$ is the solvent stress tensor with shear viscosity $\eta_s$. 
$\bm\tau_{\textrm p}$ is the polymer stress tensor whose detailed form
is generally unknown.  To construct $\bm\tau_{\textrm p}$, the DeePN$^2$ model seeks the approximation in terms of a set of macro-scale features $\mb c_1, \cdots, \mb c_n$, and simultaneously, the constitutive dynamics of these features, i.e.,
\begin{subequations}
\begin{align}
\bm\tau_{\textrm p} &= \mb G(\mb c_1, \cdots, \mb c_n), 
\label{eq:DeePN2_G}\\
\frac{\mathcal{D}\mb c_i}{\mathcal{D}t} &= \mb H_i(\mb c_1, \cdots, \mb c_n), \quad i = 1, \cdots, n,
\label{eq:DeePN2_H}
\end{align}
\label{eq:DeePN2_constitutive}%
\end{subequations}
where $\mb G$ and $\mb H_i$ represent the stress and constitutive models, respectively. $\frac{\mathcal{D}}{\mathcal{D}t}$ denotes
 the objective tensor derivative. 

 Eqs. \eqref{eq:momentum_transport_close} and \eqref{eq:DeePN2_constitutive} take the form similar to the conventional hydrodynamics.
 Instead of using empirical approximation to close the equation, 
%
we aim to construct a  model  
directly 
from the micro-scale description \eqref{eq:micro-scale_model} with the help of machine learning,
such that the constructed model can naturally encode the 
molecular-specific interaction beyond  empirical approximations with clear physical interpretation.

\subsection{DeePN$^2$ for arbitrary molecular structural mechanics}
To learn Eq. \eqref{eq:momentum_transport_close} from the full model  \eqref{eq:micro-scale_model}, one essential problem lies in how
to seamlessly pass the micro-scale interaction to the continuum model. To bridge the scales, we learn a set of micro-to-macro 
encoders, denoted by $\left\{\mb b_i(\mb q)\right\}_{i=1}^{n}$, 
such that the continuum modeling terms (e.g., the polymer stress $\bm\tau _{\textrm p}$) can be  well approximated in terms of the corresponding
macro-scale features $\left\{\mb c_i(\mb q)\right\}_{i=1}^{n}$ via Eq. \eqref{eq:DeePN2_G},
where $\bm\tau_{\textrm p} := n_{\textrm p} \sum_j \langle \mb q_j \otimes \nabla_{\mb q_j} V(\mb q)\rangle$,  
$\mb c_i = \left\langle \mb b_i(\mb q)\right\rangle$, $n_{\textrm p}$ is the polymer number density 
and $\left\langle \cdot \right\rangle$ denotes the average with respect to the configuration PDF. In particular, the features
$\mb c_i$ need to satisfy the proper invariant and symmetry conditions inherited from the encoders $\mb b_i(\cdot)$ such that 
the constructed continuum model can strictly preserve  frame-indifference condition: 
\begin{equation}
\widetilde{\bm\tau_{\textrm p}} = \mb Q \bm\tau_{\textrm p} \mb Q^T, \quad 
\mb G(\widetilde{\mb c}_1, \cdots, \widetilde{\mb c}_n) = \mb Q \mb G(\mb c_1, \cdots, \mb c_n)\mb Q^T,
\label{eq:symmetry_G}
\end{equation}
where the superscript $\widetilde{\cdot}$ denotes the corresponding values under an arbitrary orthogonal  transformation 
by 
$\mb Q \in \textrm{SO}(3)$.

To construct the encoder $\mb b(\cdot)$, we note that the micro-scale potential $V(\mb q)$ is translational and rotational invariant.
Accordingly, let $\mb r^{\ast} (\mb q) \in \mathbb{R}^{3N-6}$ \revision{(we consider the general case $N\geq3$ here)} denote the \revision{translational-}rotational-invariant configuration vector and $\mb r (\mb q) \in \mathbb{R}^{3N-3}$
denote the translational-invariant configuration vector consisting of $N-1$ linearly independent position vectors. Since $N_b \ge N-1$ for all
molecules, one straightforward choice is the first $N-1$ bond connection vectors, i.e., 
\begin{equation}
\begin{split}
\mb r&= \left[\mb r_1; \mb r_2; \cdots; \mb r_{N-1}\right], \quad  \mb r_j = \mb q_{j_1} - \mb q_{j_2}, \quad 1\le j \le  N-1, \\
\mathbf{r}^{\ast}
&=
\left[
\left| \mathbf{r} _{1} \right|,
\left| \mathbf{r} _{2} \right|,
\left| \mathbf{r} _{12} \right|,
\left| \mathbf{r} _{3} \right|,
\left| \mathbf{r} _{13} \right|,
\left| \mathbf{r} _{23} \right|,
\left| \mathbf{r} _{4} \right|,
\left| \mathbf{r} _{24} \right|,
\left| \mathbf{r} _{34} \right|,
\cdots,
\left| \mathbf{r} _{N-1} \right|,
\left| \mathbf{r} _{(N-2)(N-1)} \right|
\right],
\end{split}
\label{eq:r_r_ast}
\end{equation}
where $\mathbf{r}_{jk} := \mathbf{r}_j - \mathbf{r}_k$.
We note that this form  applies to  general molecular structures; $\mb r$ determines the molecular structure up to translations. Specifically, $\mb r^{\ast}$ represents the $3N-6$  degrees of freedom after eliminating  translational and rotational 
degrees of freedom, and 
$\mb r$ suffices to fully determine the translational invariant polymer configuration and strictly retains the rotational symmetry in 
accordance with $\mb q$, i.e., 
\begin{equation}
\mb r_j(\mb Q\mb q) = \mb Q \mb r_j(\mb q), \quad  \mb r^{\ast}(\mb Q\mb q) = \mb r^{\ast}(\mb q).
\nonumber
\end{equation}

To preserve  rotational symmetry, one straightforward approach is to represent $\mb b(\cdot)$ in the linear space spanned by $\left\{\mb r_j\right\}_{j=1}^{N-1}$. However, this choice
yields the trivial macro-scale feature, i.e., $\left\langle \mb r_j\right\rangle \equiv 0$, due to the rotational symmetry.
Alternatively, we construct the following second-order tensor 
\begin{equation}
\begin{split}
\mb c_i &= \langle \mb b_i(\mb r) \rangle, \quad \mb b_i = \mb f_i\mb f_i^T, \quad 
 \quad 1\le i \le n, \\
\mb f_i &= g_i(\mb r^{\ast}) \sum_{j=1}^{N-1} w_{ij} \mb r_j,\\
\end{split}
\label{eq:encoder_model}
\end{equation}
where 
$[w_{ij}]_{1\le i \le n,  1 \le j \le N-1}$ are the weights and $\{ g_{i}(\cdot) \} _{i=1}^n$ is a set of scalar functions that encodes the polymer 
intramolecular interaction. Both terms will be learned from the micro-scale description and represented by deep neural networks (DNNs). 
Rotational symmetries can be naturally inherited, i.e., $\widetilde{\mb c} = \left\langle \mb b(\widetilde{\mb r})\right\rangle \equiv \mb Q \mb c\mb Q^T$.
Compared with the special form for dumbbell molecules  in Ref. \cite{Lei_Wu_E_2020},  Eq. \eqref{eq:encoder_model} provides a general form of $\mb c$
applicable to multi-bead molecules of arbitrary structure since $\mb r$ and $\mb r^{\ast}$ fully determine the $3N-3$ translational
invariant polymer configuration. In the remaining of the paper, we will abuse the notation and denote $\mb b(\mb q)$ as $\mb b(\mb r)$.

Besides the polymer stress model \eqref{eq:DeePN2_G}, the remaining task to close Eq. \eqref{eq:momentum_transport_close} is the construction of the constitutive 
dynamics \eqref{eq:DeePN2_H} of the macro-scale features $\left\{\mb c_i\right\}_{i=1}^n$. There are two issues to deal with: the proper form of the objective time derivative of $\mb c_i$ and 
the accurate estimation of their time evolution. In the literature, the objective tensor derivative, denoted by $\frac{\mathcal{D} \mb c_i}{\mathcal{D} t}$, 
is often chosen to take some heuristic forms (e.g. the convected \cite{Oldroyd_Wilson_PRSLA_1950} and corotational \cite{Zaremba_1903} forms). 
Moreover, the time-series samples collected from the micro-scale simulations are generally super-imposed with pronounced sampling error; 
direct estimation of the time derivative  as was done in \cite{Rudy_Kutz_Science_Ad_2017, raissi2019physics,  Naoki_Takeshi_PRR_2020} will end with  noisy data. 
Fortunately, both challenges are 
addressed in DeePN$^2$ using an explicit micro-macro correspondence. The dynamics of $\mb c_i$ can
be derived from the its micro-scale correspondence $\mb b_i(\mb r)$   in the form of the micro-scale 
configuration $\mb r$, i.e.,
\begin{equation}
\begin{split}
\frac{\textrm d}{\textrm dt} \mb c_i - \bm\kappa:\left\langle \sum_{j=1}^{N-1}
\mb r_j \otimes \nabla_{\mb r_j}\otimes \mb b_i \right\rangle
 &= 
\frac{k_BT}{\gamma} \left\langle \sum_{j,k=1}^{N-1}  A_{jk} \nabla_{\mb r_j} \cdot 
\nabla_{\mb r_k} \mb b_i \right\rangle \\
&- 
\frac{1}{\gamma} \left\langle \sum_{j=1}^{N-1} \sum_{k=1}^{N_b}  A_{jk} \nabla_{\mb r_k} V(\mb r_1, \cdots, \mb r_{N_b}) \cdot \nabla_{\mb r_j} 
\mb b_i \right\rangle,
\end{split}
\label{eq:FK_B_evoluation_multi}
\end{equation}
where $\bm\kappa := \nabla \mb u^T$, 
$\gamma$ is the friction coefficient and $\mb r_j$ is the connection vector as defined in 
Eq. \eqref{eq:r_r_ast} for $j > N-1$. We abuse the notation and denote $ V(\mb q)$ as $ V(\mb r_1, \cdots, \mb r_{N_b}) = \sum_{j=1}^{N_b} V_b(r_j)$.  The molecular structure and interaction are specified via 
$\mb A \in \mathbb{R}^{N_b\times N_b}$, which is defined by 
\begin{equation}
  \mb A = \mb S \mb S^T, \quad S_{jk} = \left\{
  \begin{array}{ll}
    +1, & k = {j_1}, \\
    -1, & k = {j_2}, \\
    0, & \text{else}
  \end{array}
  \right.
  \quad 1 \le j \le N_b, \quad 1 \le k \le N,
\end{equation}
\revision{where $j_1$ and $j_2$ are the same notations as those in Eq. \eqref{eq:micro-scale_model}.}
We note that Eq. \eqref{eq:FK_B_evoluation_multi} only requires the first $(N-1)$ rows of $\mb A$ since the polymer configuration can be fully determined by $\mb r_1, \cdots, \mb r_{N-1}$. As a special case, if the molecule takes the 
chain shape, $\mb A$ recovers the standard Rouse matrix \cite{Bird_Curtiss_book_vol_2,Rouse_JCP_1953}.

Eq. \eqref{eq:FK_B_evoluation_multi} defines the dynamics for the features $\left\{\mb c_i\right\}_{i=1}^n$,
 derived from  their micro-scale correspondences. In particular,
given the proposed form of the encoder functions \eqref{eq:encoder_model}, 
we can show that  the two combined terms of the left-hand-side of Eq. \eqref{eq:FK_B_evoluation_multi} strictly preserve  rotational symmetry (see Appendix \ref{sec:proof_rotation_dynamics}). This 
leads to an important observation that the two combined terms provide the generalized form for the macro-scale objective tensor derivative
$\frac{\mathcal{D} \mb c_i}{\mathcal{D} t}$. Unlike the heuristic choices in empirical models, the new form retains a clear micro-scale
physical interpretation. Furthermore, all the modeling terms in the form of $\left\langle \cdot \right\rangle$ 
can be directly evaluated using samples collected from the micro-scale simulations under the corresponding flow condition.
This enables us to avoid estimating the time derivative values from the noise-prone time-series samples. Accordingly, the macro-scale constitutive
dynamics takes the form 
\begin{equation}
\frac{\diff \mb c_i}{\diff t} - \bm\kappa : \mathcal{E}_i = \frac{k_BT}{\gamma} \mb H_{1,i}(\mb c_1, \cdots, \mb c_n) - \frac{1}{\gamma} \mb H_{2,i}(\mb c_1, \cdots, \mb c_n),
\label{eq:c_constituitve}
\end{equation}
where the individual terms will be represented by proper neural networks and parameterized by matching their micro-scale correspondences, i.e.,
\begin{equation}
\begin{split}
\mathcal{E}_i(\mb c_1, \cdots, \mb c_n) &= \left\langle \sum_{j=1}^{\revision{N-1}} \mb r_j \otimes\nabla_{\mb r_j}\otimes \mb b_i \right\rangle, \\
\mb H_{1,i}(\mb c_1, \cdots, \mb c_n) &= \left\langle \sum_{j,k=1}^{N-1}  A_{jk} \nabla_{\mb r_j} \cdot 
\nabla_{\mb r_k} \mb b_i \right\rangle, \\
\mb H_{2,i}(\mb c_1, \cdots, \mb c_n) &= \left\langle \sum_{j=1}^{N-1} \sum_{k=1}^{N_b}  A_{jk} \nabla_{\mb r_k} V(\mb r_1, \cdots, \mb r_{N-1}) \cdot \nabla_{\mb r_j} 
\mb b_i \right\rangle .
\end{split}
\end{equation}

\subsection{Symmetry-preserving DNN models}
To complete the DeePN$^2$ model, we need to specify the DNN models.
These DNN models should also strictly preserve  rotational symmetry.
Different from the rotational-invariant scalar stress model considered in Ref. \cite{Han_Xiao_CMAME_2022}, the second-order tensors $\mb G$, $\mb H_{1,i}$, $\mb H_{2,i}$ need to satisfy the symmetry condition \eqref{eq:symmetry_G} and the fourth-order tensors $\mathcal{E}_i$ need to retain the objectivity of $\frac{\mathcal{D} \mb c_i}{\mathcal{D} t}$. 
However, there does not exist such a reference frame in which these symmetry constraints can be satisfied by the
macro-scale modeling terms.

To handle this problem, we consider the eigen-space of the feature $\mb c_1$ with a fixed form of the encoder 
$\mb b_1(\cdot)$, e.g., by setting $g_1(\cdot) = w_{1,:} \equiv 1$
\revision{and let other $\mb b_i(\cdot)$ involved in the training}. 
Let us consider the eigen-decomposition $\mb c_1 = \mb U\Lambda \mb U^T$, assuming that it has distinct eigenvalues. 
We introduce the following matrices
\begin{equation}
S^{(1)} = \begin{pmatrix} +1 & & \\
 &+1 & \\
 & &+1
\end{pmatrix}, 
S^{(2)} = 
\begin{pmatrix} +1 & & \\
 &-1 & \\
 & &+1
\end{pmatrix},
S^{(3)} = 
\begin{pmatrix} +1 & & \\
 &+1 & \\
 & &-1
\end{pmatrix}, 
S^{(4)} = 
\begin{pmatrix} +1 & & \\
 &-1 & \\
 & &-1
\end{pmatrix}.
\nonumber
\end{equation}
We denote $\mb U^{(j)} = \mb U S^{(j)}$ and ${{\hat{\mb c}}_i}^{(j)} = {{\mb U}^{(j)}}^T {\mb c}_i {\mb U}^{(j)}$. We can show that the formulation of the stress model
$\mb G = \revision{\frac{1}{4}} \sum_{j=1}^4 \mb U^{(j)} \hat{\mb G}(\hat{\mb c}_1^{(j)}, \cdots,  \hat{\mb c}_n^{(j)}) {\mb U^{(j)}}^T$ satisfies 
Eq. \eqref{eq:symmetry_G} (see Appendix \ref{sec:symmetry_preserving_DNN}). 

During simulation, the eigenvalues of $\mb c_1$ may cross each other. To account for this, we consider all the 
$6$ permutations of the three eigenvalues, i.e., 
\begin{equation}
\mb G(\mb c_1, \cdots, \mb c_n) = 
\revision{\frac{1}{24}}
 \sum_{k=0}^5 \sum_{j=1}^4 \mb U^{(j,k)} \hat{\mb G}
(\hat{\mb c}_{1}^{(j,k)}, \cdots,  \hat{\mb c}_{n}^{(j,k)}) {\mb U^{(j,k)}}^T,
\label{eq:permutation_G}
\end{equation}
where $k$ represents the rank of  permutation (e.g., in lexicographical order)
and $\mb U^{(j,k)}$ is a variation of $\mb U^{(j)} $ with corresponding column permutation. 
Furthermore, to avoid 
the eigenvector degeneracy, we set a threshold value $\epsilon$ for the eigenvalues. When two eigenvalues approach each other, e.g.,
$\vert \lambda_2 - \lambda_3 \vert < \epsilon$, \revision{we freeze all the eigenvectors until $\vert \lambda_2 - \lambda_3 \vert \geq \epsilon$}. In this work, we take $\epsilon = 10^{-3}$, and we refer to Appendix \ref{sec:symmetry_validation}
for detailed numerical studies. 

Eq. \eqref{eq:permutation_G} provides the rotation-symmetric form for the second-order stress tensor $\mb G$, where $\hat{\mb G}$
is represented by DNNs. The constitutive model terms $\mb H_{1,i}$ and $\mb H_{2,i}$ can be constructed in a similar manner.
Finally, we can show the fourth-order tensors $\left\{\mathcal{E}_i\right\}_{i=1}^n$ associated with the encoders \eqref{eq:encoder_model}
 can be constructed in the form 
\begin{equation}
\bm\kappa :\mathcal{E}_i = \bm\kappa \mb c_i + \mb c_i \bm\kappa^T + \bm\kappa :\left(\sum_{j=1}^9 {\mb E}_{1,i}^{(j)}\otimes {\mb E}_{2,i}^{(j)}\right),
\label{eq:permutation_E}
\end{equation}
where ${\mb E}_{1,i}^{(j)}$ and ${\mb E}_{2,i}^{(j)}$ are second-order tensors which respect the symmetry condition \eqref{eq:symmetry_G} and can be constructed 
in the form of Eq. \eqref{eq:permutation_G} (see Appendix \ref{sec:symmetry_preserving_DNN}). 
The constructed DeePN$^2$ model takes the form similar to 
the general hydrodynamic equations \eqref{eq:momentum_transport_close} and \eqref{eq:DeePN2_constitutive}, where some of the
 model terms are represented by DNNs in the form of Eqs. \eqref{eq:permutation_G} and \eqref{eq:permutation_E}.

\subsection{Algorithm}
\revision{We summarize the DeePN$^2$ model in  Algorithm \ref{alg:DeePN2}.
\begin{algorithm}
\hrulefill
\caption{DeePN$^2$ for polymer suspensions retaining micro-structural fidelity.}
\label{alg:DeePN2}
\begin{algorithmic}[1]
\State Conduct the micro-scale simulations (see Appendix  \ref{sec:micro_scale}) and collect time-discrete training samples (see Appendix \ref{sec:training_samples}).
\State Pre-process the training samples by pre-computing the first conformation tensor $\mb c_1 = \left\langle \mb b_1(\mb r)\right\rangle$, its eigen-decomposition, and the polymer stress based on the micro-scale polymer configurations for each training sample set.
\State End-to-end training: Establish the joint learning of the symmetry-preserving encoders ($\mb b_2(\cdot), \cdots, \mb b_n(\cdot)$) (see Appendix \ref{sec:symmetry_preserving_DNN}), the macro-scale DNN functions (the stress $\mb G(\cdot)$ and the other constitutive modeling terms $\left\{\mathcal{E}_i(\cdot), \mb{H}_{1,i}(\cdot), \mb{H}_{2,i}(\cdot)\right\}_{i=1}^n$) by matching the micro-scale counterparts derived from the structure-specific kinetic equations (see Appendix \ref{sec:training}).
\State Solve the macro-scale hydrodynamic Eqs. \eqref{eq:momentum_transport_close}  and \eqref{eq:c_constituitve}.
\end{algorithmic}
\hrulefill
\end{algorithm}
}

\section{Numerical results}
The present DeePN$^2$ model is trained using micro-scale samples collected from the homogeneous shear flow. We demonstrate the model accuracy and generalization 
ability by considering various flows in comparison with the results of the micro-scale simulations for the suspensions with three different polymer structural models as shown in Fig. \ref{fig:molecule_shape}. As we will see, the micro-scale structure does play an important role in the viscoelastic response. We will use this to examine the DeePN$^2$ model fidelity.

First, we consider the reverse Poiseuille flow in a $60\times 100\times 60$ domain (in reduced unit) 
with the opposite body force $\mb f_{\textrm{ext}} = (0.016, 0, 0)$ applied to each
half of the domain divided by the plane $y = 50$ starting from $t = 0$. At $t = 800$, the external force is removed. The relaxation process of the flow field
is recorded until the total simulation time $t = 1600$. For all the three systems, the predictions from DeePN$^2$ agree well with the micro-scale simulations 
results, as shown in Fig. \ref{fig:RPF_flow_development}. In particular, the flow velocity fields of the three
systems are nearly identical at the initial stage $t \in [0, 200]$, as the development of the flow field is dominated by the solvent 
and the near-equilibrium responses of the polymer molecules in this regime. Starting from $t = 250$, the velocity fields of the three systems
exhibit distinct evolution processes. The velocity of the chain-shaped molecule suspension exhibits the largest oscillation and the longest development stage during $t \in [250, 800]$. 
In contrast, the velocity of the star-shaped molecule suspension exhibits moderate oscillation and shows an apparent increase during $t \in [400, 800]$, indicating that
the polymer elastic energy reaches a plateau  earlier than the chain-shaped system. Moreover, the velocity of the net-shaped molecule suspension exhibits the smallest 
oscillation, indicating that the three additional side-chains further affect the rheological properties of the polymer suspension.

\begin{figure}[htbp]
\centering
\includegraphics[scale=0.25]{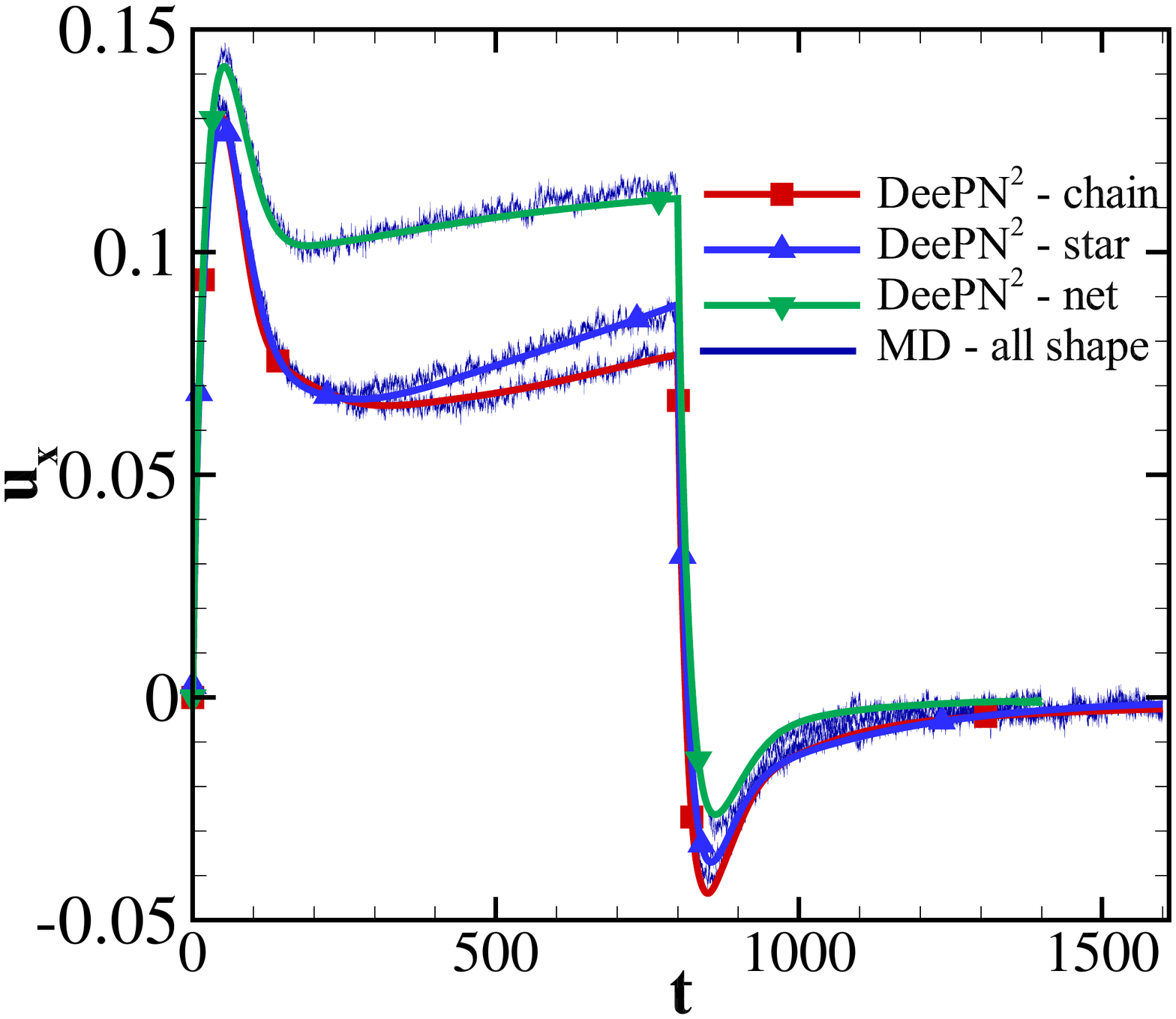}
\includegraphics[scale=0.25]{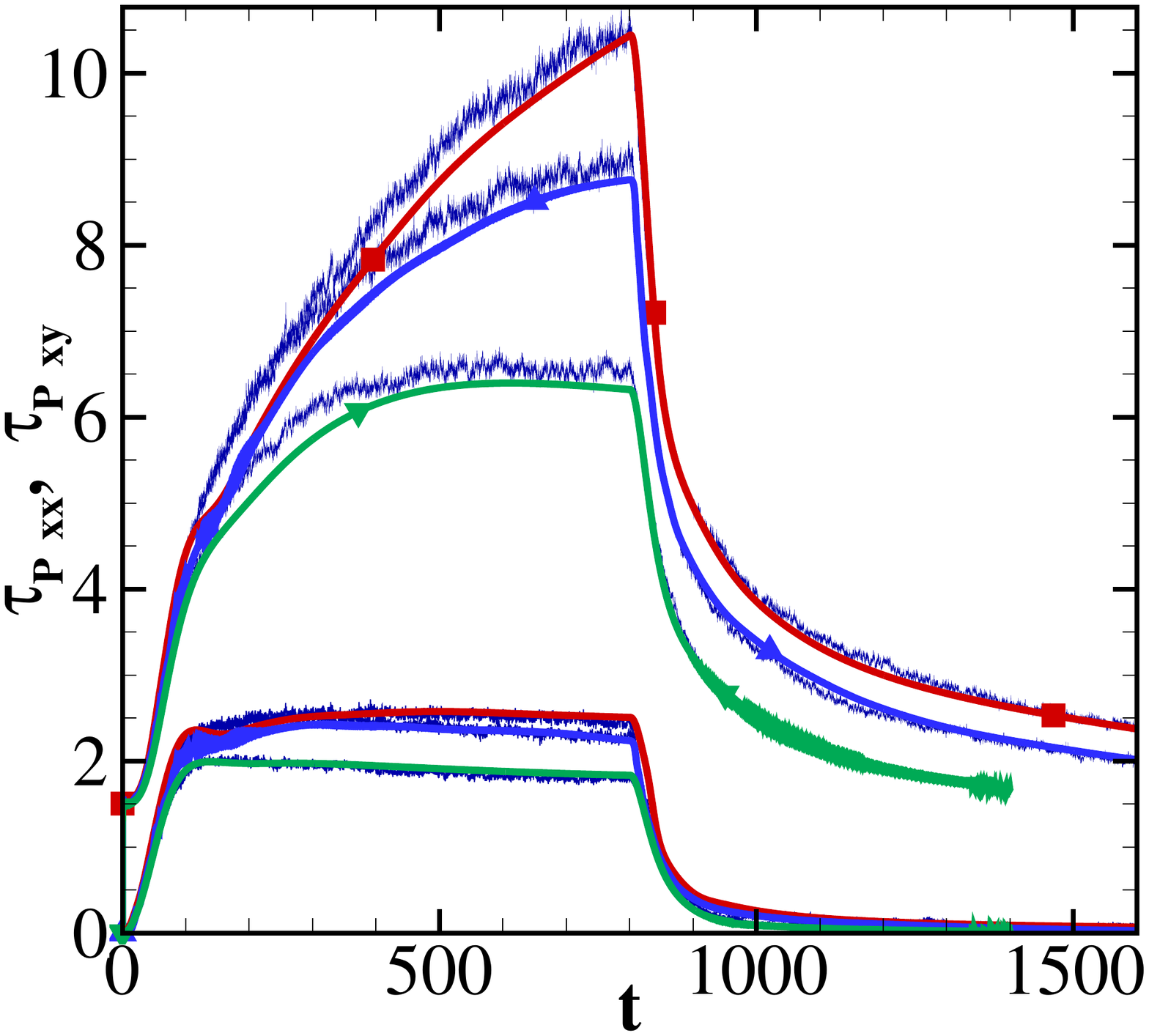}
\caption{The velocity $ u_x$ (left) and  polymer stress $\bm \tau _{\textrm p}$ (right)  of the reverse Poiseuille flow ($y = 6$) of the polymer suspensions of three different molecule structures shown in Fig. \ref{fig:molecule_shape}. $\bm \tau _{\textrm p}$ is normalized by polymer number density $n_{\textrm{p}}$, i.e., it is the stress energy per polymer (the same for the remaining figures). With the same FENE bond, the polymer suspensions exhibit different flow responses due to the different molecule structural mechanics. \revision{The dark blue lines with rough oscillations} denote the micro-scale simulation results; the solid lines with symbols denote the DeePN$^2$ predictions. }
\label{fig:RPF_flow_development}
\end{figure}

Such differences can also be studied by examining the polymer stress development.  As shown in
Fig. \ref{fig:RPF_flow_development},  the value of ${\bm \tau_{\textrm p}}_{xx}$ for the chain-shaped molecule suspension keeps increasing through the 
development stage $t \in [0, 800]$ while for the star-shaped molecule, 
${\bm \tau_{\textrm p}}_{xx}$  shows only a moderate increase. In contrast, the net-shaped molecule suspension reaches steady state at about 
$t = 400$. Moreover, the steady value of the shear stress ${\bm \tau _{\textrm{p}}}_{xy}$ of the chain-shaped molecule is also larger than the star-shaped and 
the net-shaped molecules, indicating the largest restored elastic energy. This result is also consistent with the larger velocity 
oscillation from the minimal values to $0$ during the relaxation process with $t \in [800, 1000]$.

The different rheological properties of the three polymer suspensions can be understood as follows. 
Although both the chain-shaped and star-shaped molecules 
have \emph{$6$ identical} FENE bonds, the chain-shaped molecule is less symmetric than the star-shaped molecule. Accordingly, it shows 
larger dispersion in the $\mathbb{R}^{18}$ configuration space, and hence, is more flexible than the star-shaped molecule. The 
elastic response time of the chain-shaped molecule suspension is longer than that of the 
star-shaped molecule suspension; larger elastic energy can be restored during the relaxation
stage. On the other hand, the net-shaped molecule is more rigid than the star-shaped molecule due to the additional bond interaction. 

Another important feature of non-Newtonian fluids is the hysteresis effect.  Classical models such as Hookean and FENE-P cannot 
 capture such effects \cite{Dolye_Shaqfeh_1998, Lielens_FENE_L_1998}. Fig. \ref{fig:RPF_stress_c_loop} shows the evolution of the polymer stress and conformation 
tensor for the chain- and star-shaped molecule suspensions.  The clockwise loops show the hysteresis effects during the development and relaxation processes; the non-unique stress 
values indicate that linear and mean field approximations are insufficient in 
describing the viscoelastic response of the system. 
In contrast, these effects are accurately captured with the DeePN$^2$ model.
Similar to  Fig. \ref{fig:RPF_flow_development}, the chain-shaped molecule suspension shows more pronounced 
hysteresis effect due to the larger dispersion in the configuration space, reflected as the larger ``loop area'' than the results for 
star-shaped molecule suspension. 

\begin{figure}[htbp]
\centering
\includegraphics[scale=0.25]{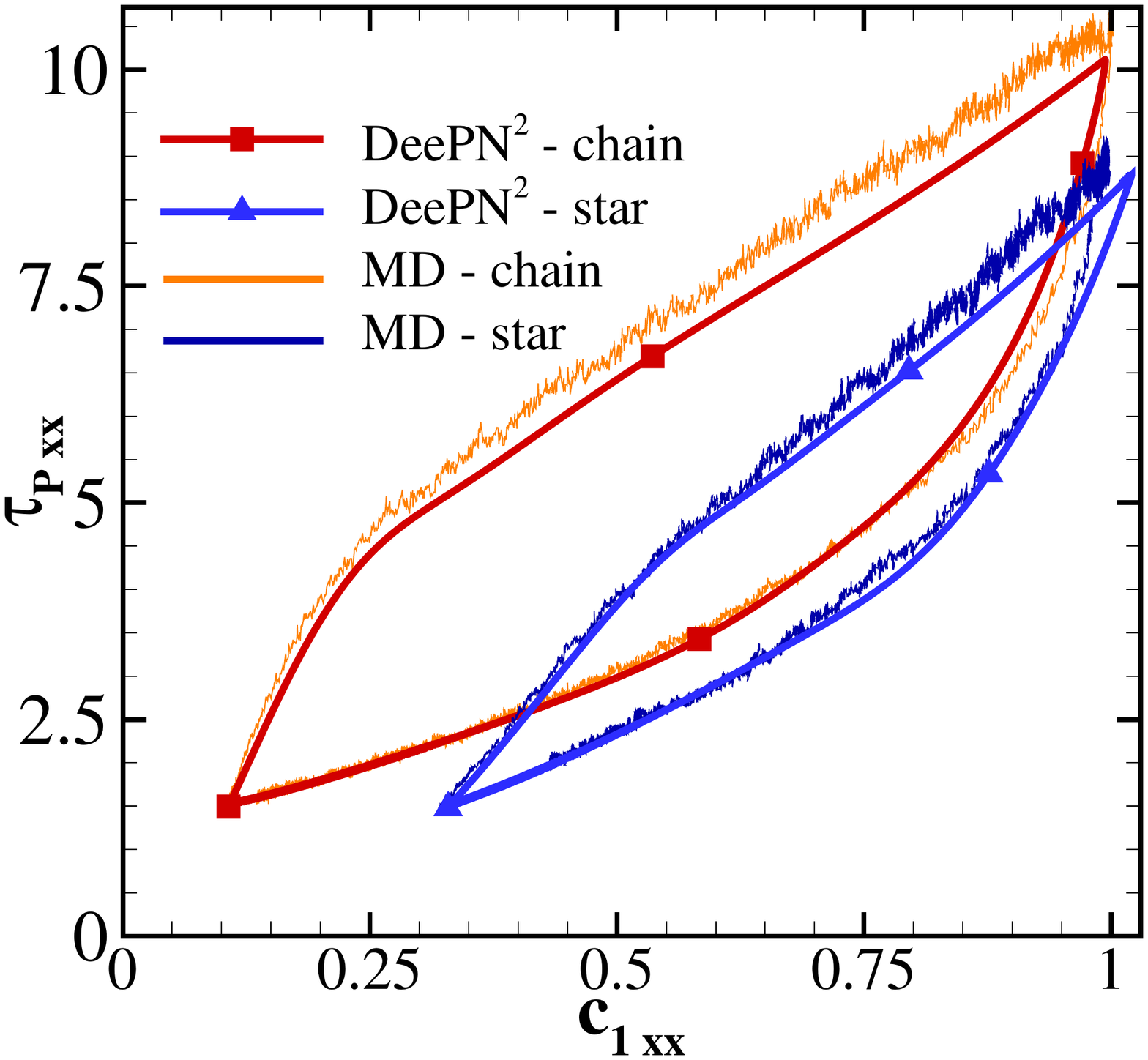}
\includegraphics[scale=0.25]{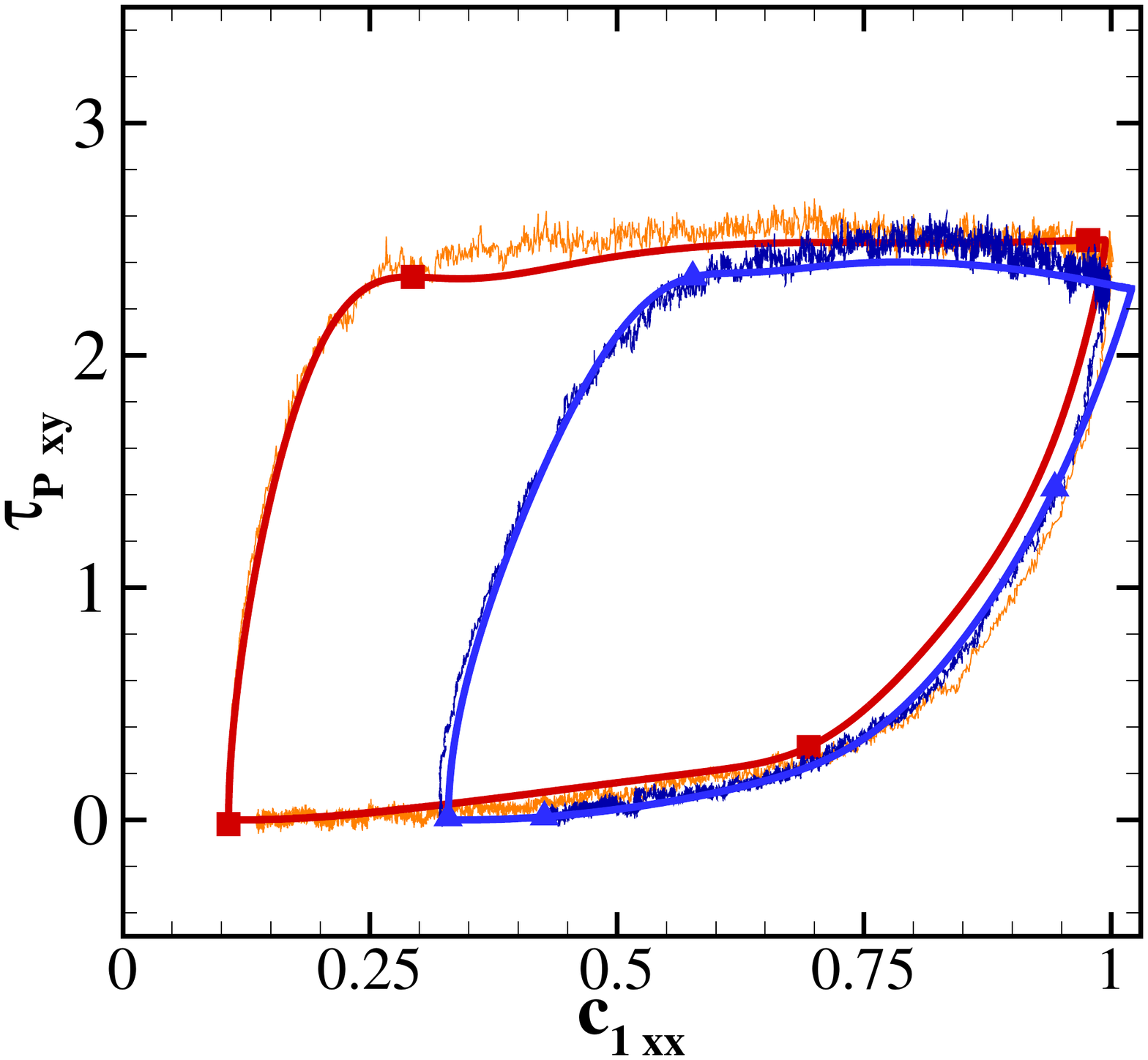}
\caption{The evolution of the polymer stress $\bm \tau _{\textrm p}$ and conformation tensor ${\mb c_1}$ obtained from the reverse Poiseuille flow ($y$ = 6) of the polymer suspensions. The clockwise loops represent the development and relaxation processes. For the visualization, the conformation tensor component ${{ c_1}_{xx}}$ is rescaled by the maximum value obtained from the micro-scale simulation. }
\label{fig:RPF_stress_c_loop}
\end{figure}

Next, we investigate the Womersley flow \cite{Womersley_flow_1955} by applying the opposite oscillating  body force 
$\mb f_\textrm{ext} = \left(\pm f_0\cos(2\pi\omega t), 0, 0\right)$ to each half of the domain along the z-direction, where we set
$f_0 = 0.012$ and $\omega = 1/3000$. Fig. \ref{fig:Womersley_flow_development} shows the velocity development of the star- and net-shaped molecule suspensions. Similar to
the reverse Poiseuille flow, the net-shaped molecule suspension shows less pronounced viscoelastic responses, reflected as the slower decay 
near $t \in [200, 400]$ and the larger oscillation due to the less elastic energy storage. For comparison, we also 
show the prediction from the conventional FENE-P model. The parameters are chosen to match the dynamics of the orientation tensor 
(the vector between two free-end particles) near equilibrium. As expected, the FENE-P model shows limitations for predicting the 
flow responses of the two suspensions. 

\begin{figure}[htbp]
\centering
\includegraphics[scale=0.25]{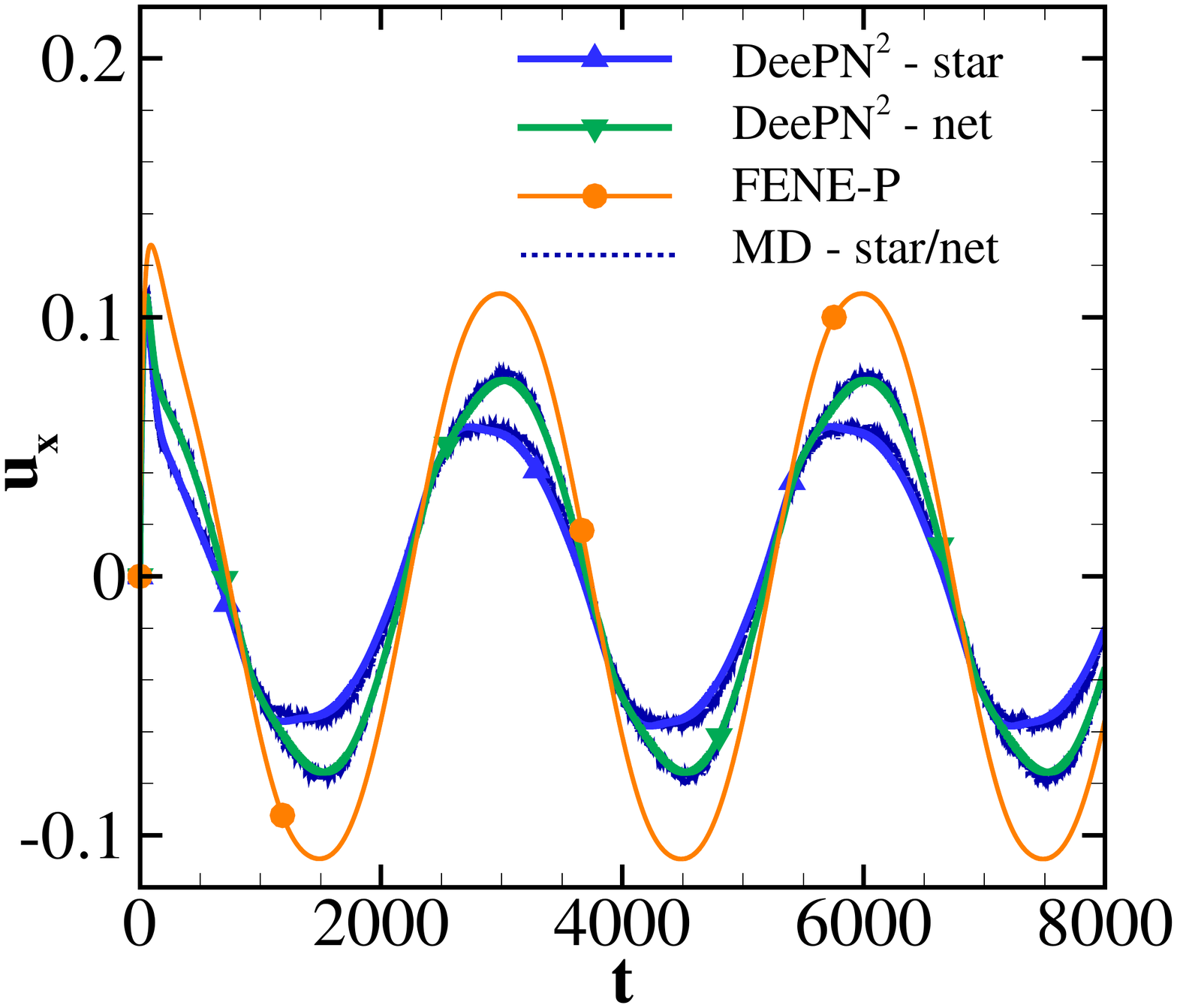}
\includegraphics[scale=0.25]{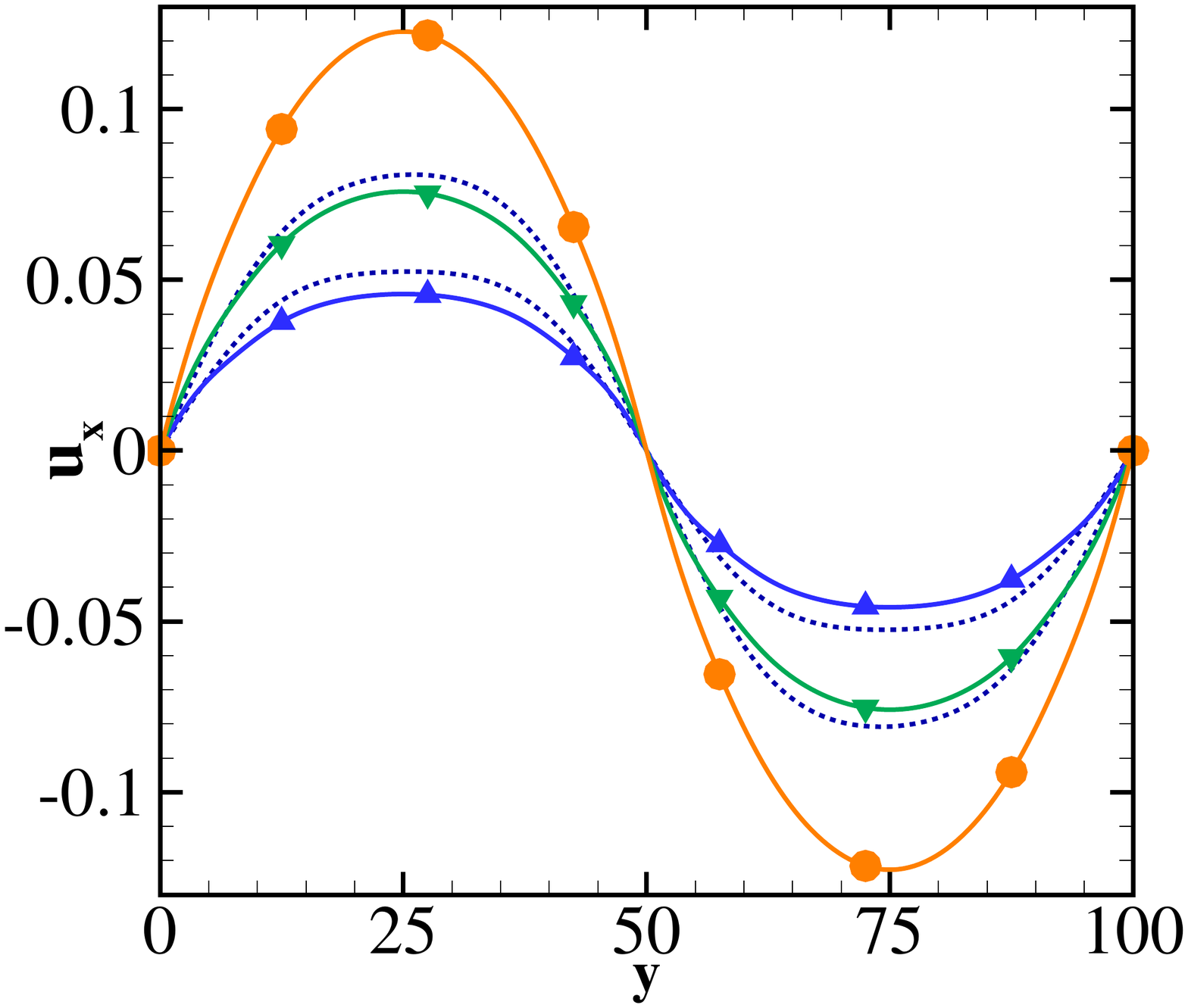}
\caption{The oscillating Womersley flow of the star- and net-shaped molecule suspensions predicted from the micro-scale simulation,  DeePN$^2$ and the FENE-P model. The FENE-P model parameters are chosen to match the dynamics of the  orientation tensor (the vector between two free-end particles) near equilibrium. Left: the velocity evolution $u_x(y, t)$  at $y = 6$. Right: the velocity profile $u_x(y, t)$ at $t= 6450$.}
\label{fig:Womersley_flow_development}
\end{figure}

The distinct viscoelastic responses of the different suspensions can be further elucidated by examining the elongation flow. We impose the traceless flow gradient $\nabla \mb u = {\textrm{diag}}(\dot{\epsilon}, -\dot{\epsilon}, 0)$ where the 
strain rate $\dot{\epsilon}$ is set to be $4\times 10^{-4}$. Fig. \ref{fig:Elongation_flow_development} shows the stress development of the 
chain- and star-shaped molecule suspensions. The micro-scale simulations are imposed by the generalized uniaxial extension flow boundary conditions \cite{Nicholson_UEF_JCP_2016, Murashima_Hagita_UEFEX_2018}.
Compared with the shear flow, the elongation flow yields larger extension and longer processes, as was  shown in experimental studies \cite{Douglas_Hazen_Science_1999}; the steady state is achieved at
about $t = 2.5\times 10^3$ and $t=10^4$ for the star- and chain-shaped molecule suspensions, respectively. Moreover, the steady stress value ${\tau _{\textrm{p}}}_{xx}$
of the chain-shaped molecule suspension is much larger than the value of the star-shaped molecule suspension. Such differences are also due to the larger 
flexibility of the chain-shaped molecule, which produces a stronger extension under external flow. DeePN$^2$ successfully
captures the different responses and shows good agreement with the micro-scale simulations for both cases.

\begin{figure}[htbp]
\centering
\includegraphics[scale=0.25]{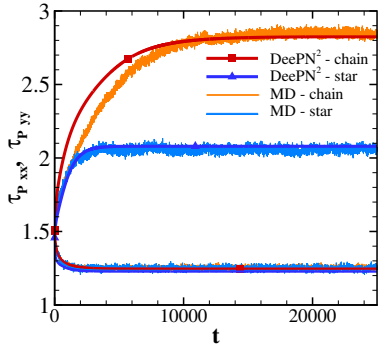}
\caption{The elongation flow of the chain- and star-shaped molecule suspensions predicted from the micro-scale simulation and DeePN$^2$. With the same bond potential and strain rate, the chain-shaped molecule suspension yields larger elongation stress. \revision{The lines with rough oscillations} denote the micro-scale simulation results; the solid lines with symbols denote the DeePN$^2$ predictions.}
\label{fig:Elongation_flow_development}
\end{figure}

Finally, we consider the Taylor-Green vortex flow \cite{Taylor_Green_1934, Thomases_Shelley_POF_2007} in a  $100\times 100\times 160$ domain (in reduced unit) of the micro-scale simulation. The external 
force $\mb f_{\textrm{ext}} = (f_x, f_y, 0)$ is applied to the domain following
\begin{equation*}
f_x(x,y) =  -2f_0 \sin\left(\frac{2\pi x}{L}\right)\cos\left(\frac{2\pi y}{L}\right),  
\quad
f_y(x,y) =  2f_0 \cos\left(\frac{2\pi x}{L}\right)\sin\left(\frac{2\pi y}{L}\right), 
\end{equation*}
where $L = 100$ and $f_0 = 6\times 10^{-3}$.  Periodic boundary conditions are imposed along all of the three directions. The force field imposes an elongation 
to the flow field along the x-direction and a compression along the y-direction.  
The flow near the center $(L/2, L/2)$ resembles the planar elongation flow. Four vortices appear at $(L/2\pm L/4, L/2\pm L/4)$. 
Figure. \ref{fig:four_mill_vel2D}(a-b) shows the steady-state velocity field. 
Compared with the star-shaped molecule suspension, the velocity field of the chain-shaped molecule suspension 
shows larger deviation from the symmetric structure of the Newtonian flow 
(i.e., $\propto$ $[-\sin\left({2\pi x}/{L}\right)\cos\left({2\pi y}/{L}\right),$  $\cos\left({2\pi x}/{L}\right)\sin\left({2\pi y}/{L}\right)]$)
due to the larger polymer stress across the flow regime.
Furthermore, the two suspensions yield different velocity magnitude, as shown in Fig. \ref{fig:four_mill_vel2D}(c). Fig. \ref{fig:four_mill_vel2D}(d)
shows the velocity development at $(75, 49)$. The velocities of both suspensions achieve a similar maximum value near $t = 30$ and decay along with
the polymer stress development. However, the star-shaped molecule suspension reaches the steady state much earlier with a larger velocity 
than the chain-shaped molecule suspension.

\begin{figure}[htbp]
\centering
\includegraphics[scale=0.25]{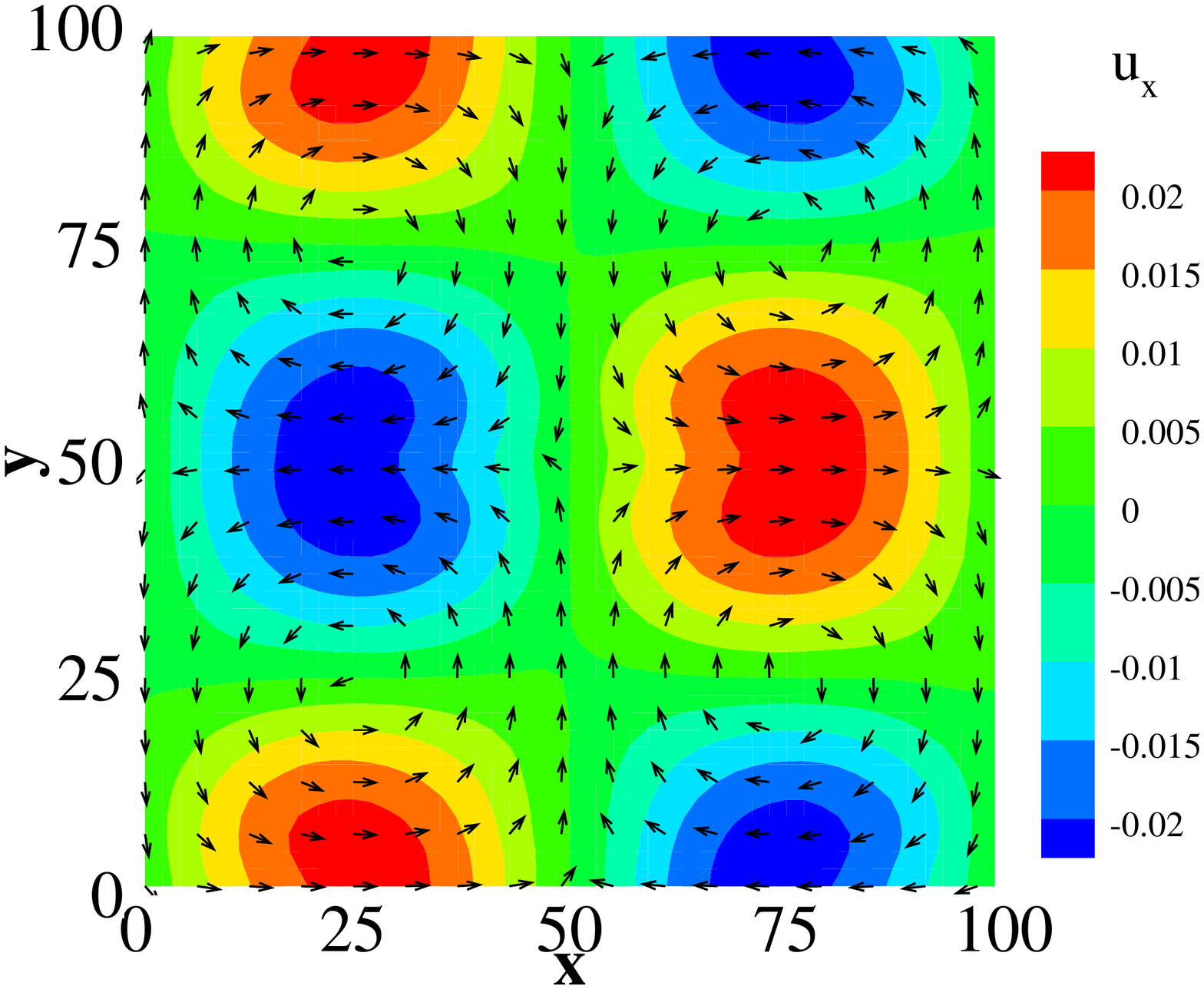}
\includegraphics[scale=0.25]{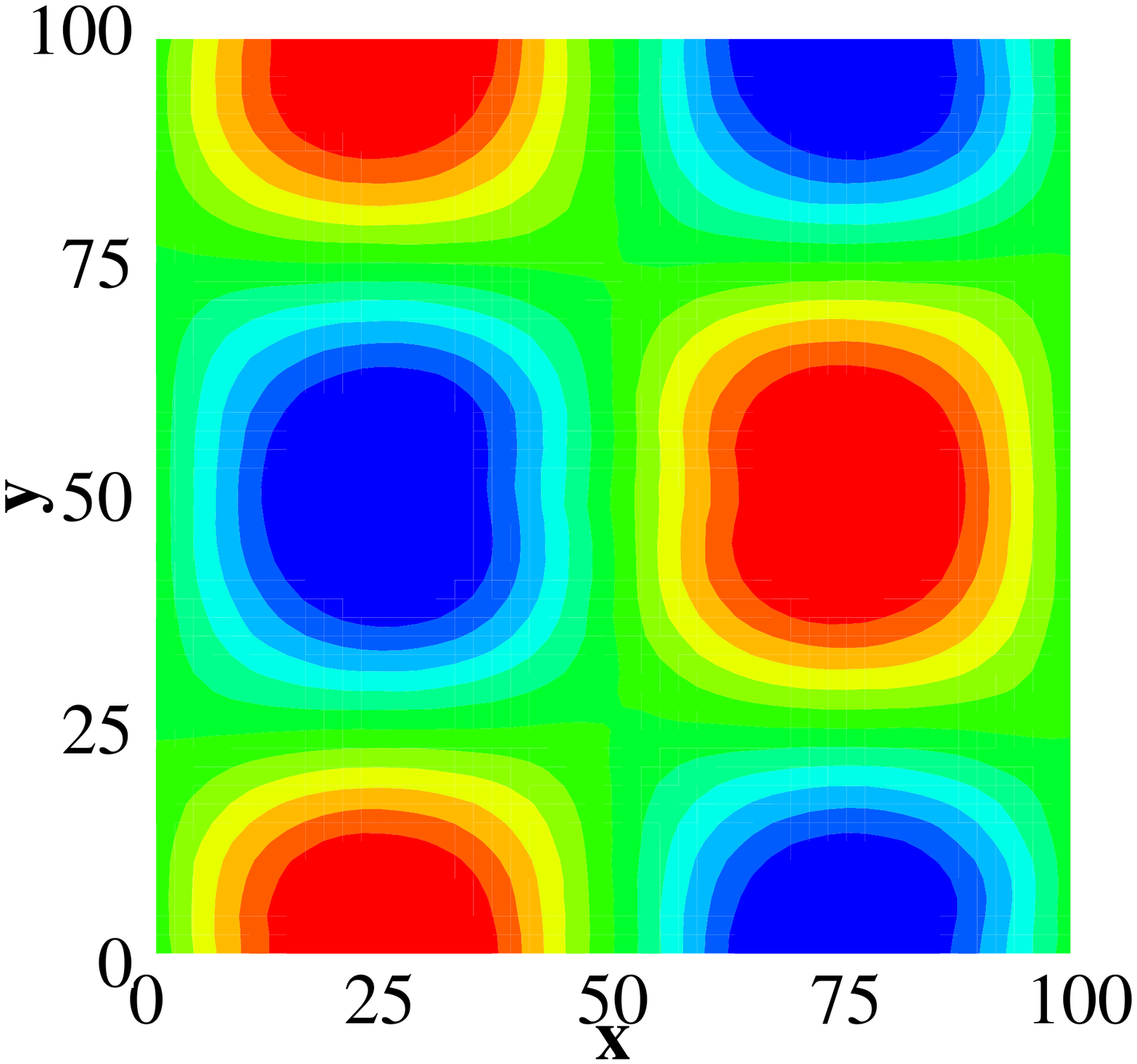}\\
\includegraphics[scale=0.25]{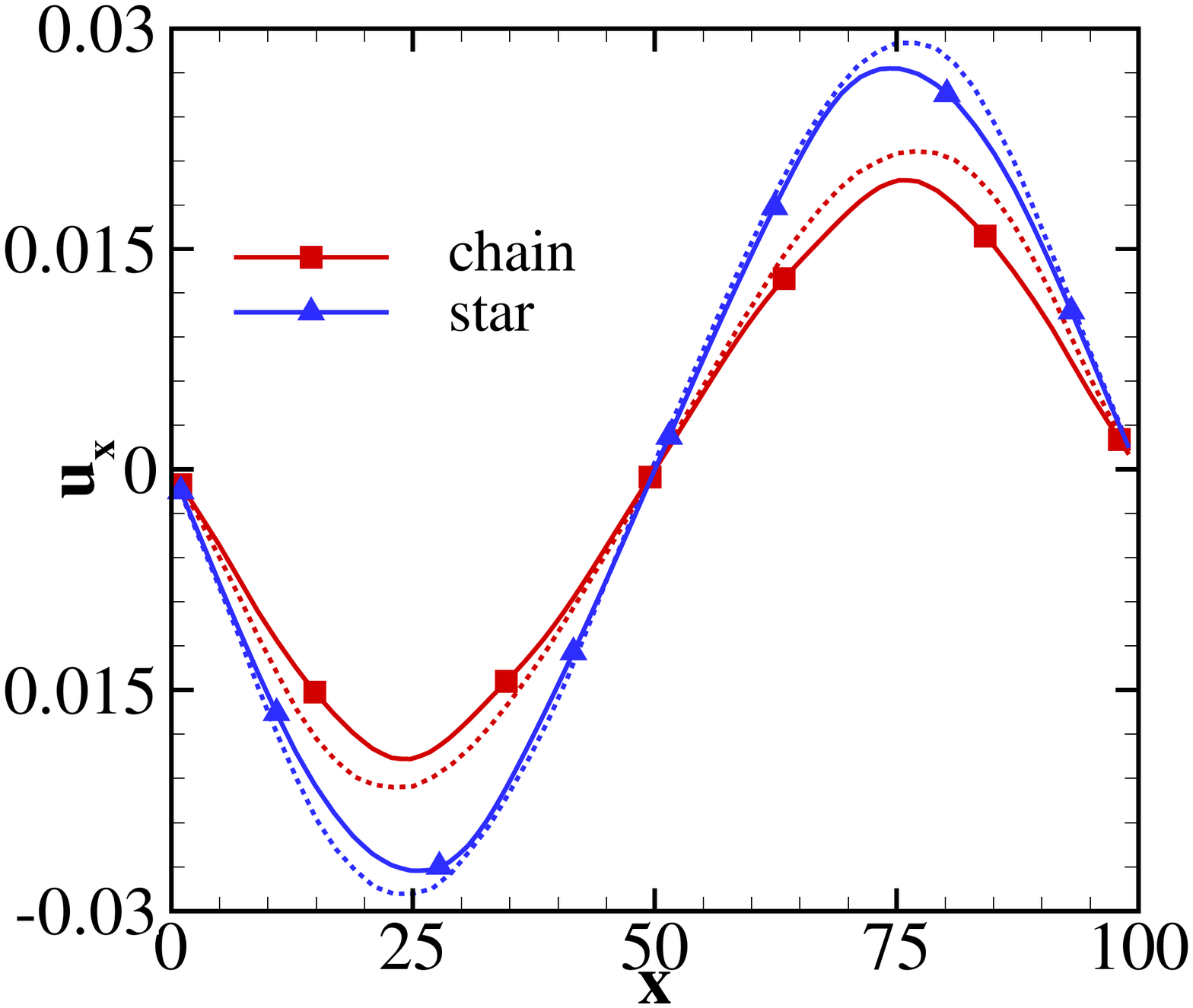}
\includegraphics[scale=0.25]{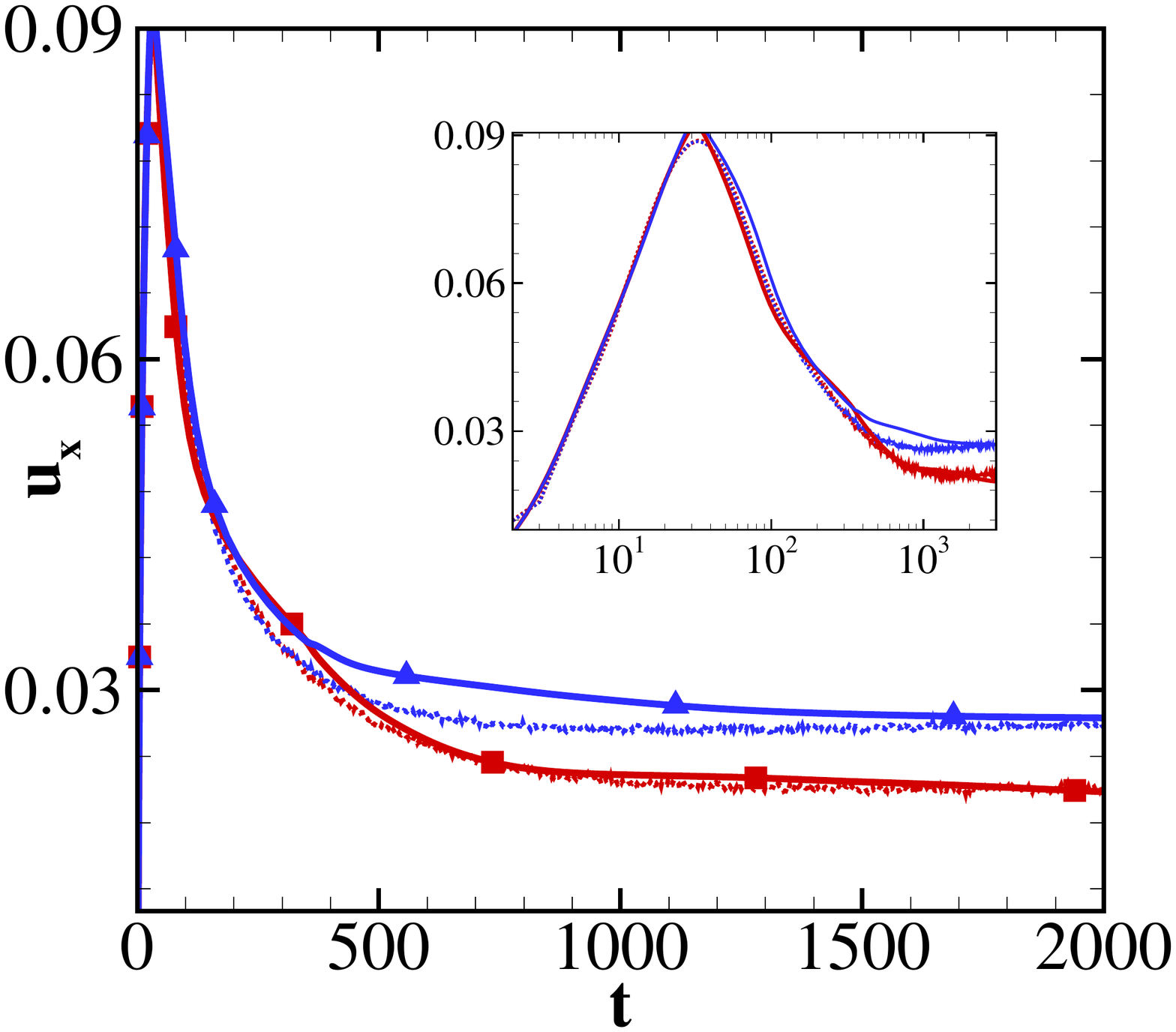}\\
\caption{The velocity field of the Taylor-Green vortex flow of the chain- and star-shaped molecule suspensions predicted from the micro-scale simulations and DeePN$^2$. 
(a-b) The 2D steady-state velocity field of the chain- and star-shaped molecule suspensions from the micro-scale simulations. 
The velocity field of the chain-shaped system yields more pronounced deviations from the symmetric Newtonian flow due to the more pronounced polymer stress 
across the flow regime. (c) The steady-state 1D velocity profile $u_x(x, y=49)$. The solid and dashed lines represent the predictions from the micro-scale simulations and the DeePN$^2$ model, respectively. (d) The time history of $u_x(x=75, y=49)$.}
\label{fig:four_mill_vel2D}
\end{figure}

Fig. \ref{fig:four_mill_stress2D} (a-b) shows the steady-state stress field for the two suspensions. 
We see that  the chain-shaped molecule suspension exhibits larger polymer stress variation
along the elongation and contraction directions, reflected in the larger loop area in Fig. \ref{fig:four_mill_stress2D}(b). Such difference is also consistent with the more
pronounced asymmetric velocity field shown in Fig. \ref{fig:four_mill_vel2D}(a-b). 
In addition, we also examine the transient states where the flow undergoes intricate and heterogeneous process. Fig. \ref{fig:four_mill_stress2D}(c) shows the stress 
development at point $(49, 35)$, where ${\tau _{\textrm{p}}}_{xx}$ and ${\tau _{\textrm{p}}}_{yy}$ cross over during the evolution. During the initial stage, ${\tau _{\textrm{p}}}_{yy}$ 
increases along with the flow development towards to the stagnation point. At $t > 150$, ${\tau _{\textrm{p}}}_{yy}$ decreases due to the compression along 
the y-direction. Meanwhile, ${\tau _{\textrm{p}}}_{xx}$ increases and achieves a steady state slightly larger than ${\tau _{\textrm{p}}}_{yy}$ for the star-shaped solution. On 
the other hand, the chain-shaped solution ends up with a significantly larger value of ${\tau _{\textrm{p}}}_{xx}$ due to the larger molecule flexibility and further
extension along the x-direction. The different viscoelastic responses are also reflected in the stress development at point $(49, 49)$.
As shown in Fig. \ref{fig:four_mill_stress2D}(d), the chain-shaped solution exhibits longer evolution of ${\tau _{\textrm{p}}}_{xx}$ and larger steady value than the star-shaped solution. 
DeePN$^2$ successfully captures such  micro-structure-induced rheological differences and shows good agreement with
the micro-scale simulation results.

\begin{figure}[htbp]
\centering
\includegraphics[scale=0.25]{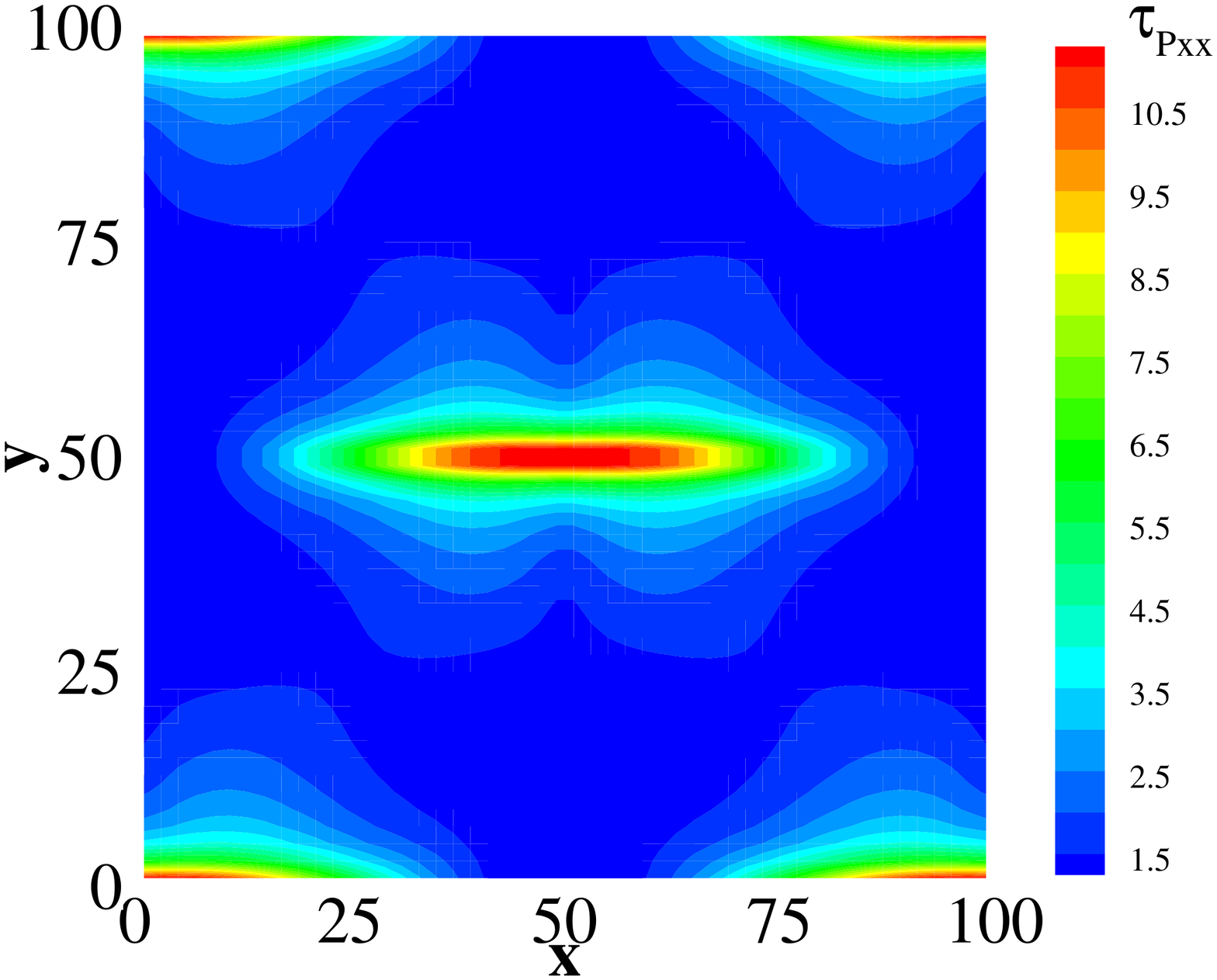}
\includegraphics[scale=0.25]{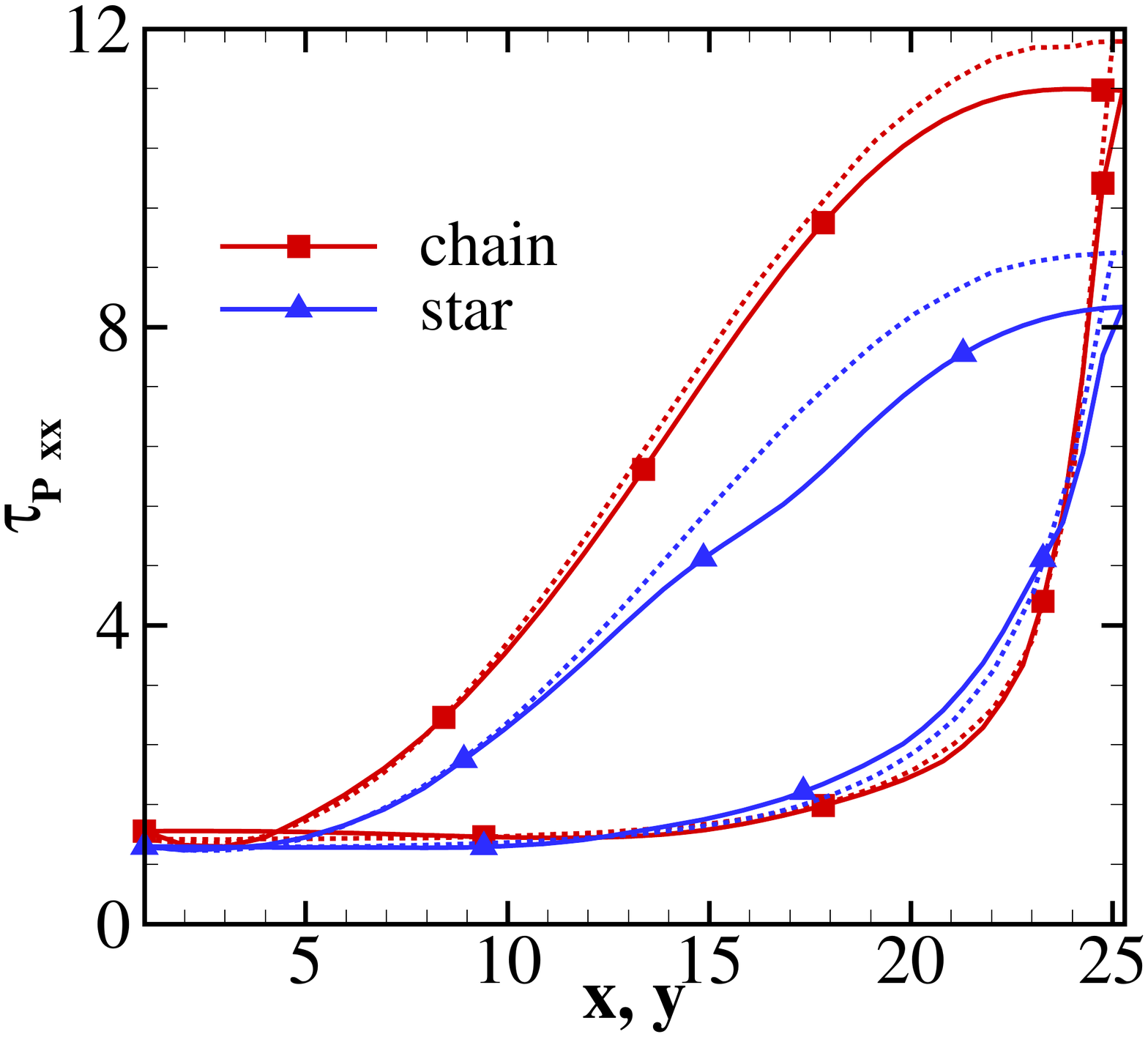}\\
\includegraphics[scale=0.25]{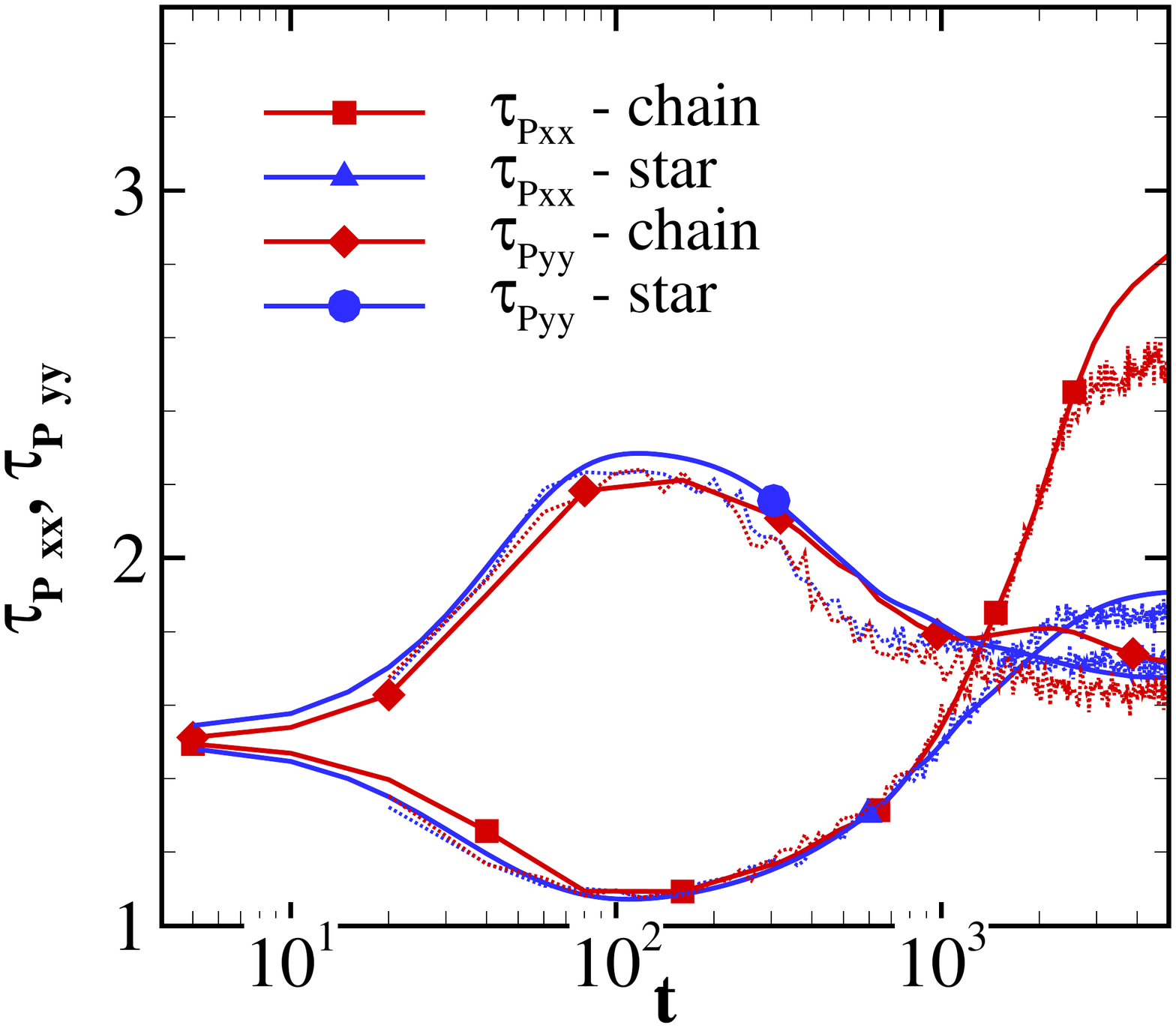}
\includegraphics[scale=0.25]{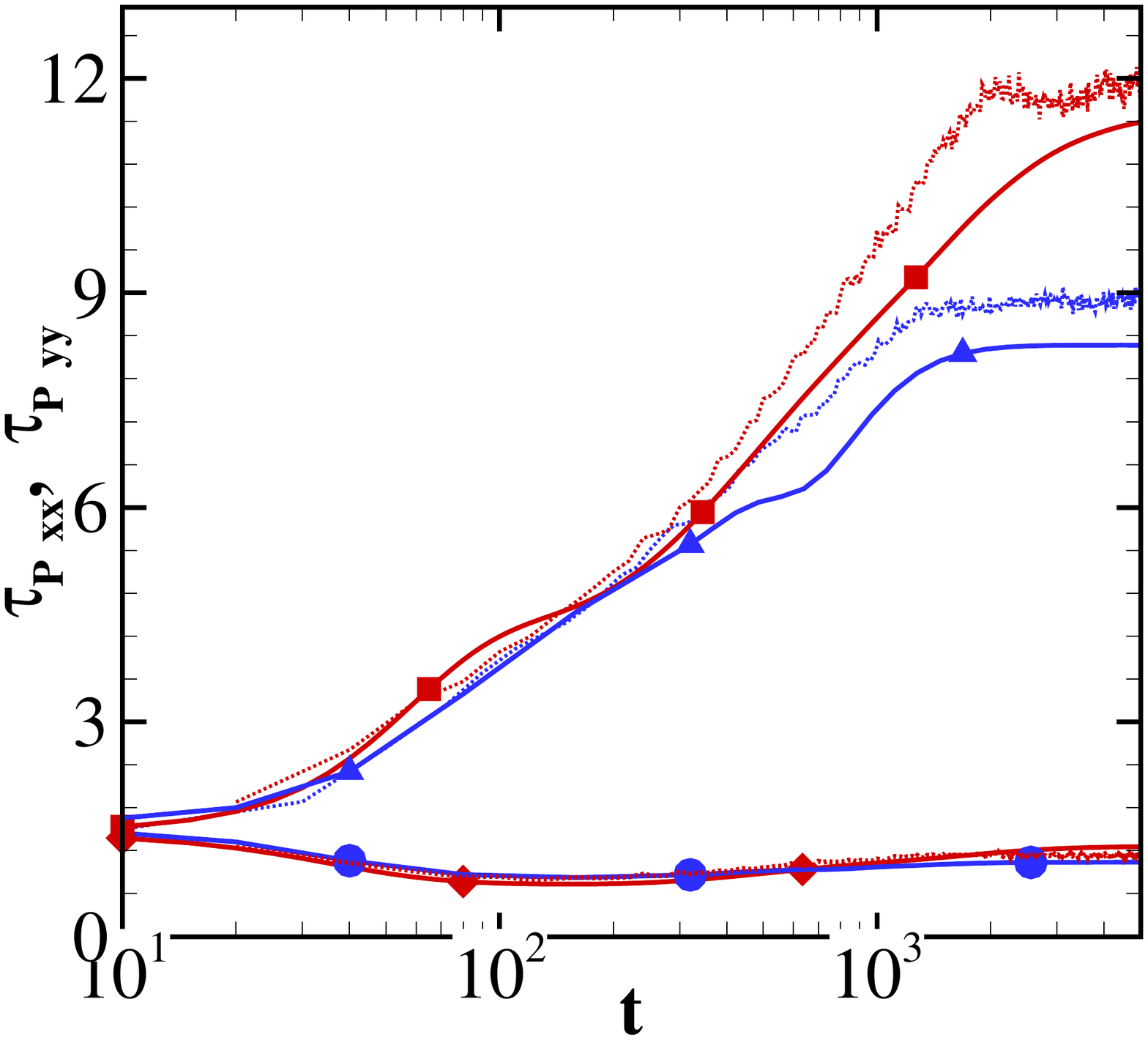}\\
\caption{The stress field of the Taylor-Green vortex flow of the chain- and star-shaped molecule suspensions predicted from the micro-scale simulations and DeePN$^2$. 
(a) The 2D steady-state stress field of the chain-shaped molecule suspension from the micro-scale simulations. (b) The 1D steady-state stress profiles ${\tau _{\textrm{p}}}_{xx}(x, y=49)$ and ${\tau _{\textrm{p}}}_{xx}(x=49, y)$. The chain-shaped molecule suspension yields larger stress variations (i.e., the ``loop area'') along the flow domain.  (c-d) The stress evolution of ${\tau _{\textrm{p}}}_{xx}(t)$ and ${\tau _{\textrm{p}}}_{yy}(t)$ at the points $(49, 35)$ and $(49, 49)$, respectively. The dashed and the solid lines denote the micro-scale simulations and the DeePN$^2$ predictions, respectively.}
\label{fig:four_mill_stress2D}
\end{figure}

\section{Discussion}
We have developed a general machine-learning based model,  DeePN$^2$, for describing the non-Newtonian hydrodynamics for polymer solutions
with arbitrary molecular structure and interaction. The constructed model retains a clear physical interpretation and faithfully encodes the micro-scale  structural information  
into the macro-scale hydrodynamics, where conventional models based on empirical closures generally show limitations. In particular, for the chain- and star-shaped molecule
suspensions with the same bead number and bond interaction,  DeePN$^2$ successfully captures the different viscoelastic responses arising from the different
molecular structural symmetry (i.e., the effective rigidity) in the configuration space without additional human intervention. 
Unlike the direct evaluation or moment-closure representations of the configurational PDF, the present DeePN$^2$ model directly learns a set of 
micro-to-macro mappings to probe the optimal approximations of the constitutive dynamics in terms of the macro-scale features, and thereby circumventing the numerical challenges due to the high-dimensionality of the polymer configuration space. This multi-scaled nature enables us to learn the constitutive dynamics of the macro-scale features directly 
from the kinetic equations of their micro-scale counterparts using only discrete 
rather than the time-derivative samples commonly used in the machine learning-based models of complex dynamic problems.

One thing we have not investigated systematically is the generation of training samples.
For DeePN$^2$ to be truly reliable, the training samples should be representative enough for all the practical situations that 
one might encounter.
However, due to the cost associated with generating such training samples, we would also like the training set to be as small as possible.
This calls for an adaptive procedure for generating the training sample, such as the concurrent learning procedure discussed 
in \cite{EHanZhang-PhysicsToday}.
The present DeePN$^2$ models are trained with samples collected from  homogeneous shear flow.
Even though the numerical predictions 
show good agreement with micro-scale simulations for a variety of flows, one should not expect this to be generally the case.
Further work on sampling is needed to make sure that one can produce truly reliable DeePN$^2$ models.
Furthermore, instead of the general form \eqref{eq:encoder_model}, a specific design of 
the encoders $\mb b(\cdot)$ accounting for the molecule symmetry and rigidity may facilitate the extraction of the macro-scale 
features $\mb c$. 
In addition, more accurate micro-scale kinetic models accounting for the heterogeneous hydrodynamic interactions \cite{Zimm_model_JCP_1956} and non-Markovianity \cite{Lei_Li_PNAS_2016, Lei_Li_JCP_2021}
can be used to construct the macro-scale constitutive dynamics.
Finally, the adaptive choice of the number of features and the enhanced sampling of the discrete micro-scale 
configurations may further improve the performance of the DeePN$^2$ model. We leave these issues for  future work.


\appendix

\section*{Appendices}

\section{Rotational frame-indifference of the constitutive dynamics for the multi-bead encoder function}
\label{sec:proof_rotation_dynamics}

We consider a polymer molecule consisting of $N$ particles. Let $\mb r = \left[\mb r_1; \mb r_2; \cdots; \mb r_{N-1}\right]$ 
denote the polymer configuration, so that there exists an invertible linear transformation between $\left[\mb r; \sum _{i=1}^N \mb q_i / N \right]$ and $ \left[\mb q_1; \mb q_2; \cdots; \mb q_{N}\right]$, where $\mb q_i$ is the position 
of the $i\mhyphen$th particle. In fact, there are multiple choices for $\mb r$, including the one we have applied in Eq. \eqref{eq:r_r_ast}, where $\mb r$ consists of $(N-1)$ edges of a spanning tree in the bead-bond structure.

We consider a second-order tensor taking the general form
\begin{equation}
\mb b = \mb f^{(1)}(\mb r) \mb f^{(2)}(\mb r)^T, \quad \mb f^{(1)}(\mb r) = \sum_{j=1}^{N-1} g^{(1)}_j(\mb r^{\ast})\mb r_j,
\quad \mb f^{(2)}(\mb r) = \sum_{j=1}^{N-1} g^{(2)}_j(\mb r^{\ast})\mb r_j,
\label{eq:general_B_tensor}
\end{equation}
where $\mb r^{\ast}$ is a  translational-rotational-invariant vector and $g^{(1)}$ and $g^{(2)}$ are two scalar functions. We note that the encoder in the form of Eq. \eqref{eq:general_B_tensor} is more general than Eq. \eqref{eq:encoder_model}.

In this appendix and the next, we consider two frames: frame 1 is static inertial, and frame 2 is rotating with respect to frame 1 with an time dependent orthogonal transformation $\mb Q (t)$.
Let $\tilde{\mb {x}}, \tilde{\mb {v}}, \tilde{\mb {b}}$ and $\mb x, \mb v, \mb b$ denote the positions, velocities, and second-order tensors in frame 1 and 2 respectively. They have the following relations:
\begin{equation}
  \tilde{\mb {x}} = \mb Q \mb x, 
  \quad 
  \tilde{\mb {v}} = \mb Q \mb v + \dot{\mb Q} \mb x, 
  \quad
  \tilde{\mb {b}} = \mb Q \mb b \mb Q ^T.
\end{equation}
The material derivatives in both frames are
\begin{equation}
  \left. \frac{\textup d}{\textup d t} \right| _{\textup{frame 1}} := \frac{\partial}{\partial t} + \tilde{\mb v} \cdot \nabla _{\tilde{\mb x}},
  \quad
  \left. \frac{\textup d}{\textup d t} \right| _{\textup{frame 2}} := \frac{\partial}{\partial t} + {\mb v} \cdot \nabla _{{\mb x}}.
\end{equation}

\begin{proposition}
With $\mb b$ defined by Eq. \eqref{eq:general_B_tensor},  we have
\begin{equation}
\begin{split}
\frac{\textup{d}}{\textup {d}t} \mb c - \bm\kappa:\left\langle \sum_{j=1}^{N-1}
\mb r_j\otimes\nabla_{\mb r_j}\otimes \mb b \right\rangle  &= 
\frac{k_BT}{\gamma}\left\langle \sum_{j,k=1}^{N-1}  A_{jk} \nabla_{\mb r_j} \cdot 
\nabla_{\mb r_k} \mb b \right\rangle\\
&- 
\frac{1}{\gamma} \left\langle \sum_{j=1}^{N-1} \sum_{k=1}^{N_b}  A_{jk} \nabla_{\mb r_k} V_{\textrm p}(\mb r) \cdot \nabla_{\mb r_j} 
\mb b \right\rangle,
\end{split}
\label{eq:FK_B_evoluation_triplet}
\end{equation}
obeys  rotational symmetry.
\label{prop:rotation_dynamics}
\end{proposition}

\begin{proof}
Let us choose the 
vector $\mb r^{\ast} = \left[\vert \mb r_1\vert, \vert \mb r_2\vert, \vert \mb r_{12}\vert, 
\vert \mb r_3\vert, \vert \mb r_{13}\vert, \vert \mb r_{23}\vert, \cdots, \vert \mb r_{N-2,N-1}\vert\right]$.
Denote by $r^{\ast}_i$ the $i\mhyphen$th element of $\mb r^{\ast}$ and  $\mb r^{\ast}_i$ the corresponding
the 3-dimensional vector, i.e., $r^{\ast}_6 = \vert \mb r_{23}\vert$ and $\mb r^{\ast}_6 = \mb r_{23}$. 
Following Eq. \eqref{eq:general_B_tensor}, $\mb b$ consists of
\begin{equation}
\mb b = \sum_{j,k=1}^{N-1} \mb b_{jk}, \quad \mb b_{jk} = g(\mb r^{\ast}) \mb r_j \mb r_k^T,
\end{equation}
where $g(\mb r^{\ast})$ denotes $g^{(1)}_j(\mb r^{\ast}) g^{(2)}_k(\mb r^{\ast})$ for simplicity.
With this general form, we have
\begin{equation}
\frac{\textrm d}{\textrm dt} \left\langle \tilde{\mb b}_{jk}\right\rangle \big\vert_{\textrm{frame}~1} = 
\dot{\mb Q} \left\langle \mb b_{jk} \right\rangle \mb Q^T
+ \mb Q \left\langle \mb b_{jk} \right\rangle \dot{\mb Q}^T
+ \mb Q \frac{\textrm d}{\textrm dt} \left\langle \mb b_{jk} \right\rangle \big\vert_{\textrm{frame}~2}
\mb Q^T.
\label{eq:FK_B_1_angle}
\end{equation}

Moreover, we note that 
\begin{equation}
\begin{split}
&\tilde{\bm\kappa}:\left(\sum_{i=1}^{N-1}
\tilde{\mb r_i}\otimes\nabla_{\tilde{\mb r_i}}\otimes 
\tilde{\mb b}_{jk}\right)  \\
& = \sum_{i=1}^{N-1} \left[ \left(\mb Q \bm\kappa \mb Q^T + \dot{\mb Q}\mb{Q}^T\right)
\cdot \mb Q \mb r_j \right] \cdot \mb Q \otimes \nabla_{\mb r_i} \otimes \left(\mb Q \mb b_{jk}\mb Q^T\right) \\
&= \sum_{i=1}^{N-1} \left(\bm\kappa\cdot\mb r_i\right)\cdot  \nabla_{\mb r_i} \left(\mb Q \mb b_{jk}\mb Q^T\right) +
(\mb Q^T \dot{\mb Q} \mb r_i)\cdot \nabla_{\mb r_i} \left(\mb Q \mb b_{jk}\mb Q^T\right) \\
&= \sum_{i=1}^{N-1} \mb Q (\bm\kappa\cdot\mb r_i) \cdot \nabla_{\mb r_i} \mb b_{jk} \mb Q^T + \mb Q\left(\mb Q^T\dot{\mb Q} \mb b_{jk} 
+ \mb b_{jk} \dot{\mb Q}^T \mb {Q}\right)\mb Q^T  \\
&~~~~ + \mb Q \left(\sum_{i=1}^{N-1} {\mb r_{i}}^T ( \dot{\mb Q}^T{\mb Q}) 
\nabla _{\mb r_i} g (\mb r^{\ast}) \right)\mb r_{j}r_k^T \mb Q^T \\
&= \sum_{i=1}^{N-1} \mb Q (\bm\kappa\cdot\mb r_i) \cdot \nabla_{\mb r_i} \mb b_{jk}\mb Q^T + 
\dot{\mb Q} \mb b_{jk} \mb Q^T
+ \mb Q \mb b_{jk} \dot{\mb Q}^T,
\end{split}
\label{eq:FK_B_2_angle}
\end{equation}
where we have used 
${\mb r_{i}}^T (\dot{\mb Q}^T{\mb Q}) \mb r_{i} \equiv 0$
since $\dot{\mb Q}^T{\mb Q}$ is anti-symmetric. Eq. \eqref{eq:FK_B_1_angle} and Eq. \eqref{eq:FK_B_2_angle} shows that 
the combination of the two terms on the left-hand-side of Eq. \eqref{eq:FK_B_evoluation_triplet} rigorously preserve 
the rotational symmetry, i.e.,
\begin{equation}
\quad \left. \left(\frac{\textrm d}{\textrm dt} \left\langle \tilde{\mb b}\right\rangle 
- \tilde{\bm\kappa}:\sum_{i=1}^{N-1}\left\langle \tilde{\mb r_i} \otimes \nabla_{\tilde{\mb r_i}}
\otimes
\tilde{\mb b}\right\rangle\right) \right\vert_{\textrm{frame}~1} 
\equiv 
\mb Q \left . \left(
\frac{\textrm d}{\textrm dt} \left\langle \mb b \right\rangle 
- \bm\kappa:\sum_{i=1}^{N-1}\left\langle \mb r_i \otimes \nabla_{\mb r_i} \otimes
\mb b \right\rangle \right) \right\vert_{\textrm{frame}~ 2}\mb Q^T. \nonumber
\end{equation}
It is straightforward to prove rotational symmetry 
for the other terms in Eq \eqref{eq:FK_B_evoluation_triplet}.
\end{proof}

\section{Symmetry-preserving neural network representation of the objective tensor derivatives}

\begin{proposition}
The following ansatz of $\left\langle\sum_{i=1}^{N-1}\mb r_i \otimes \nabla_{\mb r_i} \otimes 
\mb b\right\rangle$ 
ensures that the dynamic of evolution of $\mb c$ retains rotational invariance.
\begin{equation}
\begin{split}
\sum_{i=1}^{N-1}\left\langle\mb r_i \otimes \nabla_{\mb r_i}\otimes \mb b\right\rangle 
&= \sum_{j,k=1}^{N-1} \left\langle g^{(1)}_j(\mb r^{\ast}) g^{(2)}_k(\mb r^{\ast})(\mb r_j\otimes\nabla_{\mb r_j} + \mb r_k\otimes\nabla_{\mb r_k}) \otimes \mb r_j \mb r_k^T \right\rangle \\ 
&+ \sum_{k=1}^9 \mb E_{1}^{(k)} ({\mb c}) \otimes \mb E_{2}^{(k)} ({\mb c}),
\end{split}
\end{equation}
where 
$\mb c = ( \mb c_1, \cdots, \mb c_n )$, $\tilde{\mb c} = ( \tilde{\mb c}_1, \cdots, \tilde{\mb c}_n )$, and
$\mb E_1$ and $\mb E_2$ satisfy 
\begin{equation}
\tilde{\mb E}_1 := \mb E_1(\tilde{\mb c}) = \mb Q \mb E_1({\mb c}) \mb Q^T, \quad
\tilde{\mb E}_2 := \mb E_2(\tilde{\mb c}) = \mb Q \mb E_2({\mb c}) \mb Q^T.
\label{eq:G_F_symmetry}
\end{equation}
\label{prop:rotation_encoder}
\end{proposition}

\begin{proof}
Without loss of generality, we represent the fourth order tensor by the following two bases
\begin{align}
&\mb F_1(\mb c)\otimes \mb F_2(\mb c)\otimes \mb F_3(\mb c)  
+ \mb F_3(\mb c)\otimes \left(\mb F_2(\mb c)\otimes \mb F_1(\mb c)\right)^{T_{\{2,3\}}}, 
&& \mb F_1({\mb c}), \mb F_3({\mb c}) \in \mathbb{R}^3, \mb F_2({\mb c}) \in \mathbb{R}^{3\times 3}, 
\nonumber \\
&\mb E_1(\mb c) \otimes \mb E_2(\mb c), 
&& \mb E_1({\mb c}), \mb E_2({\mb c}) \in \mathbb{R}^{3\times 3},
\end{align}
where the super-script $T_{\{2,3\}}$ represents the transpose between the 2nd and 3rd indices; also
$\mb F_1$, $\mb F_2$, $\mb F_3$, $\mb E_1$ and $\mb E_2$ satisfy the symmetry 
conditions
\begin{equation}
\begin{split}
&\mb F_1(\tilde{\mb c}) = \mb Q \mb F_1({\mb c}), \quad    
\mb F_3(\tilde{\mb c}) = \mb Q \mb F_3({\mb c}), \\
&\mb E_1(\tilde{\mb c}) = \mb Q \mb E_1({\mb c}) \mb Q^T, \quad
\mb E_2(\tilde{\mb c}) = \mb Q \mb E_2({\mb c}) \mb Q^T, \quad 
\mb F_2(\tilde{\mb c}) = \mb Q \mb F_2({\mb c}) \mb Q^T. 
\end{split}
\end{equation}

For the term $\mb E_1(\mb c) \otimes \mb E_2(\mb c)$, we have 
\begin{equation}
\bm\kappa:\mb E_1(\mb c)\otimes \mb E_2(\mb c)
= \Tr(\bm\kappa \mb E_1({\mb c})) \mb E_2({\mb c})
\end{equation}
and
\begin{equation}
\begin{split}
\tilde{\bm\kappa}:\tilde{\mb E}_1\otimes \tilde{\mb E}_2
\big\vert_{\textrm{frame} ~1}
&= \left(\mb Q\bm\kappa \mb Q^T + \dot{\mb Q}\mb Q^T \right):\left(\mb Q\mb E_1({\mb c})\mb Q^T 
    \otimes \tilde{\mb E}_2\right)\\
&= \Tr(\bm\kappa \mb E_1({\mb c}))\tilde{\mb E}_2 + \Tr(\dot{\mb Q}\mb Q^T \mb Q\mb E_1({\mb c})\mb Q^T) 
\tilde{\mb E}_2\\
&= \Tr(\bm\kappa \mb E_1({\mb c}))\tilde{\mb E}_2 \\
&\equiv 
\mb Q\left(
\bm\kappa:\mb E_1(\mb c)\otimes \mb E_2(\mb c)\big\vert_{\textrm{frame} ~2}\right)\mb Q^T,
\end{split}
\label{eq:G_1_2_rot_symmetry}
\end{equation}
where we have used $\Tr(\dot{\mb Q}\mb Q^T) \equiv 0$. 

For the term $\mb F_1(\mb c)\otimes \mb F_2(\mb c)\otimes \mb F_3(\mb c)  
+ \mb F_3(\mb c)\otimes \left(\mb F_2(\mb c)\otimes \mb F_1(\mb c)\right)^{T_{\{2,3\}}}$,   
we have
\begin{equation}
\bm\kappa:\mb F_1(\mb c)\otimes \mb F_2(\mb c)\otimes \mb F_3(\mb c) 
  = \mb F_2({\mb c}) ^T \bm\kappa \mb F_1 ({\mb c})\mb F_3({\mb c})^T
\end{equation}
and 
\begin{equation}
\tilde{\bm\kappa}:\tilde{\mb F}_1\otimes \tilde{\mb F}_2\otimes \tilde{\mb F}_3
= \mb Q\mb F_2({\mb c})^T \bm\kappa \mb F_1({\mb c}) \mb F_3({\mb c})^T \mb Q^T + \mb Q\mb F_2({\mb c})^T \mb Q^T \dot{\mb Q}
\mb F_1({\mb c})\mb F_3({\mb c})^T \mb Q^T.
\end{equation}
On the other hand, we note that 
\begin{equation}
\frac{\diff \tilde{\mb b} }{\diff t} \big\vert_{\textrm{frame} ~1} = 
\dot{\mb Q} \mb b \mb Q^T + \mb Q \mb b \dot{\mb Q}^T + \mb Q \frac{\diff \mb b}{\diff t}
\big\vert_{\textrm{frame}~2} \mb Q^T.
\end{equation}
To ensure the rotational symmetry of $\frac{\mathcal{D}\mb b}{\mathcal{D} t}$, we have 
\begin{equation}
\mb F_2 \equiv \mb I, \quad \sum_i \mb F_1^{(i)} \otimes \mb I \otimes \mb F_3^{(i)} = \sum_{j,k=1}^{N-1}\left\langle 
g^{(1)}_j(\mb r^{\ast}) g^{(2)}_k(\mb r^{\ast}) \mb r_j \otimes \mb I \otimes\mb r_k \right\rangle.
\label{eq:choice_F_1_2_3}  
\end{equation}
Hence, we have
\begin{equation}
\begin{split}
&\frac{\textrm d}{\textrm d t}\tilde{\mb c} - \tilde{\bm\kappa}:
\left.\left( \sum_i \tilde{\mb F}_1^{(i)}\otimes \tilde{\mb F}_2^{(i)} \otimes \tilde{\mb F}_3^{(i)}
+ \tilde{\mb F}_3^{(i)}\otimes \left(\tilde{\mb F}_2^{(i)} \otimes 
 \tilde{\mb F}_1^{(i)}\right)^{T_{\{2,3\}}}\right)
\right\vert_{\textrm{frame} ~1}\\
&\equiv  
\mb Q \left.\left(
\frac{\textrm d}{\textrm d t}\mb c -  \bm\kappa:
\left( \sum_i \mb F_1^{(i)}\otimes \mb F_2^{(i)} \otimes \mb F_3^{(i)}
+ \mb F_3^{(i)}\otimes \left(\mb F_2^{(i)} \otimes \mb F_1^{(i)}\right)^{T_{\{2,3\}}}\right)
\right)\right\vert_{\textrm{frame}~2} \mb Q^T.
\label{eq:F_1_2_3_rot_symmetry}
\end{split}
\end{equation}
Furthermore, using Eq. \eqref{eq:choice_F_1_2_3}, we obtain
\begin{equation}
\begin{split}
&\sum_i \mb F_1^{(i)}\otimes \mb F_2^{(i)} \otimes \mb F_3^{(i)}
+ \mb F_3^{(i)}\otimes \left(\mb F_2^{(i)} \otimes \mb F_1^{(i)}\right)^{T_{\{2,3\}}} \\
&= \sum_{j,k=1}^{N-1} \left\langle g^{(1)}_j(\mb r^{\ast}) g^{(2)}_k(\mb r^{\ast}) 
(\mb r_j\otimes\nabla_{\mb r_j} + \mb r_k\otimes\nabla_{\mb r_k}) 
\otimes
\mb r_j {\mb r_k}^T \right\rangle.
\label{eq:F_1_2_3_sum} 
\end{split}
\end{equation}
Accordingly, the remaining part of $\sum_{i=1}^{N-1} \left\langle \mb r_i \otimes \nabla_{\mb r_i}\otimes \mb b\right\rangle$ is 
expanded by 
\begin{equation}
\left\langle\sum_{i=1}^{N-1}\mb r_i\otimes\nabla_{\mb r_i} \sum_{j,k=1}^{N-1} g^{(1)}_j(\mb r^{\ast}) g^{(2)}_k(\mb r^{\ast}) 
\otimes \mb r_j \mb r_k^T \right\rangle
= 
\sum_{i=1}^9 
\mb E_1^{(i)}(\mb c) \otimes \mb E_2^{(i)}(\mb c).
\label{eq:G_1_2_sum}
\end{equation}

Combining Eq. \eqref{eq:F_1_2_3_rot_symmetry}, \eqref{eq:F_1_2_3_sum} and
\eqref{eq:G_1_2_sum}, we conclude that the decomposition 
\begin{equation}
\begin{split}
\sum_{i=1}^{N-1} \left\langle\mb r_i\otimes\nabla_{\mb r_i}\otimes \mb b\right\rangle 
&= \sum_{j,k=1}^{N-1} \left\langle g^{(1)}_j(\mb r^{\ast}) g^{(2)}_k(\mb r^{\ast})(\mb r_j\otimes\nabla_{\mb r_j} + \mb r_k\otimes\nabla_{\mb r_k})\right. 
\left.\otimes \mb r_j \mb r_k^T \right\rangle \\ 
&+ \sum_{k=1}^9 \mb E_{1}^{(k)} (\mb c) \otimes \mb E_{2}^{(k)} (\mb c)
\end{split}
\end{equation}
ensures the objectivity of the time-derivative of $\mb c$.
\end{proof}

\section{Symmetry-preserving neural network representation of the second-order tensor}
\label{sec:symmetry_preserving_DNN}
In the DeePN$^2$ model, we construct the NN representations of the second-order tensors for the stress $\mb G$, constitutive terms $\mb H_1$, $\mb H_2$, and objective tensor derivative terms $\mb E_1$ and $\mb E_2$ that satisfy the rotational symmetry conditions, i.e.,
\begin{equation}
\mb G(\widetilde{\mb c}_1, \cdots, \widetilde{\mb c}_n) = 
\mb Q \mb G(\mb c_1, \cdots, \mb c_n) \mb Q^T,
\label{eq:G_symmetry}
\end{equation}
where $\widetilde{\mb c}_i = \mb Q \mb c_i \mb Q^T$ and $\mb Q$ is an orthogonal matrix.

To preserve the rotational symmetry condition \eqref{eq:G_symmetry}, we fix the form of encoder $\mb b_1$ and transfer the learning to the eigen-space of $\mb c_1$.
Let us assume that the eigen-decomposition $\mb c_1= \mb U\Lambda \mb U^T$ has distinct eigenvalues, where $\mb U$ is the matrix whose columns are the eigenvectors 
of $\mb c_1$. $\mb U$ is not unique due to the non-uniqueness of the eigenvectors. Without loss of generality, we further assume that the first element of $\mb u_1$ to 
be positive. With the following lemma, we show that the general form of $\mb U$ can be always written as $\mb U^{(j)} := \mb U S^{(j)}$ with $j = 1, \cdots, 4$, where $S^{(j)}$ is given by
\begin{equation}
S^{(1)} = \begin{pmatrix} +1 & & \\
 &+1 & \\
 & &+1
\end{pmatrix}, 
S^{(2)} = 
\begin{pmatrix} +1 & & \\
 &-1 & \\
 & &+1
\end{pmatrix},
S^{(3)} = 
\begin{pmatrix} +1 & & \\
 &+1 & \\
 & &-1
\end{pmatrix}, 
S^{(4)} = 
\begin{pmatrix} +1 & & \\
 &-1 & \\
 & &-1
\end{pmatrix}.
\nonumber
\end{equation}

\begin{lemma}
For a symmetry matrix $M \in \mathbb{R}^{3\times 3}$, let $S_M$ denote the set of matrices with the transformation of $S^{(j)}$, i.e., $S_M := \left\{S^{(1)} M S^{(1)}, \cdots, S^{(4)} M S^{(4)}  \right\}$. For any $M^{(j)} :=     S^{(j)} M S^{(j)}  \in S_M$, $S^{(k)} M^{(j)} S^{(k)} \in S_M$, $1\le j, k \le 4$. Furthermore, $S_M$ can be constructed by $M^{(j)}$, i.e.,  $S_M \equiv \left\{S^{(1)} M^{(j)} S^{(1)}, \cdots, S^{(4)} M^{(j)} S^{(4)}  \right\}$.
\label{lemma:S_M}
\end{lemma}

\begin{proof}
By applying $S^{(j)}$ to $M$, it is easy to see that the diagonal part of $M^{(j)}$ remains the same. Since $M^{(j)}$ is also symmetric, we only need to check the upper-triangular part, taking the four possible operations
\begin{equation*}
 \begin{pmatrix} \ast &+ &+  \\
&\ast&+ \\
&&\ast \\
\end{pmatrix} \quad
 \begin{pmatrix} \ast&- &+  \\
&\ast&- \\
&&\ast \\
\end{pmatrix} \quad
 \begin{pmatrix} \ast&+ &-  \\
&\ast&- \\
&&\ast \\
\end{pmatrix} \quad
 \begin{pmatrix} \ast& -&+  \\
&\ast&- \\
&&\ast \\
\end{pmatrix},
\end{equation*}
where ``$+$'' represents that the  element remains the same and ``$-$'' represents a sign change. We see that number of ``$-$'' operations is either $0$ or $2$. Starting from any of the above choice for  $M^{(j)}$, all of the four operators yields either $0$ or $2$ ``$-$'' operations. Therefore, $S^{(k)} M^{(j)} S^{(k)} \in S_M$. 
%
Furthermore, if the upper triangular part of $M$ has distinct absolute
values, then $\forall M ^{(j)}$, $S ^{k} M ^{j} S ^{k} \neq S ^{k'} M
^{j} S ^{k'}$ with $k \neq k'$, hence $S _{M}$ can be constructed by
$M ^{j}$. Otherwise, if some upper triangular entries of $M$ share the
same absolute value, we can draw the same conclusion accordingly.

\end{proof}

Now we consider the matrix whose columns are  the eigenvectors of  $\tilde{\mb c}_1 = \mb Q \mb c_1 \mb Q^T$, denoted by $\tilde{\mb U}$. We can write $\tilde{\mb U} = \mb Q \mb U S^{(j)}$, where $j \in \{1, 2, 3, 4\}$. Accordingly, the DNN input of $\mb c_i$ takes the form
\begin{equation*}
\begin{split}
\tilde{\mb U}^T \tilde{\mb c}_i  \tilde{\mb U} &= \left(\mb Q \mb U S^{(j)}\right)^T
\mb Q \mb c_i\mb Q^T \left(\mb Q \mb U S^{(j)}\right) 
= S^{(j)} {\mb U}^T \mb c_i \mb U S^{(j)}.
\end{split}
\end{equation*}
Let $M= \mb U^T \mb c_i \mb U$,  by using Lemma \ref{lemma:S_M}, it is easy to see that  $S_{ \mb U^T \mb c_i \mb U}$ can be constructed by taking $j = 1, \cdots, 4$.
\begin{proposition}
Let $\mb U$ be the matrix whose columns are the eigenvectors of $\mb c_1$. Let the DNN input be $\hat{\mb c}_i^{(j)} = S^{(j)} {\mb U}^T \mb c_i \mb U S^{(j)} $. The following form of $\bm \tau _{\textrm{p}}$
\begin{equation}
\mb G(\mb c_1, \cdots, \mb c_n) = \revision{\frac{1}{4}} \sum_{j=1}^4 \mb U^{(j)} \hat{\mb G}(\hat{\mb c}_1^{(j)}, \cdots,  \hat{\mb c}_n^{(j)}) {\mb U^{(j)}}^T, \quad   \mb U^{(j)} = \mb U S^{(j)}.
\label{eq:four_U_evaluation}
\end{equation}
satisfies the rotational symmetry constraint \eqref{eq:G_symmetry}.
\end{proposition}

Finally, to account for the swap of the eigenvectors when the eigenvalues cross over, we consider the $6$ permutations of the three eigenvalues of $\mb c_1$, i.e., 
\begin{equation}
\mb G(\mb c_1, \cdots, \mb c_n) = \revision{\frac{1}{24}} \sum_{k=0}^5 \sum_{j=1}^4 \mb U^{(j,k)} \hat{\mb G}
(\hat{\mb c}_{1}^{(j,k)}, \cdots,  \hat{\mb c}_{n}^{(j,k)}) {\mb U^{(j,k)}}^T,
\end{equation}
where $k$ represents the rank of  permutation (e.g., in lexicographical order)
\revision{and $\mb U^{(j,k)}$ is a variation of $\mb U^{(j)} $ with corresponding column permutation.}

\section{Validation of the rotational-symmetry preserving NN representation}
\label{sec:symmetry_validation}
To validate the performance of the proposed DNN representation, we check the accuracy of the modeling terms given a set of conformation tensors $\mb c_1, \cdots, \mb c_n$ under
different unitary transformations. Fig. \ref{fig:accuracy_eigen_policy} shows the relative error under each transformation. The DNN representation \eqref{eq:four_U_evaluation}
yields the same results under all the transformation. In contrast, the DNN without accounting for the four transformations yields significant error due to
the non-uniqueness of the eigenvectors of $\mb c_1$.

\begin{figure}
 \centering
 \includegraphics[width = 0.45\textwidth]{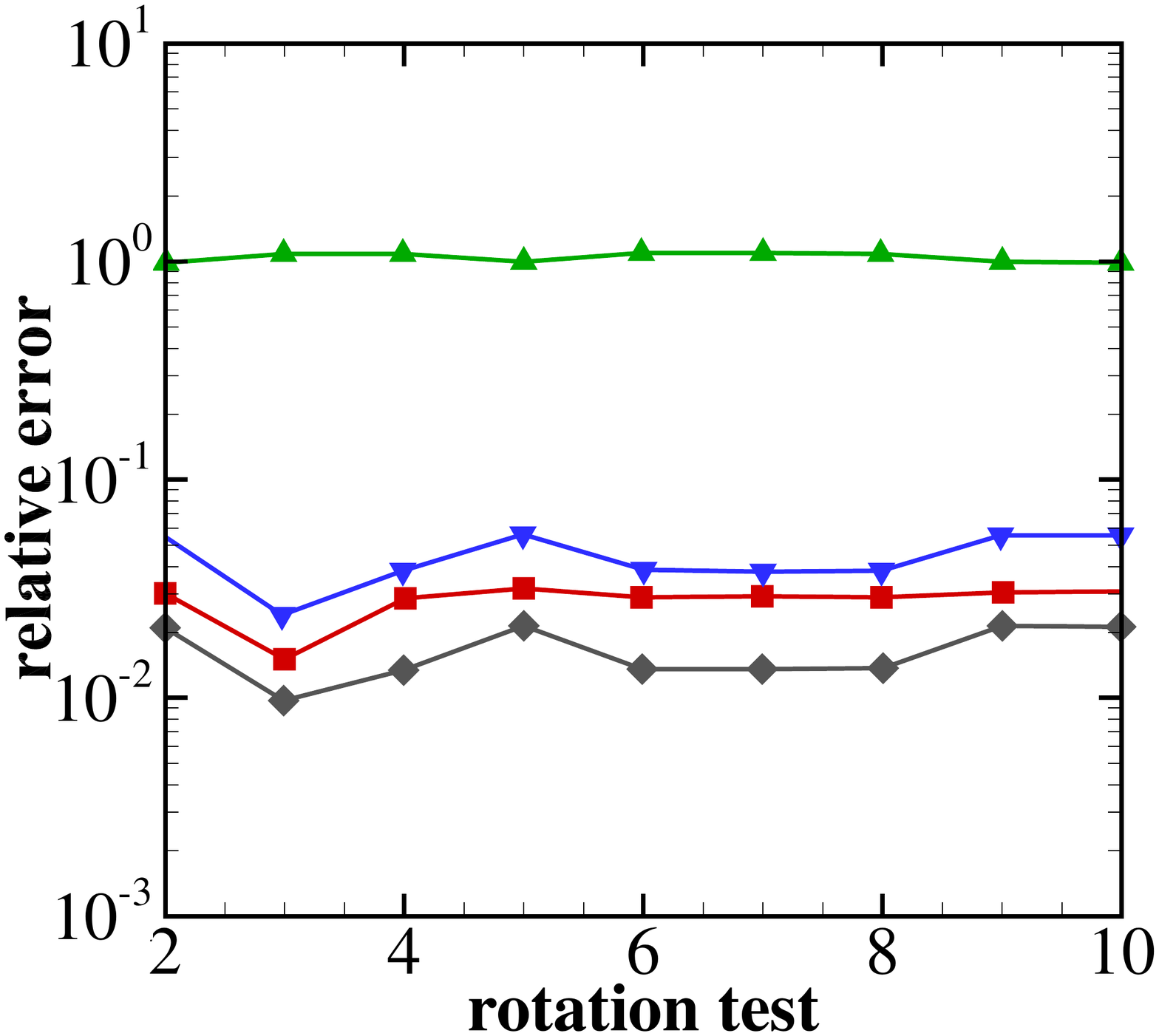}
 \includegraphics[width = 0.45\textwidth]{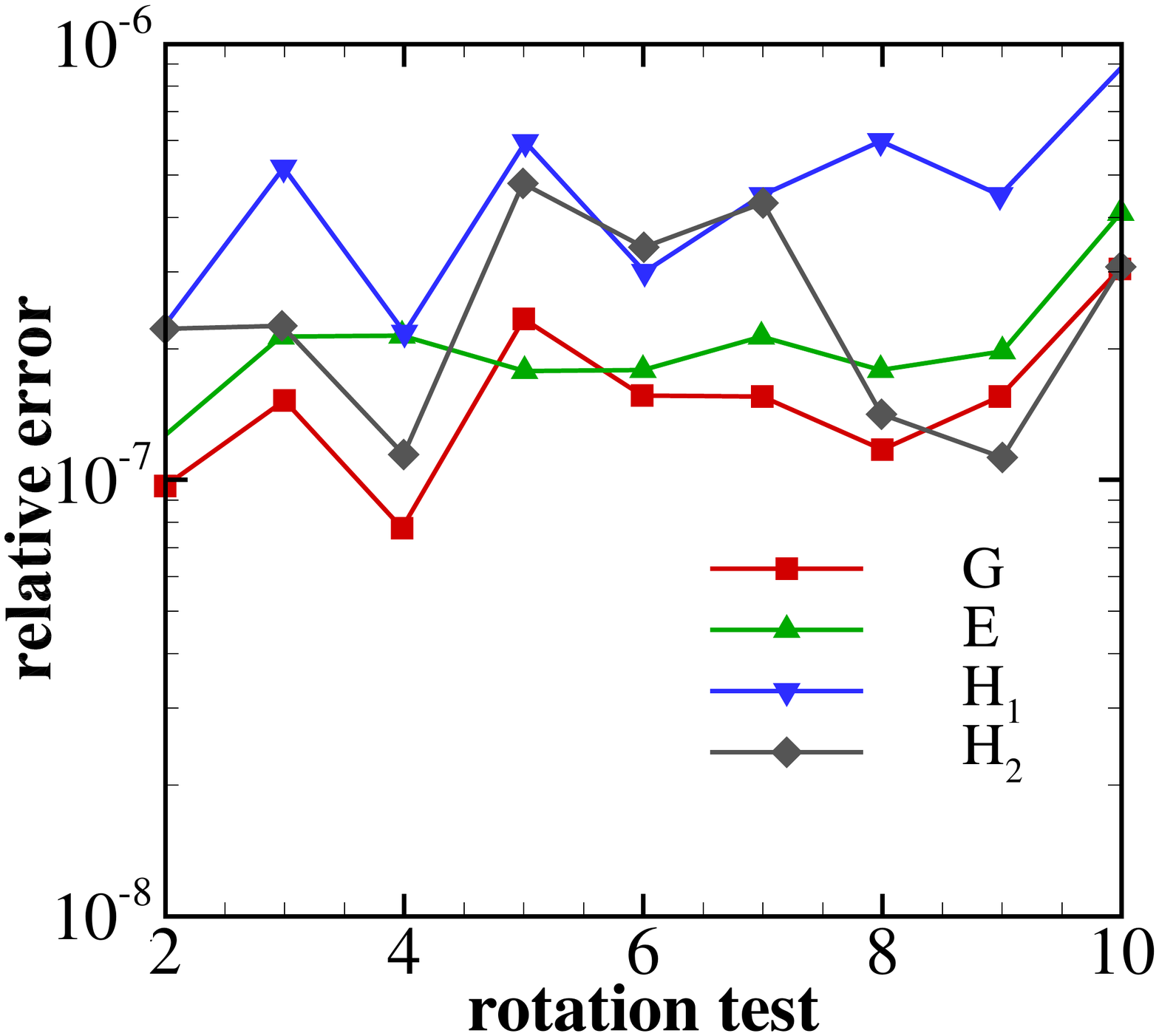}
 \caption{The relative $l_{\infty}$ error of the model prediction under randomly chosen orthogonal transformations  without (left) and with (right) accounting for the four eigen-space transformations in Eq. \eqref{eq:four_U_evaluation}.}
 \label{fig:accuracy_eigen_policy}
\end{figure}

In addition, we examine the 2D Taylor-Green vortex flow where the evolution of $\mb c_1$ becomes degenerate at certain points. 
Fig. \ref{fig:stress_c_o24_X45_Y37} 
shows the stress evolution at $(45, 37)$. At $t = 1080$, the eigenvalues $\lambda_2$ and $\lambda_3$ cross over. 
Concurrently, the prediction of the polymer stress $\bm\tau _{\textrm{p}}$ from the model without considering the swap of $\mb u_2$ and $\mb u_3$ 
shows apparent deviations near the regime as shown in Fig. \ref{fig:stress_c_o24_X45_Y37}. In contrast, the prediction from the 
model retaining the eigenvalue permutation trained by Eq. \eqref{eq:permutation_G} shows good agreement with the MD results.  

\begin{figure}[htpb]
    \centering
     \includegraphics[width =0.45 \textwidth]{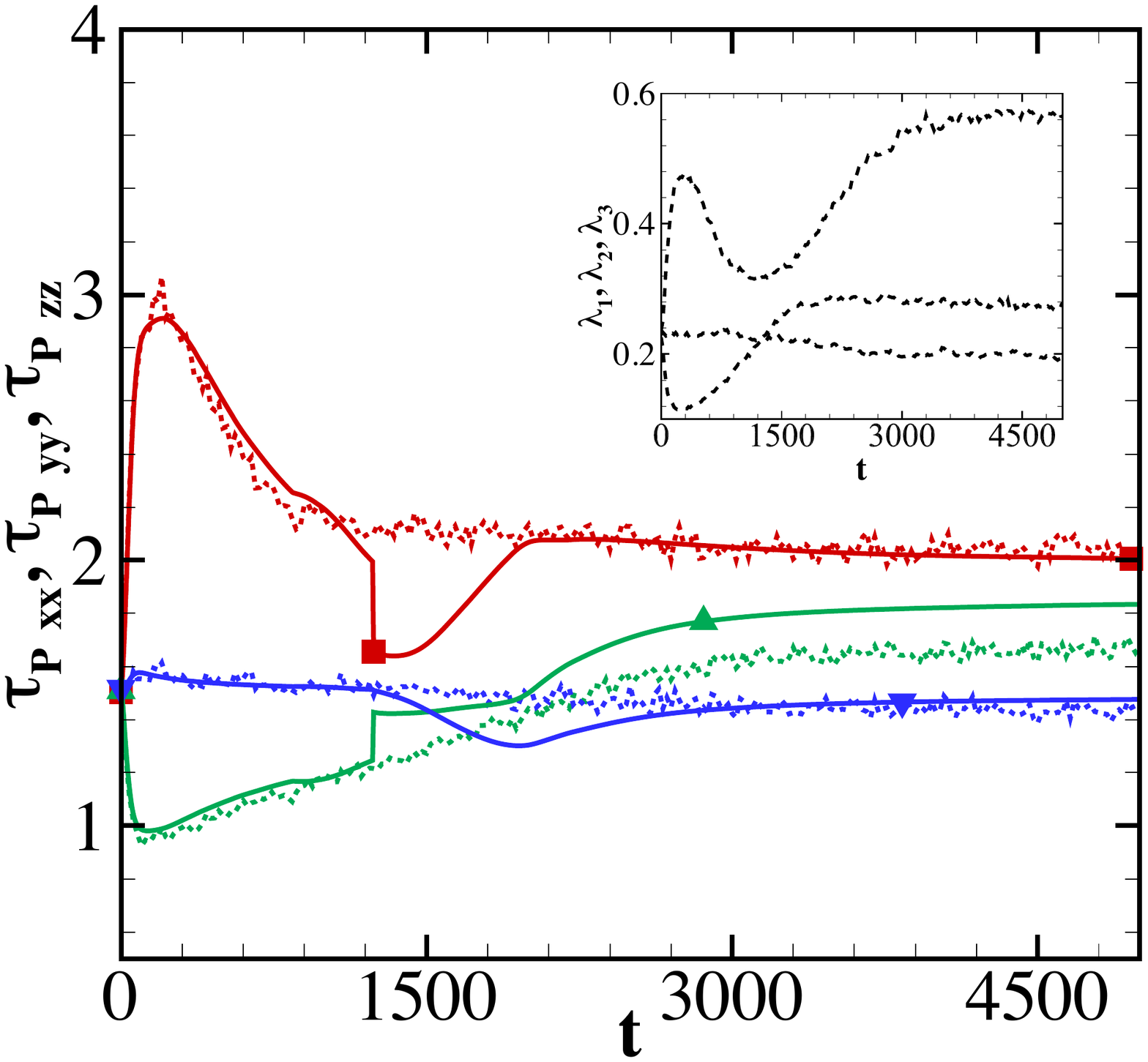}
    \includegraphics[width = 0.45\textwidth]{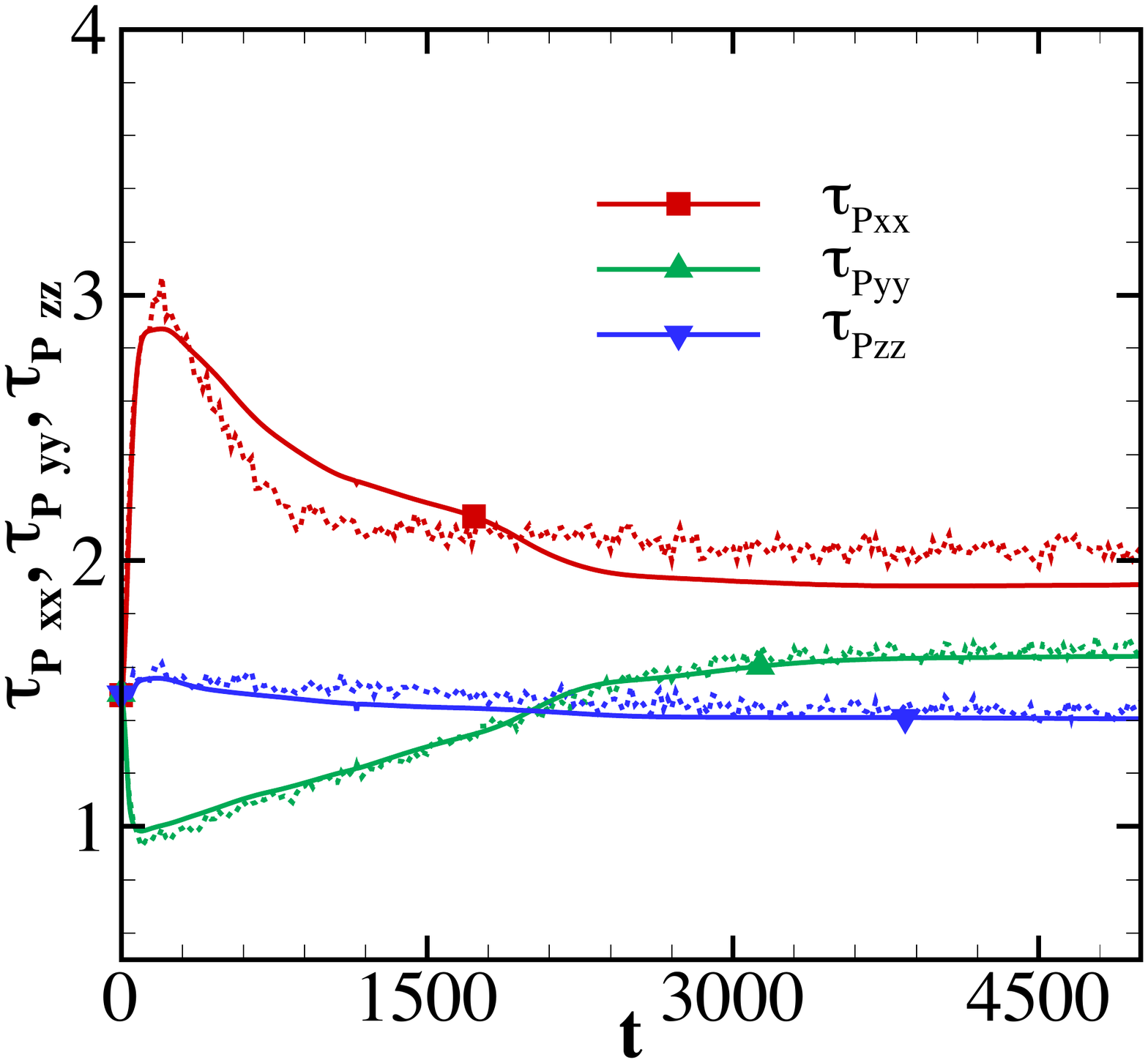}
    \caption{Stress evolution of the Taylor-Green vortex flow at position $(45, 37)$ of the chain-shaped molecule suspension. Left: prediction without considering the swap of eigenvectors when the two eigenvalues approaches near $t = 1255$ as shown in the inset plot. Right: predictions from the model retaining the eigenvalue permutation trained by Eq. \eqref{eq:permutation_G}. The dashed and the solid lines denote the micro-scale simulations and the DeePN$^2$ predictions, respectively.}
    \label{fig:stress_c_o24_X45_Y37}
  \end{figure}


\section{Micro-scale model of the polymer solutions}
\label{sec:micro_scale}
In the present study, we consider suspensions with three different polymer structures as shown in Fig. \ref{fig:molecule_shape}. 
Each polymer molecule consists of $N = 7$ beads connected with $N_b$ FENE bonds, i.e.,
\begin{equation}
V(\mb q) = \sum_{j=1}^{N_b} V_b\left(\vert \mb q_{j_1} - \mb q_{j_2}\vert\right), \quad 
V_b (l) =-\frac{k_s}{2}l^2_{0} \log\left[ 1-\frac{l^2}{l^2_{0}}\right],
\end{equation}
where $k_s$ represents the spring constant and $l_0$ is the maximum of the extension length. The chain- and star-shaped molecules have $N_b = 6$ bonds with the same bond parameters $k_s = 0.1$ and $l_0 = 2.3$ (in reduced unit). 
The net-shaped molecule is similar to the star-shaped molecule with the same parameters for the first $6$ bonds; $3$ additional
bonds connect the side chain particles with $k_s = 0.1$ and $l_0 = 3.7$. The polymer number density of the three 
suspensions is $n_{\textrm p} = 0.3$. The solvent is modeled by the dissipative particle dynamics (DPD) \cite{Hoogerbrugge_SMH_1992, Groot_DPD_1997} with number density $n_s = 4.0$.
The pairwise interaction between particle $i$ and $j$ takes the standard form
\begin{align}
\mb{F}_{ij} &= \mathbf{F}_{ij}^C +  \mathbf{F}_{ij}^D +   \mathbf{F}_{ij}^R, &
\mathbf{F}_{ij}^C &= 
\begin{cases}
a (1.0 - r_{ij}/r_c) \mathbf{e}_{ij}, & r_{ij} < r_c \\
    0, &r_{ij} > r_c
\end{cases}, \nonumber
\\
\mathbf{F}_{ij}^D &= 
\begin{cases}
-\gamma w^{D}(r_{ij})(\mathbf{v}_{ij} \cdot \mathbf{e}_{ij}) \mathbf{e}_{ij}, & r_{ij} < r_c\\
    0, &r_{ij} > r_c
\end{cases}, &
\mathbf{F}_{ij}^R &= 
\begin{cases}
\sigma w^{R}(r_{ij}) \xi_{ij} \mathbf{e}_{ij}, & r_{ij} < r_c\\
    0, &r_{ij} > r_c
\end{cases}, \nonumber 
\end{align}
where $\mathbf{r}_{ij} = \mathbf{r}_{i} - \mathbf{r}_{j}$, $r_{ij} = |\mathbf{r}_{ij}|$, $\mathbf{e}_{ij} = 
\mathbf{r}_{ij}/r_{ij}$, and $\mathbf{v}_{ij} = \mathbf{v}_{i} - \mathbf{v}_{j}$,
$\xi_{ij}$ are independent identically distributed (i.i.d.) Gaussian random variables with zero mean and unit 
variance. $\gamma$ and $\sigma$ are related  with the system temperature by the second fluctuation-dissipation 
theorem \cite{Espanol_SMO_1995} as $\sigma^2 = 2 \gamma k_BT$, where $k_BT$ is set to $0.25$. The detailed 
parameters are given in Tab. \ref{tab:polymer_model_parameter}.
\begin{table}[htbp]
\centering
\caption{Parameters (in reduced unit) of the micro-scale model of the polymer solution (S-solvent, P-polymer). }
\begin{tabular*}
{0.45\textwidth}{c @{\extracolsep{\fill}} cccccc}
\hline\hline
&  $a$ & $\gamma$ & $\sigma$ & $k$ & $r_c$\\
\hline
\text{S-S} & $4.0$ & $5.0$ & $1.58$ &$0.25$ & $1.0$\\ 
\text{S-P} & $0.0$ & $40.0$ & $4.47$ &$0.0$ & $1.0$\\ 
\text{P-P} & $4.0$ & $0.01$ & $0.071$ &$1.0$ & $0.7$\\
\hline
\label{tab:polymer_model_parameter}
\end{tabular*}
\end{table}

\section{Collecting training samples}
\label{sec:training_samples}
Collecting training samples is one of the most important steps in the construction of DeePN$^2$.
To obtain reliable models, we need to ensure that the training sample set is representative enough of all
the practical situations that the model is intended for.
In the present study, we collect the training samples in shear flow with shear rate $\dot{\gamma} \in [0, 0.09]$. 
Since the training of the DeePN$^2$ model only requires discrete polymer configurations rather than time-series samples, one
convenient approach is to consecutively increase the shear rate and collect the discrete configurations during the shear extension
and relaxation process, 
where the inclusion of the relaxation process can facilitate the sampling of polymer configuration
phase space due to the viscoelastic hysteresis effect. $32000$ samples are collected where each sample consists of $5000$ polymer configurations, which will be employed
to evaluate the constitutive dynamics terms $\left\langle \cdot \right\rangle$. Due to the permutation symmetry of the 
the particle label, the effective number of configurations per sample is $1\times 10^4$ for the chain-shaped molecule 
and $3\times 10^4$ for the star- and net-shaped molecules.

\section{Training procedure}
\label{sec:training}
The DeePN$^2$ model is constructed via the training of the NN representations of the encoder mappings $\left\{g_j(\mb r^{\ast})\right\}_{j=1}^{n}$,
stress model $\mb G$,  evolution dynamics $\left\{\mb H_{1,j}\right\}_{j=1}^n$, $\left\{\mb H_{2,j}\right\}_{j=1}^n$  
and the $4$th order tensors $\left\{\mathcal{E}_j\right\}_{j=1}^{n}$ of the objective tensor derivatives. In this study, we choose $n=3$ encoders and fix $g_1(\mb r^{\ast}) \equiv 1$. For the chain-shaped molecule, we set $w_{1,i} = 1 - i / N, 1\le i \le N-1$ and $\sum_i w_{1,i}\mb r_i$ represents the orientation between the free-end particle and the center of mass. For the star- and net-shaped molecules, we set $w_{1,1} = 1$ and $w_{1,i} = 0$ for $i \ge 2$. All terms
are represented by the fully connected NN. The number of hidden layers are set to be { $(120,  120, 120)$, $(300, 300, 300)$, 
$(400, 400, 400)$, $(450, 450, 450)$,  $(560, 560, 560)$}, respectively. The activation function is taken
to be the hyperbolic tangent. We emphasize that the mappings $\left\{g_j(\mb r^{\ast})\right\}_{j=1}^{n}$ and weights 
$w \in \mathbb{R}^{n\times(N-1)}$ involve in the training process for the joint learning of  
the encoders $\left\{\mb b_j(\mb r)\right\}_{j=1}^{n}$ defined in Eq. \eqref{eq:encoder_model} and the macro-scale features $\left\{\mb c_j\right\}_{j=1}^{n}$, although
they do not appear explicitly in the macro-scale hydrodynamic equations.

The DNNs are trained by the Adam stochastic gradient descent method \cite{Kingma_Ba_Adam_2015} for $20$ epochs, using
$5$ samples per batch size. The initial learning rate is $2.8\times 10^{-4}$ and decay rate is $0.75$ per $20000$ steps.

\revision{Similar to Ref. \cite{Lei_Wu_E_2020}, the loss function is defined by
\begin{equation}
L = \lambda_{G} L_{G} + \lambda_{H_1} L_{H_1} + \lambda_{H_2} L_{H_2} + \lambda_{\mathcal{E}} L_{\mathcal{E}},
\nonumber
\end{equation}
where $\lambda_{G} = 0.2$, $\lambda_{H_1} = 0.1$, $\lambda_{H_2} = 0.6$ and $\lambda_{\mathcal{E}} = 0.1$ are hyperparameters. For each
training batch of $m$ training samples,  $L_{G}$, $L_{H_1}$, $L_{H_2}$, $L_{\mathcal{E}}$ of the system are given by }
\begin{equation}
\begin{split}
L_{G} &= \sum_{l=1}^m \sum_{i=1}^{n}\left\Vert {\mb G}_{i}(\mb c^{(l)})
- \left\langle \sum _{k=1}^{N_b} \mb r_k \otimes \nabla_{\mb r_k}  V \right\rangle^{(l)} \right\Vert^2  \\
L_{H_1} &= \sum_{l=1}^m \sum_{i=1}^{n}\left\Vert {\mb H}_{1,i}(\mb c^{(l)})
- \left\langle \sum_{j,k=1}^{N-1}  A_{jk} \nabla_{\mb r_j} \cdot 
\nabla_{\mb r_k} \mb b_i \right\rangle^{(l)} \right\Vert^2  \\
L_{H_2} &= \sum_{l=1}^m \sum_{i=1}^{n}\left\Vert {\mb H}_{2,i}(\mb c^{(l)})
- \left\langle \sum_{j=1}^{N-1} \sum_{k=1}^{N_b}  A_{jk} \nabla_{\mb r_k} V \cdot \nabla_{\mb r_j} 
\mb b_i \right\rangle^{(l)} \right\Vert^2  \\
L_{\mathcal{E}} &= \sum_{l=1}^m \sum_{i=1}^{n} \left\Vert \sum_{s=1}^{9}{\mb E}_{1,i}^{(s)} (\mb c^{(l)})    
\otimes {\mb E}_{2,i}^{(s)} (\mb c^{(l)})
- \left\langle\sum_{k=1}^{N-1}\mb r_k\otimes\nabla_{{\mb r_k}} g_i^2 
\otimes \sum_{j,j'=1}^{N-1}w_{ij}w_{ij'} {\mb r_j} {\mb r_{j'}}^T \right\rangle^{(l)}
\right \Vert^2,
\end{split}
\end{equation}
where $\Vert \cdot \Vert^2$ denotes the total sum of squares of the entries in the tensor, and $\mb c^{(l)} = (\mb c_1^{(l)}, \cdots, \mb c_n^{(l)})$.

\begin{acknowledgments}
We thank Lei Wu and Liyao Lyu for helpful discussions.
The work of Fang, Ge, and Lei is supported in part by the Extreme Science and Engineering Discovery Environment (XSEDE) Bridges at the Pittsburgh Supercomputing Center through allocation MTH210005 and the High Performance Computing Center at Michigan State University. Fang acknowledges the support from Shanghai Jiao Tong University during 08/2020--07/2021. The work of E is supported in part by a gift from iFlytek to Princeton University.
\end{acknowledgments}


\begin{thebibliography}{61}%
\makeatletter
\providecommand \@ifxundefined [1]{%
 \@ifx{#1\undefined}
}%
\providecommand \@ifnum [1]{%
 \ifnum #1\expandafter \@firstoftwo
 \else \expandafter \@secondoftwo
 \fi
}%
\providecommand \@ifx [1]{%
 \ifx #1\expandafter \@firstoftwo
 \else \expandafter \@secondoftwo
 \fi
}%
\providecommand \natexlab [1]{#1}%
\providecommand \enquote  [1]{``#1''}%
\providecommand \bibnamefont  [1]{#1}%
\providecommand \bibfnamefont [1]{#1}%
\providecommand \citenamefont [1]{#1}%
\providecommand \href@noop [0]{\@secondoftwo}%
\providecommand \href [0]{\begingroup \@sanitize@url \@href}%
\providecommand \@href[1]{\@@startlink{#1}\@@href}%
\providecommand \@@href[1]{\endgroup#1\@@endlink}%
\providecommand \@sanitize@url [0]{\catcode `\\12\catcode `\$12\catcode
  `\&12\catcode `\#12\catcode `\^12\catcode `\_12\catcode `\%12\relax}%
\providecommand \@@startlink[1]{}%
\providecommand \@@endlink[0]{}%
\providecommand \url  [0]{\begingroup\@sanitize@url \@url }%
\providecommand \@url [1]{\endgroup\@href {#1}{\urlprefix }}%
\providecommand \urlprefix  [0]{URL }%
\providecommand \Eprint [0]{\href }%
\providecommand \doibase [0]{http://dx.doi.org/}%
\providecommand \selectlanguage [0]{\@gobble}%
\providecommand \bibinfo  [0]{\@secondoftwo}%
\providecommand \bibfield  [0]{\@secondoftwo}%
\providecommand \translation [1]{[#1]}%
\providecommand \BibitemOpen [0]{}%
\providecommand \bibitemStop [0]{}%
\providecommand \bibitemNoStop [0]{.\EOS\space}%
\providecommand \EOS [0]{\spacefactor3000\relax}%
\providecommand \BibitemShut  [1]{\csname bibitem#1\endcsname}%
\let\auto@bib@innerbib\@empty
\bibitem [{\citenamefont {Larson}(1988)}]{Larson88}%
  \BibitemOpen
  \bibfield  {author} {\bibinfo {author} {\bibfnamefont {R.~G.}\ \bibnamefont
  {Larson}},\ }\href@noop {} {\emph {\bibinfo {title} {Constitutive Equations
  for Polymer Melts and Solutions}}}\ (\bibinfo  {publisher}
  {Butterworth-Heinemann Press},\ \bibinfo {year} {1988})\BibitemShut {NoStop}%
\bibitem [{\citenamefont {Owens}\ and\ \citenamefont
  {Phillips}(2002)}]{Owens_Phillips_2002}%
  \BibitemOpen
  \bibfield  {author} {\bibinfo {author} {\bibfnamefont {R.~G.}\ \bibnamefont
  {Owens}}\ and\ \bibinfo {author} {\bibfnamefont {T.~N.}\ \bibnamefont
  {Phillips}},\ }\href@noop {} {\emph {\bibinfo {title} {Computational
  Rheology}}}\ (\bibinfo  {publisher} {Imperial College Press},\ \bibinfo
  {year} {2002})\BibitemShut {NoStop}%
\bibitem [{\citenamefont {Oldroyd}\ and\ \citenamefont
  {Wilson}(1950)}]{Oldroyd_Wilson_PRSLA_1950}%
  \BibitemOpen
  \bibfield  {author} {\bibinfo {author} {\bibfnamefont {J.~G.}\ \bibnamefont
  {Oldroyd}}\ and\ \bibinfo {author} {\bibfnamefont {A.~H.}\ \bibnamefont
  {Wilson}},\ }\href {\doibase 10.1098/rspa.1950.0035} {\bibfield  {journal}
  {\bibinfo  {journal} {Proceedings of the Royal Society of London. Series A.
  Mathematical and Physical Sciences}\ }\textbf {\bibinfo {volume} {200}},\
  \bibinfo {pages} {523} (\bibinfo {year} {1950})}\BibitemShut {NoStop}%
\bibitem [{\citenamefont {Lin}\ \emph {et~al.}(2005)\citenamefont {Lin},
  \citenamefont {Liu},\ and\ \citenamefont
  {Zhang}}]{Energy_Variation_Lin_Liu_Zhang_CPAM_2005}%
  \BibitemOpen
  \bibfield  {author} {\bibinfo {author} {\bibfnamefont {F.-H.}\ \bibnamefont
  {Lin}}, \bibinfo {author} {\bibfnamefont {C.}~\bibnamefont {Liu}}, \ and\
  \bibinfo {author} {\bibfnamefont {P.}~\bibnamefont {Zhang}},\ }\href@noop {}
  {\bibfield  {journal} {\bibinfo  {journal} {Communications on Pure and
  Applied Mathematics}\ }\textbf {\bibinfo {volume} {58}},\ \bibinfo {pages}
  {1437} (\bibinfo {year} {2005})}\BibitemShut {NoStop}%
\bibitem [{\citenamefont {Peterlin}(1966)}]{Peterlin_Polymer_Science_1966}%
  \BibitemOpen
  \bibfield  {author} {\bibinfo {author} {\bibfnamefont {A.}~\bibnamefont
  {Peterlin}},\ }\href@noop {} {\bibfield  {journal} {\bibinfo  {journal}
  {Journal of Polymer Science Part B: Polymer Letters}\ }\textbf {\bibinfo
  {volume} {4}},\ \bibinfo {pages} {287} (\bibinfo {year} {1966})}\BibitemShut
  {NoStop}%
\bibitem [{\citenamefont {Bird}\ \emph {et~al.}(1980)\citenamefont {Bird},
  \citenamefont {Dotson},\ and\ \citenamefont
  {Johnson}}]{Bird_Doston_JNNFM_1980}%
  \BibitemOpen
  \bibfield  {author} {\bibinfo {author} {\bibfnamefont {R.~B.}\ \bibnamefont
  {Bird}}, \bibinfo {author} {\bibfnamefont {P.~J.}\ \bibnamefont {Dotson}}, \
  and\ \bibinfo {author} {\bibfnamefont {N.}~\bibnamefont {Johnson}},\
  }\href@noop {} {\bibfield  {journal} {\bibinfo  {journal} {Journal of
  Non-Newtonian Fluid Mechanics}\ }\textbf {\bibinfo {volume} {7}},\ \bibinfo
  {pages} {213 } (\bibinfo {year} {1980})}\BibitemShut {NoStop}%
\bibitem [{\citenamefont {Giesekus}(1982)}]{Giesekus_JNNFM_1982}%
  \BibitemOpen
  \bibfield  {author} {\bibinfo {author} {\bibfnamefont {H.}~\bibnamefont
  {Giesekus}},\ }\href@noop {} {\bibfield  {journal} {\bibinfo  {journal}
  {Journal of Non-Newtonian Fluid Mechanics}\ }\textbf {\bibinfo {volume}
  {11}},\ \bibinfo {pages} {69 } (\bibinfo {year} {1982})}\BibitemShut
  {NoStop}%
\bibitem [{\citenamefont {Thien}\ and\ \citenamefont
  {Tanner}(1977)}]{PTT_JNNFM_1977}%
  \BibitemOpen
  \bibfield  {author} {\bibinfo {author} {\bibfnamefont {N.~P.}\ \bibnamefont
  {Thien}}\ and\ \bibinfo {author} {\bibfnamefont {R.~I.}\ \bibnamefont
  {Tanner}},\ }\href@noop {} {\bibfield  {journal} {\bibinfo  {journal}
  {Journal of Non-Newtonian Fluid Mechanics}\ }\textbf {\bibinfo {volume}
  {2}},\ \bibinfo {pages} {353} (\bibinfo {year} {1977})}\BibitemShut {NoStop}%
\bibitem [{\citenamefont {Ammar}(2010)}]{ammar2010lattice}%
  \BibitemOpen
  \bibfield  {author} {\bibinfo {author} {\bibfnamefont {A.}~\bibnamefont
  {Ammar}},\ }\href@noop {} {\bibfield  {journal} {\bibinfo  {journal} {Journal
  of non-Newtonian fluid mechanics}\ }\textbf {\bibinfo {volume} {165}},\
  \bibinfo {pages} {1082} (\bibinfo {year} {2010})}\BibitemShut {NoStop}%
\bibitem [{\citenamefont {Fan}(1989)}]{Fan_Acta_1989}%
  \BibitemOpen
  \bibfield  {author} {\bibinfo {author} {\bibfnamefont {X.}~\bibnamefont
  {Fan}},\ }\href@noop {} {\bibfield  {journal} {\bibinfo  {journal} {Acta
  Mechanica Sinica}\ }\textbf {\bibinfo {volume} {1}},\ \bibinfo {pages} {49}
  (\bibinfo {year} {1989})}\BibitemShut {NoStop}%
\bibitem [{\citenamefont {Lozinski}\ and\ \citenamefont
  {Chauviere}(2003)}]{lozinski2003fast}%
  \BibitemOpen
  \bibfield  {author} {\bibinfo {author} {\bibfnamefont {A.}~\bibnamefont
  {Lozinski}}\ and\ \bibinfo {author} {\bibfnamefont {C.}~\bibnamefont
  {Chauviere}},\ }\href@noop {} {\bibfield  {journal} {\bibinfo  {journal}
  {Journal of Computational Physics}\ }\textbf {\bibinfo {volume} {189}},\
  \bibinfo {pages} {607} (\bibinfo {year} {2003})}\BibitemShut {NoStop}%
\bibitem [{\citenamefont {Chauvi{\`e}re}\ and\ \citenamefont
  {Lozinski}(2004)}]{chauviere2004simulation}%
  \BibitemOpen
  \bibfield  {author} {\bibinfo {author} {\bibfnamefont {C.}~\bibnamefont
  {Chauvi{\`e}re}}\ and\ \bibinfo {author} {\bibfnamefont {A.}~\bibnamefont
  {Lozinski}},\ }\href@noop {} {\bibfield  {journal} {\bibinfo  {journal}
  {Computers \& fluids}\ }\textbf {\bibinfo {volume} {33}},\ \bibinfo {pages}
  {687} (\bibinfo {year} {2004})}\BibitemShut {NoStop}%
\bibitem [{\citenamefont {Shen}\ and\ \citenamefont
  {Yu}(2012)}]{shen2012approximation}%
  \BibitemOpen
  \bibfield  {author} {\bibinfo {author} {\bibfnamefont {J.}~\bibnamefont
  {Shen}}\ and\ \bibinfo {author} {\bibfnamefont {H.}~\bibnamefont {Yu}},\
  }\href@noop {} {\bibfield  {journal} {\bibinfo  {journal} {SIAM Journal on
  Numerical Analysis}\ }\textbf {\bibinfo {volume} {50}},\ \bibinfo {pages}
  {1136} (\bibinfo {year} {2012})}\BibitemShut {NoStop}%
\bibitem [{\citenamefont {Carrillo}\ \emph {et~al.}(2019)\citenamefont
  {Carrillo}, \citenamefont {Craig},\ and\ \citenamefont
  {Patacchini}}]{carrillo2019blob}%
  \BibitemOpen
  \bibfield  {author} {\bibinfo {author} {\bibfnamefont {J.~A.}\ \bibnamefont
  {Carrillo}}, \bibinfo {author} {\bibfnamefont {K.}~\bibnamefont {Craig}}, \
  and\ \bibinfo {author} {\bibfnamefont {F.~S.}\ \bibnamefont {Patacchini}},\
  }\href@noop {} {\bibfield  {journal} {\bibinfo  {journal} {Calculus of
  Variations and Partial Differential Equations}\ }\textbf {\bibinfo {volume}
  {58}},\ \bibinfo {pages} {1} (\bibinfo {year} {2019})}\BibitemShut {NoStop}%
\bibitem [{\citenamefont {Degond}\ and\ \citenamefont
  {Mustieles}(1990)}]{degond1990deterministic}%
  \BibitemOpen
  \bibfield  {author} {\bibinfo {author} {\bibfnamefont {P.}~\bibnamefont
  {Degond}}\ and\ \bibinfo {author} {\bibfnamefont {F.-J.}\ \bibnamefont
  {Mustieles}},\ }\href@noop {} {\bibfield  {journal} {\bibinfo  {journal}
  {SIAM Journal on Scientific and Statistical Computing}\ }\textbf {\bibinfo
  {volume} {11}},\ \bibinfo {pages} {293} (\bibinfo {year} {1990})}\BibitemShut
  {NoStop}%
\bibitem [{\citenamefont {Lacombe}\ and\ \citenamefont
  {Mas-Gallic}(1999)}]{lacombe1999presentation}%
  \BibitemOpen
  \bibfield  {author} {\bibinfo {author} {\bibfnamefont {G.}~\bibnamefont
  {Lacombe}}\ and\ \bibinfo {author} {\bibfnamefont {S.}~\bibnamefont
  {Mas-Gallic}},\ }in\ \href@noop {} {\emph {\bibinfo {booktitle} {ESAIM:
  Proceedings}}},\ Vol.~\bibinfo {volume} {7}\ (\bibinfo {organization} {EDP
  Sciences},\ \bibinfo {year} {1999})\ pp.\ \bibinfo {pages}
  {225--233}\BibitemShut {NoStop}%
\bibitem [{\citenamefont {Wang}\ \emph {et~al.}(2021)\citenamefont {Wang},
  \citenamefont {Chen}, \citenamefont {Liu},\ and\ \citenamefont
  {Kang}}]{wang2021particle}%
  \BibitemOpen
  \bibfield  {author} {\bibinfo {author} {\bibfnamefont {Y.}~\bibnamefont
  {Wang}}, \bibinfo {author} {\bibfnamefont {J.}~\bibnamefont {Chen}}, \bibinfo
  {author} {\bibfnamefont {C.}~\bibnamefont {Liu}}, \ and\ \bibinfo {author}
  {\bibfnamefont {L.}~\bibnamefont {Kang}},\ }\href@noop {} {\bibfield
  {journal} {\bibinfo  {journal} {Statistics and Computing}\ }\textbf {\bibinfo
  {volume} {31}},\ \bibinfo {pages} {1} (\bibinfo {year} {2021})}\BibitemShut
  {NoStop}%
\bibitem [{\citenamefont {Bao}\ \emph {et~al.}(2021)\citenamefont {Bao},
  \citenamefont {Liu},\ and\ \citenamefont {Wang}}]{bao2021deterministic}%
  \BibitemOpen
  \bibfield  {author} {\bibinfo {author} {\bibfnamefont {X.}~\bibnamefont
  {Bao}}, \bibinfo {author} {\bibfnamefont {C.}~\bibnamefont {Liu}}, \ and\
  \bibinfo {author} {\bibfnamefont {Y.}~\bibnamefont {Wang}},\ }\href@noop {}
  {\bibfield  {journal} {\bibinfo  {journal} {arXiv preprint arXiv:2112.10970}\
  } (\bibinfo {year} {2021})}\BibitemShut {NoStop}%
\bibitem [{\citenamefont {Laso}\ and\ \citenamefont
  {{\"O}ttinger}(1993)}]{Laso_Ottinger_JNNFM_1993}%
  \BibitemOpen
  \bibfield  {author} {\bibinfo {author} {\bibfnamefont {M.}~\bibnamefont
  {Laso}}\ and\ \bibinfo {author} {\bibfnamefont {H.~C.}\ \bibnamefont
  {{\"O}ttinger}},\ }\href {\doibase
  https://doi.org/10.1016/0377-0257(93)80042-A} {\bibfield  {journal} {\bibinfo
   {journal} {Journal of Non-Newtonian Fluid Mechanics}\ }\textbf {\bibinfo
  {volume} {47}},\ \bibinfo {pages} {1 } (\bibinfo {year} {1993})}\BibitemShut
  {NoStop}%
\bibitem [{\citenamefont {Hulsen}\ \emph {et~al.}(1997)\citenamefont {Hulsen},
  \citenamefont {van Heel},\ and\ \citenamefont {van~den
  Brule}}]{Hulsen_Heel_JNNFM_1997}%
  \BibitemOpen
  \bibfield  {author} {\bibinfo {author} {\bibfnamefont {M.}~\bibnamefont
  {Hulsen}}, \bibinfo {author} {\bibfnamefont {A.}~\bibnamefont {van Heel}}, \
  and\ \bibinfo {author} {\bibfnamefont {B.}~\bibnamefont {van~den Brule}},\
  }\href {\doibase https://doi.org/10.1016/S0377-0257(96)01503-0} {\bibfield
  {journal} {\bibinfo  {journal} {Journal of Non-Newtonian Fluid Mechanics}\
  }\textbf {\bibinfo {volume} {70}},\ \bibinfo {pages} {79 } (\bibinfo {year}
  {1997})}\BibitemShut {NoStop}%
\bibitem [{\citenamefont {Ren}\ and\ \citenamefont
  {E}(2005)}]{REN_E_HMM_complex_fluid_2005}%
  \BibitemOpen
  \bibfield  {author} {\bibinfo {author} {\bibfnamefont {W.}~\bibnamefont
  {Ren}}\ and\ \bibinfo {author} {\bibfnamefont {W.}~\bibnamefont {E}},\
  }\href@noop {} {\bibfield  {journal} {\bibinfo  {journal} {Journal of
  Computational Physics}\ }\textbf {\bibinfo {volume} {204}},\ \bibinfo {pages}
  {1 } (\bibinfo {year} {2005})}\BibitemShut {NoStop}%
\bibitem [{\citenamefont {Warner}(1972{\natexlab{a}})}]{Warner_IECF_1972}%
  \BibitemOpen
  \bibfield  {author} {\bibinfo {author} {\bibfnamefont {H.~R.}\ \bibnamefont
  {Warner}},\ }\href {\doibase 10.1021/i160043a017} {\bibfield  {journal}
  {\bibinfo  {journal} {Industrial \& Engineering Chemistry Fundamentals}\
  }\textbf {\bibinfo {volume} {11}},\ \bibinfo {pages} {379} (\bibinfo {year}
  {1972}{\natexlab{a}})}\BibitemShut {NoStop}%
\bibitem [{\citenamefont {Warner}(1971)}]{Warner_PhD_1971}%
  \BibitemOpen
  \bibfield  {author} {\bibinfo {author} {\bibfnamefont {H.~R.}\ \bibnamefont
  {Warner}},\ }\href@noop {} {Ph.D. thesis},\ \bibinfo  {school} {University of
  Wisconsin} (\bibinfo {year} {1971})\BibitemShut {NoStop}%
\bibitem [{\citenamefont {Zhao}\ \emph {et~al.}(2018)\citenamefont {Zhao},
  \citenamefont {Li}, \citenamefont {Caswell}, \citenamefont {Ouyang},\ and\
  \citenamefont {Karniadakis}}]{Zhao_Li_JCP_2018}%
  \BibitemOpen
  \bibfield  {author} {\bibinfo {author} {\bibfnamefont {L.}~\bibnamefont
  {Zhao}}, \bibinfo {author} {\bibfnamefont {Z.}~\bibnamefont {Li}}, \bibinfo
  {author} {\bibfnamefont {B.}~\bibnamefont {Caswell}}, \bibinfo {author}
  {\bibfnamefont {J.}~\bibnamefont {Ouyang}}, \ and\ \bibinfo {author}
  {\bibfnamefont {G.~E.}\ \bibnamefont {Karniadakis}},\ }\href@noop {}
  {\bibfield  {journal} {\bibinfo  {journal} {Journal of Computational
  Physics}\ }\textbf {\bibinfo {volume} {363}},\ \bibinfo {pages} {116 }
  (\bibinfo {year} {2018})}\BibitemShut {NoStop}%
\bibitem [{\citenamefont {Grosso}\ \emph {et~al.}(2000)\citenamefont {Grosso},
  \citenamefont {Maffettone}, \citenamefont {Halin}, \citenamefont {Keunings},\
  and\ \citenamefont {Legat}}]{Grosso_Maffettone_JNNFM_2000}%
  \BibitemOpen
  \bibfield  {author} {\bibinfo {author} {\bibfnamefont {M.}~\bibnamefont
  {Grosso}}, \bibinfo {author} {\bibfnamefont {P.}~\bibnamefont {Maffettone}},
  \bibinfo {author} {\bibfnamefont {P.}~\bibnamefont {Halin}}, \bibinfo
  {author} {\bibfnamefont {R.}~\bibnamefont {Keunings}}, \ and\ \bibinfo
  {author} {\bibfnamefont {V.}~\bibnamefont {Legat}},\ }\href@noop {}
  {\bibfield  {journal} {\bibinfo  {journal} {Journal of Non-Newtonian Fluid
  Mechanics}\ }\textbf {\bibinfo {volume} {94}},\ \bibinfo {pages} {119}
  (\bibinfo {year} {2000})}\BibitemShut {NoStop}%
\bibitem [{\citenamefont {Feng}\ \emph {et~al.}(1998)\citenamefont {Feng},
  \citenamefont {Chaubal},\ and\ \citenamefont {Leal}}]{Feng_Leal_JoR_1998}%
  \BibitemOpen
  \bibfield  {author} {\bibinfo {author} {\bibfnamefont {J.}~\bibnamefont
  {Feng}}, \bibinfo {author} {\bibfnamefont {C.~V.}\ \bibnamefont {Chaubal}}, \
  and\ \bibinfo {author} {\bibfnamefont {L.~G.}\ \bibnamefont {Leal}},\
  }\href@noop {} {\bibfield  {journal} {\bibinfo  {journal} {Journal of
  Rheology}\ }\textbf {\bibinfo {volume} {42}},\ \bibinfo {pages} {1095}
  (\bibinfo {year} {1998})}\BibitemShut {NoStop}%
\bibitem [{\citenamefont {Wang}(1997)}]{WANG_JNNFM_1997_1}%
  \BibitemOpen
  \bibfield  {author} {\bibinfo {author} {\bibfnamefont {Q.}~\bibnamefont
  {Wang}},\ }\href@noop {} {\bibfield  {journal} {\bibinfo  {journal} {Journal
  of Non-Newtonian Fluid Mechanics}\ }\textbf {\bibinfo {volume} {72}},\
  \bibinfo {pages} {141 } (\bibinfo {year} {1997})}\BibitemShut {NoStop}%
\bibitem [{\citenamefont {Forest}\ \emph {et~al.}(2003)\citenamefont {Forest},
  \citenamefont {Zhou},\ and\ \citenamefont {Wang}}]{Forest_Zhou_Wang_JR_2003}%
  \BibitemOpen
  \bibfield  {author} {\bibinfo {author} {\bibfnamefont {G.~M.}\ \bibnamefont
  {Forest}}, \bibinfo {author} {\bibfnamefont {R.}~\bibnamefont {Zhou}}, \ and\
  \bibinfo {author} {\bibfnamefont {Q.}~\bibnamefont {Wang}},\ }\href {\doibase
  10.1122/1.1530617} {\bibfield  {journal} {\bibinfo  {journal} {Journal of
  Rheology}\ }\textbf {\bibinfo {volume} {47}},\ \bibinfo {pages} {105}
  (\bibinfo {year} {2003})}\BibitemShut {NoStop}%
\bibitem [{\citenamefont {Lielens}\ \emph {et~al.}(1999)\citenamefont
  {Lielens}, \citenamefont {Keunings},\ and\ \citenamefont
  {Legat}}]{FENE_L_S_JNNFM_1999}%
  \BibitemOpen
  \bibfield  {author} {\bibinfo {author} {\bibfnamefont {G.}~\bibnamefont
  {Lielens}}, \bibinfo {author} {\bibfnamefont {R.}~\bibnamefont {Keunings}}, \
  and\ \bibinfo {author} {\bibfnamefont {V.}~\bibnamefont {Legat}},\
  }\href@noop {} {\bibfield  {journal} {\bibinfo  {journal} {Journal of
  Non-Newtonian Fluid Mechanics}\ }\textbf {\bibinfo {volume} {87}},\ \bibinfo
  {pages} {179 } (\bibinfo {year} {1999})}\BibitemShut {NoStop}%
\bibitem [{\citenamefont {Yu}\ \emph {et~al.}(2005)\citenamefont {Yu},
  \citenamefont {Du},\ and\ \citenamefont {Liu}}]{Yu_Du_mms_2005}%
  \BibitemOpen
  \bibfield  {author} {\bibinfo {author} {\bibfnamefont {P.}~\bibnamefont
  {Yu}}, \bibinfo {author} {\bibfnamefont {Q.}~\bibnamefont {Du}}, \ and\
  \bibinfo {author} {\bibfnamefont {C.}~\bibnamefont {Liu}},\ }\href@noop {}
  {\bibfield  {journal} {\bibinfo  {journal} {Multiscale Modeling \&
  Simulation}\ }\textbf {\bibinfo {volume} {3}},\ \bibinfo {pages} {895}
  (\bibinfo {year} {2005})}\BibitemShut {NoStop}%
\bibitem [{\citenamefont {Hyon}\ \emph {et~al.}(2008)\citenamefont {Hyon},
  \citenamefont {Du},\ and\ \citenamefont {Liu}}]{Hyon_Du_mms_2008}%
  \BibitemOpen
  \bibfield  {author} {\bibinfo {author} {\bibfnamefont {Y.}~\bibnamefont
  {Hyon}}, \bibinfo {author} {\bibfnamefont {Q.}~\bibnamefont {Du}}, \ and\
  \bibinfo {author} {\bibfnamefont {C.}~\bibnamefont {Liu}},\ }\href@noop {}
  {\bibfield  {journal} {\bibinfo  {journal} {Multiscale Modeling \&
  Simulation}\ }\textbf {\bibinfo {volume} {7}},\ \bibinfo {pages} {978}
  (\bibinfo {year} {2008})}\BibitemShut {NoStop}%
\bibitem [{\citenamefont {Lei}\ \emph {et~al.}(2020)\citenamefont {Lei},
  \citenamefont {Wu},\ and\ \citenamefont {E}}]{Lei_Wu_E_2020}%
  \BibitemOpen
  \bibfield  {author} {\bibinfo {author} {\bibfnamefont {H.}~\bibnamefont
  {Lei}}, \bibinfo {author} {\bibfnamefont {L.}~\bibnamefont {Wu}}, \ and\
  \bibinfo {author} {\bibfnamefont {W.}~\bibnamefont {E}},\ }\href@noop {}
  {\bibfield  {journal} {\bibinfo  {journal} {Physics Review E}\ }\textbf
  {\bibinfo {volume} {102}},\ \bibinfo {pages} {043309} (\bibinfo {year}
  {2020})}\BibitemShut {NoStop}%
\bibitem [{\citenamefont {Rudy}\ \emph {et~al.}(2017)\citenamefont {Rudy},
  \citenamefont {Brunton}, \citenamefont {Proctor},\ and\ \citenamefont
  {Kutz}}]{Rudy_Kutz_Science_Ad_2017}%
  \BibitemOpen
  \bibfield  {author} {\bibinfo {author} {\bibfnamefont {S.~H.}\ \bibnamefont
  {Rudy}}, \bibinfo {author} {\bibfnamefont {S.~L.}\ \bibnamefont {Brunton}},
  \bibinfo {author} {\bibfnamefont {J.~L.}\ \bibnamefont {Proctor}}, \ and\
  \bibinfo {author} {\bibfnamefont {J.~N.}\ \bibnamefont {Kutz}},\ }\href@noop
  {} {\bibfield  {journal} {\bibinfo  {journal} {Science Advances}\ }\textbf
  {\bibinfo {volume} {3}} (\bibinfo {year} {2017})}\BibitemShut {NoStop}%
\bibitem [{\citenamefont {Schaeffer}\ \emph {et~al.}(2018)\citenamefont
  {Schaeffer}, \citenamefont {Tran},\ and\ \citenamefont
  {Ward}}]{schaeffer2018extracting}%
  \BibitemOpen
  \bibfield  {author} {\bibinfo {author} {\bibfnamefont {H.}~\bibnamefont
  {Schaeffer}}, \bibinfo {author} {\bibfnamefont {G.}~\bibnamefont {Tran}}, \
  and\ \bibinfo {author} {\bibfnamefont {R.}~\bibnamefont {Ward}},\ }\href@noop
  {} {\bibfield  {journal} {\bibinfo  {journal} {SIAM Journal on Applied
  Mathematics}\ }\textbf {\bibinfo {volume} {78}},\ \bibinfo {pages} {3279}
  (\bibinfo {year} {2018})}\BibitemShut {NoStop}%
\bibitem [{\citenamefont {Raissi}\ \emph {et~al.}(2019)\citenamefont {Raissi},
  \citenamefont {Perdikaris},\ and\ \citenamefont
  {Karniadakis}}]{raissi2019physics}%
  \BibitemOpen
  \bibfield  {author} {\bibinfo {author} {\bibfnamefont {M.}~\bibnamefont
  {Raissi}}, \bibinfo {author} {\bibfnamefont {P.}~\bibnamefont {Perdikaris}},
  \ and\ \bibinfo {author} {\bibfnamefont {G.~E.}\ \bibnamefont
  {Karniadakis}},\ }\href@noop {} {\bibfield  {journal} {\bibinfo  {journal}
  {Journal of Computational Physics}\ }\textbf {\bibinfo {volume} {378}},\
  \bibinfo {pages} {686} (\bibinfo {year} {2019})}\BibitemShut {NoStop}%
\bibitem [{\citenamefont {Qin}\ \emph {et~al.}(2019)\citenamefont {Qin},
  \citenamefont {Wu},\ and\ \citenamefont {Xiu}}]{qin2019data}%
  \BibitemOpen
  \bibfield  {author} {\bibinfo {author} {\bibfnamefont {T.}~\bibnamefont
  {Qin}}, \bibinfo {author} {\bibfnamefont {K.}~\bibnamefont {Wu}}, \ and\
  \bibinfo {author} {\bibfnamefont {D.}~\bibnamefont {Xiu}},\ }\href@noop {}
  {\bibfield  {journal} {\bibinfo  {journal} {Journal of Computational
  Physics}\ }\textbf {\bibinfo {volume} {395}},\ \bibinfo {pages} {620}
  (\bibinfo {year} {2019})}\BibitemShut {NoStop}%
\bibitem [{\citenamefont {Han}\ \emph {et~al.}(2019)\citenamefont {Han},
  \citenamefont {Ma}, \citenamefont {Ma},\ and\ \citenamefont
  {E}}]{Han_Ma_PNAS_2019}%
  \BibitemOpen
  \bibfield  {author} {\bibinfo {author} {\bibfnamefont {J.}~\bibnamefont
  {Han}}, \bibinfo {author} {\bibfnamefont {C.}~\bibnamefont {Ma}}, \bibinfo
  {author} {\bibfnamefont {Z.}~\bibnamefont {Ma}}, \ and\ \bibinfo {author}
  {\bibfnamefont {W.}~\bibnamefont {E}},\ }\href {\doibase
  10.1073/pnas.1909854116} {\bibfield  {journal} {\bibinfo  {journal}
  {Proceedings of the National Academy of Sciences}\ }\textbf {\bibinfo
  {volume} {116}},\ \bibinfo {pages} {21983} (\bibinfo {year}
  {2019})}\BibitemShut {NoStop}%
\bibitem [{\citenamefont {Seryo}\ \emph {et~al.}(2020)\citenamefont {Seryo},
  \citenamefont {Sato}, \citenamefont {Molina},\ and\ \citenamefont
  {Taniguchi}}]{Naoki_Takeshi_PRR_2020}%
  \BibitemOpen
  \bibfield  {author} {\bibinfo {author} {\bibfnamefont {N.}~\bibnamefont
  {Seryo}}, \bibinfo {author} {\bibfnamefont {T.}~\bibnamefont {Sato}},
  \bibinfo {author} {\bibfnamefont {J.~J.}\ \bibnamefont {Molina}}, \ and\
  \bibinfo {author} {\bibfnamefont {T.}~\bibnamefont {Taniguchi}},\ }\href@noop
  {} {\bibfield  {journal} {\bibinfo  {journal} {Phys. Rev. Research}\ }\textbf
  {\bibinfo {volume} {2}},\ \bibinfo {pages} {033107} (\bibinfo {year}
  {2020})}\BibitemShut {NoStop}%
\bibitem [{\citenamefont {Yu}\ \emph {et~al.}(2020)\citenamefont {Yu},
  \citenamefont {Tian}, \citenamefont {E},\ and\ \citenamefont
  {Li}}]{Yu_E_Onsagernet_2020}%
  \BibitemOpen
  \bibfield  {author} {\bibinfo {author} {\bibfnamefont {H.}~\bibnamefont
  {Yu}}, \bibinfo {author} {\bibfnamefont {X.}~\bibnamefont {Tian}}, \bibinfo
  {author} {\bibfnamefont {W.}~\bibnamefont {E}}, \ and\ \bibinfo {author}
  {\bibfnamefont {Q.}~\bibnamefont {Li}},\ }\href@noop {} {\  (\bibinfo {year}
  {2020})},\ \Eprint {http://arxiv.org/abs/arXiv:2009.02327} {arXiv:2009.02327}
  \BibitemShut {NoStop}%
\bibitem [{\citenamefont {Huang}\ \emph {et~al.}(2021)\citenamefont {Huang},
  \citenamefont {Ma}, \citenamefont {Zhou},\ and\ \citenamefont
  {Yong}}]{Huang_Yong_JNET_2021}%
  \BibitemOpen
  \bibfield  {author} {\bibinfo {author} {\bibfnamefont {J.}~\bibnamefont
  {Huang}}, \bibinfo {author} {\bibfnamefont {Z.}~\bibnamefont {Ma}}, \bibinfo
  {author} {\bibfnamefont {Y.}~\bibnamefont {Zhou}}, \ and\ \bibinfo {author}
  {\bibfnamefont {W.-A.}\ \bibnamefont {Yong}},\ }\href@noop {} {\bibfield
  {journal} {\bibinfo  {journal} {Journal of Non-Equilibrium Thermodynamics}\
  }\textbf {\bibinfo {volume} {46}},\ \bibinfo {pages} {355} (\bibinfo {year}
  {2021})}\BibitemShut {NoStop}%
\bibitem [{\citenamefont {Warner}(1972{\natexlab{b}})}]{Warner_FENE_1972}%
  \BibitemOpen
  \bibfield  {author} {\bibinfo {author} {\bibfnamefont {H.~R.}\ \bibnamefont
  {Warner}},\ }\href@noop {} {\bibfield  {journal} {\bibinfo  {journal}
  {Industrial \& Engineering Chemistry Fundamentals}\ }\textbf {\bibinfo
  {volume} {11}},\ \bibinfo {pages} {379} (\bibinfo {year}
  {1972}{\natexlab{b}})}\BibitemShut {NoStop}%
\bibitem [{\citenamefont {Zaremba}(1903)}]{Zaremba_1903}%
  \BibitemOpen
  \bibfield  {author} {\bibinfo {author} {\bibfnamefont {S.}~\bibnamefont
  {Zaremba}},\ }\href@noop {} {\bibfield  {journal} {\bibinfo  {journal} {Bull.
  Int. Acad. Sci. Cracovie}\ ,\ \bibinfo {pages} {594}} (\bibinfo {year}
  {1903})}\BibitemShut {NoStop}%
\bibitem [{\citenamefont {Bird}\ \emph {et~al.}(1987)\citenamefont {Bird},
  \citenamefont {Curtiss}, \citenamefont {Armstrong},\ and\ \citenamefont
  {Hassager}}]{Bird_Curtiss_book_vol_2}%
  \BibitemOpen
  \bibfield  {author} {\bibinfo {author} {\bibfnamefont {R.~B.}\ \bibnamefont
  {Bird}}, \bibinfo {author} {\bibfnamefont {C.~F.}\ \bibnamefont {Curtiss}},
  \bibinfo {author} {\bibfnamefont {R.~C.}\ \bibnamefont {Armstrong}}, \ and\
  \bibinfo {author} {\bibfnamefont {O.}~\bibnamefont {Hassager}},\ }\href@noop
  {} {\emph {\bibinfo {title} {Dynamics of Polymeric Liquids, Volume 2: Kinetic
  Theory, 2nd Edition}}},\ \bibinfo {edition} {2nd}\ ed.\ (\bibinfo
  {publisher} {Wiley},\ \bibinfo {year} {1987})\BibitemShut {NoStop}%
\bibitem [{\citenamefont {Rouse}(1953)}]{Rouse_JCP_1953}%
  \BibitemOpen
  \bibfield  {author} {\bibinfo {author} {\bibfnamefont {P.~E.}\ \bibnamefont
  {Rouse}},\ }\href {\doibase 10.1063/1.1699180} {\bibfield  {journal}
  {\bibinfo  {journal} {The Journal of Chemical Physics}\ }\textbf {\bibinfo
  {volume} {21}},\ \bibinfo {pages} {1272} (\bibinfo {year}
  {1953})}\BibitemShut {NoStop}%
\bibitem [{\citenamefont {Zhou}\ \emph {et~al.}(2021)\citenamefont {Zhou},
  \citenamefont {Han},\ and\ \citenamefont {Xiao}}]{Han_Xiao_CMAME_2022}%
  \BibitemOpen
  \bibfield  {author} {\bibinfo {author} {\bibfnamefont {X.-H.}\ \bibnamefont
  {Zhou}}, \bibinfo {author} {\bibfnamefont {J.}~\bibnamefont {Han}}, \ and\
  \bibinfo {author} {\bibfnamefont {H.}~\bibnamefont {Xiao}},\ }\href {\doibase
  10.1016/j.cma.2021.114211} {\  (\bibinfo {year} {2021}),\
  10.1016/j.cma.2021.114211},\ \Eprint {http://arxiv.org/abs/arXiv:2103.06685}
  {arXiv:2103.06685} \BibitemShut {NoStop}%
\bibitem [{\citenamefont {Doyle}\ \emph {et~al.}(1998)\citenamefont {Doyle},
  \citenamefont {Shaqfeh}, \citenamefont {McKinley},\ and\ \citenamefont
  {Spiegelberg}}]{Dolye_Shaqfeh_1998}%
  \BibitemOpen
  \bibfield  {author} {\bibinfo {author} {\bibfnamefont {P.~S.}\ \bibnamefont
  {Doyle}}, \bibinfo {author} {\bibfnamefont {E.~S.}\ \bibnamefont {Shaqfeh}},
  \bibinfo {author} {\bibfnamefont {G.~H.}\ \bibnamefont {McKinley}}, \ and\
  \bibinfo {author} {\bibfnamefont {S.~H.}\ \bibnamefont {Spiegelberg}},\
  }\href@noop {} {\bibfield  {journal} {\bibinfo  {journal} {Journal of
  Non-Newtonian Fluid Mechanics}\ }\textbf {\bibinfo {volume} {76}},\ \bibinfo
  {pages} {79} (\bibinfo {year} {1998})}\BibitemShut {NoStop}%
\bibitem [{\citenamefont {Lielens}\ \emph {et~al.}(1998)\citenamefont
  {Lielens}, \citenamefont {Halin}, \citenamefont {Jaumain}, \citenamefont
  {Keunings},\ and\ \citenamefont {Legat}}]{Lielens_FENE_L_1998}%
  \BibitemOpen
  \bibfield  {author} {\bibinfo {author} {\bibfnamefont {G.}~\bibnamefont
  {Lielens}}, \bibinfo {author} {\bibfnamefont {P.}~\bibnamefont {Halin}},
  \bibinfo {author} {\bibfnamefont {I.}~\bibnamefont {Jaumain}}, \bibinfo
  {author} {\bibfnamefont {R.}~\bibnamefont {Keunings}}, \ and\ \bibinfo
  {author} {\bibfnamefont {V.}~\bibnamefont {Legat}},\ }\href@noop {}
  {\bibfield  {journal} {\bibinfo  {journal} {Journal of Non-Newtonian Fluid
  Mechanics}\ }\textbf {\bibinfo {volume} {76}},\ \bibinfo {pages} {249}
  (\bibinfo {year} {1998})}\BibitemShut {NoStop}%
\bibitem [{\citenamefont {Womersley}(1955)}]{Womersley_flow_1955}%
  \BibitemOpen
  \bibfield  {author} {\bibinfo {author} {\bibfnamefont {J.~R.}\ \bibnamefont
  {Womersley}},\ }\href@noop {} {\bibfield  {journal} {\bibinfo  {journal} {The
  Journal of Physiology}\ }\textbf {\bibinfo {volume} {127}},\ \bibinfo {pages}
  {553} (\bibinfo {year} {1955})}\BibitemShut {NoStop}%
\bibitem [{\citenamefont {Nicholson}\ and\ \citenamefont
  {Rutledge}(2016)}]{Nicholson_UEF_JCP_2016}%
  \BibitemOpen
  \bibfield  {author} {\bibinfo {author} {\bibfnamefont {D.~A.}\ \bibnamefont
  {Nicholson}}\ and\ \bibinfo {author} {\bibfnamefont {G.~C.}\ \bibnamefont
  {Rutledge}},\ }\href {\doibase 10.1063/1.4972894} {\bibfield  {journal}
  {\bibinfo  {journal} {The Journal of Chemical Physics}\ }\textbf {\bibinfo
  {volume} {145}},\ \bibinfo {pages} {244903} (\bibinfo {year}
  {2016})}\BibitemShut {NoStop}%
\bibitem [{\citenamefont {Murashima}\ \emph {et~al.}(2018)\citenamefont
  {Murashima}, \citenamefont {Hagita},\ and\ \citenamefont
  {Kawakatsu}}]{Murashima_Hagita_UEFEX_2018}%
  \BibitemOpen
  \bibfield  {author} {\bibinfo {author} {\bibfnamefont {T.}~\bibnamefont
  {Murashima}}, \bibinfo {author} {\bibfnamefont {K.}~\bibnamefont {Hagita}}, \
  and\ \bibinfo {author} {\bibfnamefont {T.}~\bibnamefont {Kawakatsu}},\
  }\href@noop {} {\bibfield  {journal} {\bibinfo  {journal} {Nihon Reoroji
  Gakkaishi}\ }\textbf {\bibinfo {volume} {46}},\ \bibinfo {pages} {207}
  (\bibinfo {year} {2018})}\BibitemShut {NoStop}%
\bibitem [{\citenamefont {Smith}\ \emph {et~al.}(1999)\citenamefont {Smith},
  \citenamefont {Babcock},\ and\ \citenamefont
  {Chu}}]{Douglas_Hazen_Science_1999}%
  \BibitemOpen
  \bibfield  {author} {\bibinfo {author} {\bibfnamefont {D.~E.}\ \bibnamefont
  {Smith}}, \bibinfo {author} {\bibfnamefont {H.~P.}\ \bibnamefont {Babcock}},
  \ and\ \bibinfo {author} {\bibfnamefont {S.}~\bibnamefont {Chu}},\
  }\href@noop {} {\bibfield  {journal} {\bibinfo  {journal} {Science}\ }\textbf
  {\bibinfo {volume} {283}},\ \bibinfo {pages} {1724} (\bibinfo {year}
  {1999})}\BibitemShut {NoStop}%
\bibitem [{\citenamefont {Taylor}(1934)}]{Taylor_Green_1934}%
  \BibitemOpen
  \bibfield  {author} {\bibinfo {author} {\bibfnamefont {G.~I.}\ \bibnamefont
  {Taylor}},\ }\href@noop {} {\bibfield  {journal} {\bibinfo  {journal}
  {Proceedings of the Royal Society of London. Series A, Containing Papers of a
  Mathematical and Physical Character}\ }\textbf {\bibinfo {volume} {146}},\
  \bibinfo {pages} {501} (\bibinfo {year} {1934})}\BibitemShut {NoStop}%
\bibitem [{\citenamefont {Thomases}\ and\ \citenamefont
  {Shelley}(2007)}]{Thomases_Shelley_POF_2007}%
  \BibitemOpen
  \bibfield  {author} {\bibinfo {author} {\bibfnamefont {B.}~\bibnamefont
  {Thomases}}\ and\ \bibinfo {author} {\bibfnamefont {M.}~\bibnamefont
  {Shelley}},\ }\href@noop {} {\bibfield  {journal} {\bibinfo  {journal}
  {Physics of Fluids}\ }\textbf {\bibinfo {volume} {19}},\ \bibinfo {pages}
  {103103} (\bibinfo {year} {2007})}\BibitemShut {NoStop}%
\bibitem [{\citenamefont {E}\ \emph {et~al.}(2021)\citenamefont {E},
  \citenamefont {Han},\ and\ \citenamefont {Zhang}}]{EHanZhang-PhysicsToday}%
  \BibitemOpen
  \bibfield  {author} {\bibinfo {author} {\bibfnamefont {W.}~\bibnamefont {E}},
  \bibinfo {author} {\bibfnamefont {J.}~\bibnamefont {Han}}, \ and\ \bibinfo
  {author} {\bibfnamefont {L.}~\bibnamefont {Zhang}},\ }\href@noop {}
  {\bibfield  {journal} {\bibinfo  {journal} {Physics Today}\ }\textbf
  {\bibinfo {volume} {74}},\ \bibinfo {pages} {36} (\bibinfo {year}
  {2021})}\BibitemShut {NoStop}%
\bibitem [{\citenamefont {Zimm}(1956)}]{Zimm_model_JCP_1956}%
  \BibitemOpen
  \bibfield  {author} {\bibinfo {author} {\bibfnamefont {B.~H.}\ \bibnamefont
  {Zimm}},\ }\href@noop {} {\bibfield  {journal} {\bibinfo  {journal} {The
  Journal of Chemical Physics}\ }\textbf {\bibinfo {volume} {24}},\ \bibinfo
  {pages} {269} (\bibinfo {year} {1956})}\BibitemShut {NoStop}%
\bibitem [{\citenamefont {Lei}\ \emph {et~al.}(2016)\citenamefont {Lei},
  \citenamefont {Baker},\ and\ \citenamefont {Li}}]{Lei_Li_PNAS_2016}%
  \BibitemOpen
  \bibfield  {author} {\bibinfo {author} {\bibfnamefont {H.}~\bibnamefont
  {Lei}}, \bibinfo {author} {\bibfnamefont {N.~A.}\ \bibnamefont {Baker}}, \
  and\ \bibinfo {author} {\bibfnamefont {X.}~\bibnamefont {Li}},\ }\href@noop
  {} {\bibfield  {journal} {\bibinfo  {journal} {Proc. Natl. Acad. Sci.}\
  }\textbf {\bibinfo {volume} {113}},\ \bibinfo {pages} {14183} (\bibinfo
  {year} {2016})}\BibitemShut {NoStop}%
\bibitem [{\citenamefont {Lei}\ and\ \citenamefont
  {Li}(2021)}]{Lei_Li_JCP_2021}%
  \BibitemOpen
  \bibfield  {author} {\bibinfo {author} {\bibfnamefont {H.}~\bibnamefont
  {Lei}}\ and\ \bibinfo {author} {\bibfnamefont {X.}~\bibnamefont {Li}},\
  }\href@noop {} {\bibfield  {journal} {\bibinfo  {journal} {J. Chem. Phys.}\
  }\textbf {\bibinfo {volume} {154}},\ \bibinfo {pages} {184108} (\bibinfo
  {year} {2021})}\BibitemShut {NoStop}%
\bibitem [{\citenamefont {Hoogerbrugge}\ and\ \citenamefont
  {Koelman}(1992)}]{Hoogerbrugge_SMH_1992}%
  \BibitemOpen
  \bibfield  {author} {\bibinfo {author} {\bibfnamefont {P.~J.}\ \bibnamefont
  {Hoogerbrugge}}\ and\ \bibinfo {author} {\bibfnamefont {J.~M. V.~A.}\
  \bibnamefont {Koelman}},\ }\href@noop {} {\bibfield  {journal} {\bibinfo
  {journal} {Europhys. Lett.}\ }\textbf {\bibinfo {volume} {19}},\ \bibinfo
  {pages} {155} (\bibinfo {year} {1992})}\BibitemShut {NoStop}%
\bibitem [{\citenamefont {Groot}\ and\ \citenamefont
  {Warren}(1997)}]{Groot_DPD_1997}%
  \BibitemOpen
  \bibfield  {author} {\bibinfo {author} {\bibfnamefont {R.~D.}\ \bibnamefont
  {Groot}}\ and\ \bibinfo {author} {\bibfnamefont {P.~B.}\ \bibnamefont
  {Warren}},\ }\href@noop {} {\bibfield  {journal} {\bibinfo  {journal}
  {Journal of Chemical Physics}\ }\textbf {\bibinfo {volume} {107}},\ \bibinfo
  {pages} {4423} (\bibinfo {year} {1997})}\BibitemShut {NoStop}%
\bibitem [{\citenamefont {Espa{\~n}ol}\ and\ \citenamefont
  {Warren}(1995)}]{Espanol_SMO_1995}%
  \BibitemOpen
  \bibfield  {author} {\bibinfo {author} {\bibfnamefont {P.}~\bibnamefont
  {Espa{\~n}ol}}\ and\ \bibinfo {author} {\bibfnamefont {P.~B.}\ \bibnamefont
  {Warren}},\ }\href@noop {} {\bibfield  {journal} {\bibinfo  {journal}
  {Europhysics Letters}\ }\textbf {\bibinfo {volume} {30}},\ \bibinfo {pages}
  {191} (\bibinfo {year} {1995})}\BibitemShut {NoStop}%
\bibitem [{\citenamefont {Kingma}\ and\ \citenamefont
  {Ba}(2015)}]{Kingma_Ba_Adam_2015}%
  \BibitemOpen
  \bibfield  {author} {\bibinfo {author} {\bibfnamefont {D.}~\bibnamefont
  {Kingma}}\ and\ \bibinfo {author} {\bibfnamefont {J.}~\bibnamefont {Ba}},\
  }\href@noop {} {\bibfield  {journal} {\bibinfo  {journal} {International
  Conference on Learning Representations (ICLR)}\ } (\bibinfo {year}
  {2015})}\BibitemShut {NoStop}%
\end{thebibliography}

%

\end{document}